%% file: arxiv_v2.tex
\documentclass[english,11pt]{article}

\usepackage{etex,etoolbox}
\newtoggle{crypto}
\togglefalse{crypto}

\iftoggle{crypto}{
}
{
\usepackage{amsthm}
}

\iftoggle{crypto}{
\usepackage{thmtools,environ,cite,wrapfig}
}{}
\usepackage[english]{babel}
\usepackage{latexsym,amssymb,amsmath,amsfonts,epsfig,color}
\usepackage{graphicx}
\usepackage{float}
\iftoggle{crypto}{}{\usepackage[left=1in,right=1in, top=1in, bottom=1in]{geometry}}
\iftoggle{crypto}{\usepackage{framed}
\usepackage[leftbars,color]{changebar}}{}
\usepackage{ifthen}
\usepackage{authblk}
\usepackage{colonequals}
\iftoggle{crypto}{\usepackage{array}}{}
\iftoggle{crypto}{\usepackage{paralist}}{}
\usepackage{mathrsfs}
\usepackage{complexity}
\usepackage[hidelinks]{hyperref}
\usepackage{enumitem}
\iftoggle{crypto}{\pagestyle{headings}}{} 
\iftoggle{crypto}{\setlist{nolistsep}}{}

\usepackage{tikz}
\usetikzlibrary{decorations.shapes}
\usetikzlibrary{shapes.symbols}

\iftoggle{crypto}{\usepackage[symbol]{footmisc}}{}

\let\originalparagraph\paragraph
\renewcommand{\paragraph}[2][.]{\originalparagraph{#2#1}}

\input{qcircuit}

\definecolor{darkgreen}{rgb}{0,.5,0}

\iftoggle{crypto}{}{
\makeatletter
\newtheorem*{rep@theorem}{\rep@title}
\newcommand{\newreptheorem}[2]{%
\newenvironment{rep#1}[1]{%
 \def\rep@title{#2 \ref*{##1}}%
 \begin{rep@theorem}}%
 {\end{rep@theorem}}}
\makeatother

\newtheorem{theorem}{Theorem}[section]
\newtheorem{lemma}[theorem]{Lemma}
\newtheorem{corollary}[theorem]{Corollary}

\newreptheorem{theorem}{Theorem}
\newreptheorem{lemma}{Lemma}

\newtheorem{definition}[theorem]{Definition}
}

\def\ket#1{{\lvert}#1\rangle}
\def\bra#1{{\langle}#1\rvert}

\newcommand{\ketbra}[2]{\ket{#1}\bra{#2}}			
\newcommand{\egoketbra}[1]{\ketbra{#1}{#1}}			
\DeclareMathOperator{\Tr}{Tr}

\renewcommand{\th}[1]{${#1}^{\textrm{th}}$}
\renewcommand{\(}{\left(}
\renewcommand{\)}{\right)}

\newcommand{\regR}{{\mathcal{R}}}
\newcommand{\regC}{{\mathcal{C}}}
\newcommand{\regA}{{\mathcal{A}}}
\newcommand{\regB}{{\mathcal{B}}}
\newcommand{\regE}{{\mathcal{E}}}

\newcommand{\advA}{{\mathscr{A}}}

\newcommand{\classS}{{\mathscr{S}}}

\newcommand{\hgate}{{\sf H}}
\newcommand{\tgate}{{\sf T}}
\newcommand{\pgate}{{\sf P}}

\newcommand{\xgate}{{\sf X}}

\newcommand{\zgate}{{\sf Z}}
\newcommand{\cnot}{{\sf CNOT}}
\newcommand{\QEnc}{{\mathsf{QEnc}}}
\newcommand{\QDec}{{\mathsf{QDec}}}
\newcommand{\Eval}{{\mathsf{Eval}}}
\newcommand{\KeyGen}{{\mathsf{KeyGen}}}
\newcommand{\CL}{{\sf CL}}
\newcommand{\QHE}{{\sf QHE}}
\newcommand{\HE}{{\sf HE}}
\newcommand{\AUX}{{\sf AUX}}
\newcommand{\TRIV}{{\sf TRIV}}
\newcommand{\EPR}{{\sf EPR}}
\newcommand{\Enc}{{\sf Enc}}
\newcommand{\Dec}{{\sf Dec}}
\newcommand{\sk}{{\it sk}}
\newcommand{\pk}{{\it pk}}
\newcommand{\evk}{{\it evk}}
\newcommand{\id}{{\mathbb I}}
\newcommand{\maxmix}{\$}

\newcommand{\meas}{
\begin{tikzpicture}
\filldraw[fill=white] (0,.25) rectangle (.7,-.25);
\draw (.67,-.1) arc (50:130:.5);
\draw (.35,-.2)--(.525,.2);
\end{tikzpicture}
}

\newcommand{\target}{
\begin{tikzpicture}
\draw (0,-.15)--(0,.15);
\draw (-.15,0)--(.15,0);
\draw (0,0) circle (.15);
\end{tikzpicture}
}

\newcommand{\cntrl}{
\begin{tikzpicture}
\fill (0,0) circle (.08);
\end{tikzpicture}
}

\title{Quantum homomorphic encryption\\
 for circuits of low $\tgate$-gate complexity}

\iftoggle{crypto}{
\author{
Anne Broadbent\footnote{Department of Mathematics and Statistics, University of Ottawa, Ottawa, Ontario, Canada; \texttt{abroadbe@uottawa.ca}.}
\and Stacey Jeffery\footnote{Institute for Quantum Information and Matter, California Institute of Technology, Pasadena, California, USA; \texttt{sjeffery@caltech.edu}.}
}
\institute{}
}
{
\author{Anne Broadbent\footnote{Department of Mathematics and Statistics, University of Ottawa, Ottawa, Ontario, Canada; \texttt{abroadbe@uottawa.ca}.}
\;\;\;\; and \;\;\;\; Stacey Jeffery\footnote{Institute for Quantum Information and Matter, California Institute of Technology, Pasadena, California, USA; \texttt{sjeffery@caltech.edu}.}}

}
\date{}

\begin{document}

\maketitle

\begin{abstract}
Fully homomorphic encryption is an encryption method  with the property that any computation on the plaintext can be performed by a party having  access to the ciphertext only.
Here, we formally define and give schemes for \emph{quantum} homomorphic encryption, which is the encryption of \emph{quantum} information such that \emph{quantum} computations can be performed given the ciphertext only. Our schemes allow for arbitrary Clifford group gates, but become inefficient for circuits with
large complexity, measured in terms of the non-Clifford portion of the circuit (we use the ``$\pi/8$'' non-Clifford group gate, also known as the $\tgate$-gate).

More specifically, two schemes are proposed: the first scheme has a decryption procedure whose complexity scales with the square of  the \emph{number} of $\tgate$-gates (compared with a trivial scheme in which the complexity scales with the total number of gates); the second scheme
uses a quantum evaluation key of length given by a polynomial of degree exponential in the circuit's
$\tgate$-gate depth, yielding a homomorphic scheme for quantum circuits with constant $\tgate$-depth. Both schemes build on a classical fully homomorphic encryption scheme.

A further contribution of ours is to formally define the security of encryption schemes for quantum messages: we define  \emph{quantum indistinguishability under chosen plaintext attacks} in both the public- and private-key settings. In this context, we show the equivalence of several definitions.

Our schemes are the first of their kind that are secure under modern cryptographic definitions, and can be seen as a quantum analogue of classical results establishing homomorphic encryption for circuits with a limited number of \emph{multiplication} gates. Historically, such results appeared as precursors to the breakthrough result establishing classical fully homomorphic encryption.

\end{abstract}

\iftoggle{crypto}{}{
\newpage
\tableofcontents
\newpage
}
\vspace{-20pt}
\section{Introduction}
\vspace{-5pt}

An encryption scheme is \emph{homomorphic} over some set of circuits $\classS$ if any circuit in $\classS$ can be evaluated on an encrypted input. That is, given an encryption of the message~$m$, it is possible to produce a ciphertext that decrypts to the output of the circuit $\sf C$ on input $m$, for any ${\sf C}\in\classS$. In \emph{fully homomorphic encryption (FHE)}, $\classS$ is the set of all classical circuits. FHE was introduced in 1978\iftoggle{crypto}{}{ by Rivest, Adleman and
Dertouzos}~\cite{RAD78}, but the existence of such a scheme was an open problem for over~30 years.
 Some early public-key encryption schemes were homomorphic over the set of circuits consisting of only additions~\cite{GM84,Paillier99} or \iftoggle{crypto}{}{over the set of circuits consisting of }only multiplications~\cite{ElGamal85}.
Several steps were made towards FHE,
with schemes that were homomorphic over increasingly large circuit classes, such as circuits containing additions and a single multiplication~\cite{BGN05}, or of logarithmic depth~\cite{SYY99},
until finally in 2009, Gentry established a breakthrough result by giving the first fully homomorphic encryption scheme~\cite{Gen09}. Follow-up work showed that FHE could be simplified~\cite{MGHV10}, and based
on standard assumptions, such as \emph{learning with errors}~\cite{BV11}.
The advent of FHE has unleashed a series of far-reaching consequences, such as delegating computations\iftoggle{crypto}{}{ in a cloud architecture}, and functional encryption~\cite{Goldwasser:2013:RGC:2488608.2488678}.
For a survey  on \iftoggle{crypto}{FHE}{fully homomorphic encryption}, see~\cite{Vaikuntanathan:2011}.

\iftoggle{crypto}{}{
Quantum cryptography is the study of cryptography in light of quantum information. One branch of quantum cryptography revisits \emph{classical} primitives in the light of quantum information, establishing either no-go results (\emph{e.g.}~\cite{Lo-Chau, Mayers}), or qualitative improvements achieved with quantum information (\emph{e.g.}~\cite{BB84}). Another branch of quantum cryptography seeks to establish \emph{quantum} cryptographic functionality, for instance in multiparty quantum computation~\cite{QMPC} or quantum message authentication~\cite{BCGST02}.
The study of quantum cryptography is notorious for its subtleties and challenges, ranging from dealing with ``purification attacks''~\cite{Lo-Chau, Mayers} to dealing with situations that are unique to the quantum world (such as ``quantum rewinding''~\cite{W06, U12}).}

A number of works have studied \iftoggle{crypto}{}{the cryptographic implications of }the secure delegation of quantum computation\iftoggle{crypto}{ \cite{Chi05,BFK09,ABE10,DFPR13,BGS12,QCED,B15}}{, including:
Childs~\cite{Chi05}; Broadbent, Fitzsimons and Kashefi~\cite{BFK09};  Aharonov, Ben-Or and Eban~\cite{ABE10};  Vedran, Fitzsimons, Portmann and Renner~\cite{DFPR13}; Broadbent, Gutoski and Stebila~\cite{BGS12}; Fisher~\emph{et al.}~\cite{QCED}; and Broadbent~\cite{B15}}.
None \iftoggle{crypto}{}{of these works, however }directly address the question of quantum homomorphic encryption, since they are interactive schemes, and the work of the client is proportional to the size of the circuit being evaluated (and thus, they do not satisfy the \emph{compactness} requirement of \iftoggle{crypto}{FHE}{fully homomorphic encryption}, even if we allow interaction).
Non-interactive approaches are given by\iftoggle{crypto}{}{ Arrighi and Salvail}~\cite{AS06},\iftoggle{crypto}{}{ Rohde, Fitzsimons and  Gilchrist}~\cite{RFJG12} and\iftoggle{crypto}{}{ Tan, Kettlewell, Ouyang, Chen and Fitzsimons}~\cite{TKOCF14}.  However, {none of these approaches are applicable to universal circuit families.}
Furthermore, in the case of~\cite{AS06}, security is given only in terms of cheat sensitivity, while both~\cite{RFJG12} and~\cite{TKOCF14} only bound the leakage of their encoding schemes.

Recent work\iftoggle{crypto}{}{ by Yu, P\'erez-Delgado and  Fitzsimons}~\cite{YPF14} examines the question of perfect security and correctness for quantum fully homomorphic encryption (QFHE), concluding that the trivial scheme is optimal in this context. In light of this result, it is natural to consider computational assumptions in achieving QFHE. Indeed, the question of computationally secure QFHE remains an open problem; our contribution makes progress in this direction by presenting the first schemes that are homomorphic for a large class of quantum circuits.

\vspace{-10pt}
\subsection{Summary of Contributions and Techniques}
\label{sec:intro-summary}
\vspace{-5pt}
We introduce schemes for \emph{quantum homomorphic encryption (QHE)}, the quantum version of homomorphic encryption; we are \iftoggle{crypto}{}{thus }interested in \iftoggle{crypto}{}{establishing functionality for }the evaluation of \emph{quantum} circuits on encrypted \emph{quantum} data.
In terms of definitions, we contribute by giving the first definition of quantum homomorphic encryption (QHE) in the computational setting, in the case of both public-key and symmetric-key cryptosystems. As a consequence, we give the first formal definition (and scheme) for the public-key encryption of quantum information, where security is given in terms of \emph{quantum indistinguishability under chosen plaintext attacks}---for which we show the equivalence of a number of definitions, including security for multiple messages. {Prior work considered the computational setting for quantum encryption of classical plaintexts only~\cite{OTU00,K07,XY12}.}

In terms of QHE schemes, we start by using straightforward techniques to construct a scheme that is homomorphic for Clifford circuits\iftoggle{crypto}{}{ (or, more generally, stabilizer circuits)}.
This can  be seen as an analogue to a classical scheme that is homomorphic for linear circuits (circuits performing only additions). While Clifford circuits are not universal for quantum computation, this already yields a range of applications for quantum information processing, including encoding and decoding into stabilizer codes. Our quantum public-key encryption scheme is a hybrid of a classical public-key fully homomorphic encryption scheme and the quantum one-time pad~\cite{AMTW00}. Intuitively, the scheme works by encrypting the quantum register with a quantum one-time pad, and then encrypting the one-time pad encryption keys with a classical public-key FHE scheme. Since Clifford circuits conjugate Pauli operators to Pauli operators, any Clifford circuit can be directly applied to the encrypted quantum register; the homomorphic property of the classical encryption scheme is used to update the encryption key.
Of course, we specify that the classical FHE scheme should be secure against quantum adversaries. By using,  \emph{e.g.},~the scheme from~\cite{BV11}, we get security based on the \emph{learning with errors} (LWE) assumption \cite{Regev2005, Regev2009}; this has been equated with worst-case hardness of ``short vector problems" on arbitrary lattices \cite{MG09}, which is widely believed  to be a quantum-safe (or ``post-quantum'') assumption.

For universal quantum computations, we must evaluate a non-Clifford gate, for which we choose the ``$\tgate$'' gate (also known as ``$\mathsf{R}$'' or ``$\pi/8$''). Applying the above principle  we run into trouble, since $\tgate \xgate^a \zgate^b = \xgate^a \zgate^{a \oplus b} \pgate^a \tgate$. That is, conditioned on the quantum one-time pad encryption key $a, b \in \{0,1\}$, the output picks up an undesirable non-Pauli error.
 Our main contribution is to present two schemes, $\EPR$~and~$\AUX$, that deal with this situation in two different ways:

 \begin{description}
 \item[$\EPR$:]  The main idea of $\EPR$ \iftoggle{crypto}{}{(named after the famous Einstein-Podolski-Rosen trio~\cite{EPR35})} is to use
 entangled quantum registers
       to enable corrections  \emph{within the circuit} at the time of decryption.
      This scheme is efficient for any quantum circuit, however, it fails to meet a requirement for fully homomorphic encryption called \emph{compactness}, which requires that the complexity of the decryption procedure be independent of the evaluated circuit.
       More specifically, the complexity of the decryption procedure for $\EPR$ scales with the square of the number of $\tgate$-gates. This gives an advantage over the trivial scheme whenever the number of $\tgate$-gates in the evaluated circuit is less than the squareroot of the number of gates. (The \emph{trivial} scheme consists of
   appending to the ciphertext a description of the circuit to be evaluated, and specifying that it should be applied as part of the decryption procedure.)


 \item[$\AUX$:]
 Compared to $\EPR$, the scheme $\AUX$ takes a more proactive approach to performing the correction required for a $\tgate$-gate: to do this, it uses a number of auxiliary qubits  that are given as part of the evaluation key.  Intuitively, these auxiliary qubits encode the required corrections. In order to ensure universality, a large number of possible corrections must be available --- the length of the evaluation key is thus given by a polynomial of degree exponential in the circuit's $\tgate$-gate \emph{depth}, yielding  a homomorphic scheme that is efficient for quantum circuits with constant $\tgate$-depth.
 \end{description}
\newpage
The two main schemes\iftoggle{crypto}{}{ $\EPR$ and $\AUX$} are incomparable\iftoggle{crypto}{}{; for some circuits, $\EPR$ is more desirable, while for others, it is preferable to use $\AUX$}. The scheme $\EPR$ becomes less \emph{compact} (and therefore less interesting, since it approaches the trivial scheme), as the \emph{number} of $\tgate$-gates increases, while the scheme~$\AUX$ becomes inefficient (\emph{extremely} rapidly) as the \emph{depth} of $\tgate$-gates increases.

Our results can be viewed as a quantum analogue of precursory results to classical fully homomorphic encryption, which established the homomorphic property of encryption schemes that tolerate a limited amount of operations.
One difference is that, while these schemes started with the modest goal of just a \emph{single} multiplication (the addition operation being ``easy''), we have already allowed for at the very least a \emph{constant} number, and, depending on the circuit, up to a polynomial number of ``hard'' operations, namely of $\tgate$-gates.

Our schemes use the existence of classical FHE, although at the expense of a slightly more complicated exposition, a classical scheme that is homomorphic only for linear circuits would actually suffice. We see the relationship between our schemes and classical FHE as a strength of our result, via the following interpretation: classical FHE is sufficient to enable QHE for a large family of circuits, and perhaps by taking greater advantage of the \emph{fully} homomorphic property of the classical scheme in some as yet unknown way, our ideas might be extended to larger classes of quantum circuits. With this in mind, and for ease of exposition, we use a classical fully homomorphic encryption scheme for all of our quantum homomorphic encryption schemes.

\iftoggle{crypto}{}{
An additional contribution of ours is conceptual: in the context of quantum circuits, it had been known for some time now that the non-Clifford part of a quantum computation is the ``difficult'' one (this phenomena appears, \emph{e.g.}~in the context of quantum simulations~\cite{Got98}, fault-tolerant quantum computation~\cite{BK05} and quantum secure function evaluation~\cite{DNS10, DNS12, QMPC}).
This has motivated a series of theoretical work seeking to optimize quantum circuits in terms of their $\tgate$-gate complexity~\cite{Sel13, KMM13}. In particular, Amy,  Maslov, Mosca, and Roetteler~\cite{AMMR13}  recently proposed $\tgate$-depth as a cost function, the idea being to count the number of $\tgate$-layers in a quantum circuit and optimize over this parameter. Our contribution adds to this understanding, showing that, in the context of quantum homomorphic encryption, the main challenge is to evaluate non-Clifford gates, the bottleneck being, more precisely, the \emph{depth} of the $\tgate$-gate part of the circuit.
}

\iftoggle{crypto}{}{
\paragraph{Organization}}
Some preliminaries and notation are given in Sec.~\ref{sec:Prelim}. We give formal definitions of quantum homomorphic encryption and related concepts, including security definitions, in Sec.~\ref{sec:QFHE-Defs}; this allows us to formally state our results in Sec.~\ref{sec:main-contributions}. Sec.~\ref{sec:scheme-Clifford} contains a basic quantum homomorphic encryption scheme,~$\CL$, for Clifford circuits that is used as a basis for $\EPR$\iftoggle{crypto}{}{, the entanglement-based quantum homomorphic encryption scheme} (Sec.~\ref{sec:scheme-EPR}), and
\iftoggle{crypto}{}{for }$\AUX$\iftoggle{crypto}{ }{, the auxiliary-qubit based quantum homomorphic encryption scheme }(Sec.~\ref{sec:scheme-AUX}).
\iftoggle{crypto}{Further details, including proofs of our main theorems, can be found in the full version \cite{BJ15}.}{}

\vspace{-5pt}
\section{Preliminaries and Notation}
\label{sec:Prelim}
\vspace{-5pt}
\iftoggle{crypto}{}{
\subsection{Notation}
}
A negligible function, 
$\eta(\cdot)$, is a function such that for every polynomial $p(\cdot)$, there exists an~$N$ such that for all integers $n > N$ it holds that $\eta(n) < \frac{1}{p(n)}$.
As a convention, if $a$ is a classical plaintext, we denote its encryption by $\tilde{a}$.  Throughout this work we use $\kappa$ to indicate the security parameter.\looseness=-1

\iftoggle{crypto}{}{
For a detailed and rigorous introduction to quantum information theory, we refer the reader to \cite{WatrousNotes}. In the remainder of this section, we give a brief overview of some of the necessary concepts, as well as our specific notation.}

A \emph{quantum register} is a quantum system, which we view as a physical object that stores quantum information. The contents of a quantum register are mathematically modelled as the set of trace-1, positive semidefinite operators, called \emph{density operators}, on $\mathcal{X}$, where $\mathcal{X}$
is a  complex Euclidean space. We denote the set of density operators on any space $\mathcal{X}$ by $D(\mathcal{X})$.

Quantum registers are denoted with calligraphic typeset\iftoggle{crypto}{}{, such as $\cal X$, $\cal Y$}.
Two \iftoggle{crypto}{}{(or more) }quantum systems, $\mathcal{X}$ and $\mathcal{Y}$, form a composite system by the tensor product\iftoggle{crypto}{}{ of the subsystems}, $\mathcal{X}\otimes \mathcal{Y}$. If $\rho\in D(\mathcal{X}\otimes \mathcal{Y})$ is a state on the joint system, we write $\rho^{\mathcal{X}}$ to denote $\Tr_{\mathcal{Y}}(\rho)$.
If $\mathcal{X}$ and $\mathcal{Y}$ have the same dimension, we denote this by $\mathcal{X} \equiv \mathcal{Y}$.
\iftoggle{crypto}{}{

}
The \emph{trace distance} between two states, $\rho$ and $\sigma$, is defined $\Delta(\rho, \sigma):=\Tr\(\sqrt{(\rho-\sigma)^\dagger(\rho-\sigma)}\)$.

A density matrix that is diagonal in the computational basis corresponds to a classical random variable. For a random variable $X$ on some set $\Sigma_X$, we define $\rho(X):=\sum_{x\in\Sigma_X}\Pr[X=x]\ket{x}\bra{x}$, the density matrix corresponding to $X$.
 A \emph{classical-quantum} state is a state of the form $\rho^{\mathcal{M}\mathcal{A}} = \sum_x \Pr[X=x]\egoketbra{x}^\mathcal{M} \otimes \rho_x^\mathcal{A}$.

One special quantum state on any system $\cal X$ is the \emph{completely mixed state}, $\frac{1}{\dim \cal X}\mathbb{I}_{\cal X}$, which we will sometimes denote by $\maxmix$ (where $\cal X$ should be implicit from the context). When $\cal X$ is interpreted as $\mathbb{C}^S$ for some finite set $S$, then $\maxmix$ corresponds to the uniform distribution on $S$.

A \emph{quantum channel} $\Phi:D(\regA)\to D(\regB)$ refers to any physically-realizable mapping on quantum registers.
 The identity channel on register~$\regR$ is denoted~$\mathbb{I}_\regR$.
 Let $\Phi$ be a quantum channel acting on register $\regA$, and $\rho^{\regA\regE}$ a quantum system held in the joint registers $\regA \otimes \regE$. Then to simplify notation, when it is clear from the context, we  write $\Phi(\rho^{\regA\regE})$ to mean  $(\Phi \otimes \id)(\rho^{\regA\regE})$.

\iftoggle{crypto}{}{
We mention a special type of channel, a \emph{conditional quantum channel}, which, on input the classical-quantum state $\sum_x \Pr[x]\egoketbra{x}^\mathcal{M} \otimes \rho_x^\mathcal{A}$,  outputs the quantum state:
\iftoggle{crypto}{
$\Tr_M \left(\sum_x \Pr[x]\egoketbra{x}^\mathcal{M} \otimes \Phi_x (\rho_x^\mathcal{A})\right)$ for quantum channels $\Phi_x$.
}{
\begin{equation*}
\Tr_M \left(\sum_x \Pr[x]\egoketbra{x}^\mathcal{M} \otimes \Phi_x (\rho_x^\mathcal{A})\right)
\end{equation*}
for quantum channels $\Phi_x: D(\regA)\to D(\regB)$.}}

\iftoggle{crypto}{}{Unless otherwise specified, a quantum \emph{measurement} refers to a measurement in the computational basis.
A quantum \emph{algorithm} is a polynomial-time uniform family of quantum circuits, implementing a family of quantum channels.}

\iftoggle{crypto}{
}{
\subsection{Quantum Circuits}
\label{sec:prelim-circuits}
}

We work with the \iftoggle{crypto}{gate set $\{\xgate,\zgate,\pgate,\cnot,\hgate\}$.}{set of quantum gates consisting of single-qubit preparation in the~$\ket{0}$ state, single-qubit measurements, as well as the gates the following unitary gates:
$$\xgate = \left[\begin{array}{cc} 0 & 1\\ 1 & 0\end{array}\right],
\quad\zgate = \left[\begin{array}{cc} 1 & 0\\ 0 & -1\end{array}\right],
\quad\pgate = \left[\begin{array}{cc} 1 & 0\\ 0 & i\end{array}\right],
\quad\tgate = \left[\begin{array}{cc} 1 & 0\\ 0 & e^{i\pi/4}\end{array}\right],$$
$$\hgate = \frac{1}{\sqrt{2}}\left[\begin{array}{cc}1 & 1\\1 & -1\end{array}\right],\quad\mbox{and}
\quad\cnot = \left[\begin{array}{cccc} 1 & 0 & 0 & 0\\ 0 & 1 & 0 & 0\\ 0 & 0 & 0 & 1\\ 0 & 0 & 1 & 0\end{array}\right].$$}
\iftoggle{crypto}{This gate set}{The set $\{\xgate,\zgate,\pgate,\cnot,\hgate\}$} applied to arbitrary wires (redundantly) generates the Clifford group, and adding any non-Clifford gate, such as $\tgate$, gives a generating set for all quantum circuits.
\iftoggle{crypto}{}{We note the following relations between these gates:
$$\xgate\zgate = - \zgate\xgate,\quad \tgate^2=\pgate,\quad\pgate^2=\zgate,\quad\hgate\xgate\hgate=\zgate,\quad \tgate\pgate=\pgate\tgate,\quad\pgate\zgate=\zgate\pgate.$$}

\iftoggle{crypto}{}{A classical circuit is \emph{layered} if it consists of alternating layers of either all `$+$' gates or all `$\times$' gates. The \emph{multiplicative depth} of a layered circuit is the number of `$\times$' layers. As we see in this work, a natural quantum analogue of  `$+$' gates are Clifford group gates, while the analogue of the~`$\times$' gate is the~$\tgate$-gate.\footnote{The analogy is due to the ``easiness'' of performing Clifford group computations on encrypted data, versus the ``hardness'' of performing non-Clifford group computations. Another way of seeing this is that the (reversible) quantum analogue of multiplication is the Toffoli gate: $\ket{x,y,z} \mapsto \ket{x,y,x\cdot y \oplus z}$. The Toffoli is a non-Clifford group gate that can be expressed in terms of $\tgate$-gates~\cite{Sel13}.} Thus, a \emph{layered quantum circuit} consists of alternating layers of either all Clifford group gates, or $\tgate$-gates. Then the $\tgate$-depth of a layered quantum circuit is the number of  such $\tgate$ layers~\cite{AMMR13}.}

\iftoggle{crypto}{
}{
\subsection{Quantum One-time Pad}
\label{sec:prelim-QOTP}
}
For a single-qubit \iftoggle{crypto}{}{system $\rho$ in }register $\regR$, and $a,b \in \{0,1\}$, we denote by $\QEnc_{a,b}: \regR \rightarrow \regR $ the quantum one-time pad encryption and by $\QDec_{a,b}:  \regR \rightarrow \regR$ the quantum one-time pad decryption~\cite{AMTW00}, \iftoggle{crypto}{$\mathsf{QEnc}_{a,b} : \rho \mapsto \xgate^a \zgate^b \rho \zgate^b \xgate^a$ and $\mathsf{QDec}_{a,b} = \mathsf{QEnc}_{a,b}$.}{namely:
\begin{equation}
\mathsf{QEnc}_{a,b} : \rho \mapsto \xgate^a \zgate^b \rho \zgate^b \xgate^a\qquad\mbox{and}\qquad
\mathsf{QDec}_{a,b} : \rho \mapsto \xgate^a \zgate^b \rho \zgate^b \xgate^a.
\end{equation}

}
It is easy to see that  $\QDec_{a,b} \circ \QEnc_{a,b} = \id_{\regR}$. By specifying that~$(a,b)$ be chosen  uniformly at random, we get that the encryption maps any input to the completely mixed state (from the point of view of the adversary), since for all~$\rho$,
\iftoggle{crypto}{$\frac{1}{4}\sum_{a,b} \xgate^a \zgate^b \rho \zgate^b \xgate^a = \frac{\id_2}{2}$.}{
\begin{equation}
\frac{1}{4}\sum_{a,b} \xgate^a \zgate^b \rho \zgate^b \xgate^a = \frac{\id_2}{2}\,.
\end{equation}
}

\vspace{-5pt}
\section{Definitions}
\label{sec:QFHE-Defs}
\vspace{-5pt}

\iftoggle{crypto}{We now}{In this section, we} formally define QHE schemes and their properties. In Sec.~\ref{sec:QFHE}, we \iftoggle{crypto}{}{first review classical FHE, and then }define~QHE in the public-key setting. Sec.~\ref{sec:security-QHE} carefully defines the security of QHE, \iftoggle{crypto}{giving}{by considering} two definitions for security under chosen plaintext attacks, \iftoggle{crypto}{shown in the full version \cite{BJ15} to be equivalent}{and showing that they are equivalent}.  Sec.~\ref{sec:QHE-correctness} defines correctness and compactness for QHE, culminating in a complete definition of quantum fully homomorphic encryption\iftoggle{crypto}{}{ (Def.~\ref{defn:QFHE})}.  Sec.~\ref{sec:indivisible} deals with an important subtlety that arises in the quantum case: due to the no-cloning theorem, when a large system is encrypted with some auxiliary quantum information needed for decryption, that auxiliary information cannot be copied and given to every subsystem, but rather, the system must now be decrypted as a whole, rather than subsystem-by-subsystem.  We also define compactness and quasi-compactness in this context. Finally, \iftoggle{crypto}{}{for technical reasons, }one of our schemes ($\AUX$) must be used in the symmetric-key setting, \iftoggle{crypto}{defined}{which we define} in Sec.~\ref{sec:symmetric-key}.
{We do not address the issue of \emph{circuit privacy}~\cite{GHV10}, leaving this question for future work}.

\vspace{-10pt}
\subsection{Classical and Quantum  Homomorphic Encryption}
\label{sec:QFHE}
\vspace{-5pt}

Our schemes rely on a classical fully homomorphic encryption scheme\iftoggle{crypto}{}{; for completeness, we include a definition in App.~\ref{appendix:Classical-FHE}}.
Since our adversaries are modelled as being \emph{quantum} polynomial-time, we need a further security guarantee on the classical scheme, namely that it is secure against \emph{quantum} adversaries (see Def.~\ref{defn:classical-cipher-q-IND-CPA}). Fortunately, much of classical fully homomorphic encryption uses lattice-based cryptography, which exploits one of the few conjectured ``quantum-safe'' assumptions~\cite{MG09}. Among all known solutions, the scheme of~\cite{BV11} appears to be the best for our purposes, as it bases its security on the \emph{learning with errors} (LWE) assumption \cite{Regev2005, Regev2009}, which has been equated with worst-case hardness of ``short vector problems" on arbitrary lattices.

\iftoggle{crypto}{\vspace{-5pt}}{}
\begin{definition} [q-IND-CPA]
\label{defn:classical-cipher-q-IND-CPA}
A classical homomorphic encryption scheme $\mathsf{HE}$ is q-IND-CPA secure if  for any \emph{quantum} polynomial-time adversary $\advA$, there exists a negligible function $\eta$ such that for $(pk, evk, sk) \leftarrow \mathsf{HE.Keygen(1^\kappa)}$:
\vspace{-5pt}
\begin{equation*}
\left|\mathrm{Pr}[\advA(pk, evk, \mathsf{HE.Enc}_{pk}(0)) = 1] - \mathrm{Pr}[\advA(pk, evk, \mathsf{HE.Enc}_{pk}(1)) = 1] \right|\leq \eta(\kappa)\,.  \end{equation*}
\end{definition}
\noindent\iftoggle{crypto}{Although a classical scheme that is q-IND-CPA is also IND-CPA, the converse may not be true.}{ We note that a number of recent works examine the security of classical schemes against quantum superposition attacks~\cite{Zandry12, BDFS11, BZ13b}.  In this context, our definition of q-IND-CPA above models security for classical plaintexts only (with an arbitrary learning phase, since the public key is given). Furthermore, we note that a classical homomorphic encryption scheme that is q-IND-CPA, is also IND-CPA. The converse, however, may not be true (in particular, if the IND-CPA property depends on a computational assumption that is broken by quantum computers).} Note, however,  that any proof that a scheme is IND-CPA can potentially be turned into a proof for q-IND-CPA if all statements still hold when ``probabilistic polynomial-time adversary'' is replaced by ``quantum polynomial-time adversary'' (see~\cite{Song14}).

We now give our new definitions for quantum homomorphic encryption. In our definitions, both $pk$, the public encryption key, and $sk$, the secret decryption key, are classical, whereas the evaluation key is allowed to be a quantum state. \iftoggle{crypto}{}{This choice is simply based on what is needed by our schemes.}

\iftoggle{crypto}{\vspace{-5pt}}{}
\begin{definition}[QHE]A quantum homomorphic encryption scheme is a $4$-tuple of quantum algorithms $(\sf{QHE.KeyGen},\sf{QHE.Enc},\sf{QHE.Eval},\sf{QHE.Dec})$:
\begin{description}
\item[{Key Generation.}]
 $\QHE.\KeyGen: 1^\kappa \rightarrow (pk, sk, \rho_{evk})$.
 This algorithm takes a unary representation of the security parameter as input and outputs a classical public encryption key~$\pk$,  a classical secret decryption key~$\sk$ and a quantum evaluation key $\rho_{evk}\in D(\regR_{evk})$.

\item[{Encryption.}]
$\QHE.\Enc_{pk}: D(\mathcal{M}) \rightarrow D(\regC)$. For every possible\iftoggle{crypto}{}{ value of} $pk$, the quantum channel $\mathsf{Enc}_{pk}$  maps a state in the message space $\cal M$ to a state (the \emph{cipherstate}) in the cipherspace ${\cal C}$.

 \item[{Homomorphic Evaluation.}] $\QHE.\Eval^{\mathsf{C}}: D(\regR_{evk}\otimes  \mathcal{C}^{\otimes n}) \rightarrow D(\mathcal{C'}^{\otimes m})$.
For every quantum circuit $\mathsf{C}$, with induced channel $\Phi_{\mathsf{C}}:D(\mathcal{M}^{\otimes n})\rightarrow D(\mathcal{M}^{\otimes m})$, we define a channel $\mathsf{Eval}^{\mathsf{C}}$ that maps an $n$-fold cipherstate to an $m$-fold cipherstate, consuming the evaluation key in the process.\iftoggle{crypto}{}{\footnote{Since we have not specified any requirement on the behaviour of this channel, we can define $\mathsf{Eval}^{\mathsf{C}}$ to have some trivial behaviour on some, or even all quantum circuits $\mathsf{C}$. However, for the scheme to have the $\mathscr{S}$-homomorphic property (Def.~\ref{defn:C-homomorphism}), this cannot be the case for any circuit in $\mathscr{S}$.}}

\item[{Decryption.}] $\mathsf{QHE.Dec}_{sk} : D(\regC') \rightarrow D(\mathcal{M})$. For every possible \iftoggle{crypto}{\!\!}{value of} $sk$,  $\mathsf{Dec}_{sk}$ is a quantum channel that maps the state in $D(\regC')$ to a quantum state in $D(\mathcal{M})$.

 \end{description}
\end{definition}

\vspace{-12pt}
\subsection{Security of Quantum Homomorphic Encryption}
\label{sec:security-QHE}
\vspace{-5pt}

We now define a notion of security for QHE analogous to the classical notion of indistinguishability under chosen plaintext attack. \iftoggle{crypto}{}{As in the classical case, there are several possible definitions, ranging from a relatively simple experiment (Def.~\ref{def:q-IND-CPA1}) to multiple messages (Def.~\ref{def:q-IND-CPA-mult}). As evidence of the robustness of these definitions, we show that they  are equivalent; this strengthens our results since security in the most general case follows from security for the simplest definition.
The proof of equivalence is similar to the classical case (see, \emph{e.g.}~\cite{KL08}), and is included in App.~\ref{sec:appendix-equivalent-q-IND-CPA}\iftoggle{crypto}{}{ for completeness}.} We note that, by taking the evaluation key to be empty, our definitions\iftoggle{crypto}{}{ and theorems} are trivially applicable to the scenario of quantum public-key encryption (\emph{i.e.}~without a homomorphic property).

\iftoggle{crypto}{}{\paragraph{CPA security}} The CPA indistinguishability experiment is given below and illustrated in Fig.~\ref{fig:q-IND-CPA1}. The experiment interacts with an adversary $\advA$, which is a pair of \iftoggle{crypto}{poly-nomial-time}{polynomial-time} quantum algorithms $(\advA_1,\advA_2)$ (which we also \iftoggle{crypto}{call}{refer to as} adversaries). \iftoggle{crypto}{}{The first algorithm $\advA_1$ implements a quantum channel $D(\regR_{evk})\rightarrow D(\mathcal{M}\otimes\mathcal{E})$ conditioned on $pk$, where $\mathcal{E}$ is an arbitrary environment. The second algorithm $\advA_2$ maps $D(\mathcal{C}\otimes\mathcal{E})$ to a bit.}

\begin{figure}[h]
\centering
\begin{tikzpicture}

\node at (0,0) {
\begin{tikzpicture}
\draw (-.25,1.54) -- (2,1.54);
\draw (-.25,1.46) -- (2,1.46);
\draw (.5,0)--(3.5,0);
\draw (.5,1)--(3.5,1);
\draw (-.25,0.04)--(.5,0.04);
\draw (-.25,-.04)--(.5,-.04);
\draw (-.25,1) -- (.5,1);
\draw (3.9,.46)--(4.4,.46);
\draw (3.9,.54)--(4.4,.54);

\filldraw[fill=white] (-.75,-.25) rectangle (-.25,1.75);
\node[rotate=90] at (-.5,.75) {\small$\KeyGen(1^\kappa)$};
\node at (0,1.7) {$pk$};
\node at (.1,1.2) {\small $\mathcal{R}_{evk}$};
\node at (0,.2) {$pk$};

\filldraw[fill=white] (.5,1.25) rectangle (1,-.25);
\node at (.75,.5) {$\advA_1$};

\node at (1.25,1.2) {${\small\cal M}$};
\node at (1.25,.2) {${\small\cal E}$};

\filldraw[fill=white] (1.7,1.75) rectangle (2.7,.75);
\node at (2.2,1.25) {$\Xi_{\QHE}^{{\sf cpa},r}$};

\node at (2.95,1.2) {$\small\cal C$};

\filldraw[fill=white] (3.4,1.25) rectangle (3.9,-.25);
\node at (3.65,.5) {$\advA_2$};

\node at (4.65,.5) {$r'$};
\end{tikzpicture}
};

\node at (6,0) {
\begin{tikzpicture}
\node at (0,1) {$\Xi_{\QHE}^{{\sf cpa},0}$:};
\draw (.75,1)--(1.5,1);
\node at (1.65,1) {\meas};
\node at (2.5,1) {$\ket{\mathbf{0}}$};
\draw (2.75,1)--(4.75,1);
\node at (3,1.2) {\small $\cal M$};
\filldraw[fill=white] (3.25,1.25) rectangle (4.25,.75);
\node at (3.75,1) {$\mathsf{Enc}_{pk}$};
\node at (4.5,1.2) {\small $\cal C$};
\node at (1,1.2) {\small $\cal M$};

\node at (0,0) {$\Xi_{\QHE}^{{\sf cpa},1}$:};
\draw (.75,0)--(4.75,0);
\filldraw[fill=white] (3.25,.25) rectangle (4.25,-.25);
\node at (3.75,0) {$\mathsf{Enc}_{pk}$};
\node at (1,.2) {\small $\cal M$};
\node at (4.5,.2) {\small $\cal C$};

\draw[dashed] (-.75,1.5) rectangle (5,.5);

\draw[dashed] (-.75,-.5) rectangle (5,.5);

\end{tikzpicture}
};

\end{tikzpicture}
\caption{The quantum CPA indistinguishability experiment.}\label{fig:q-IND-CPA1}
\end{figure}

\noindent\textbf{The quantum CPA indistinguishability experiment $\mathsf{PubK^{cpa}_{\advA, QHE}} (\kappa)$}
\begin{enumerate}
\item $\mathsf{KeyGen}(1^\kappa)$ is run to obtain keys $(pk,sk,\rho_{evk})$.
\item Adversary $\advA_1$ is given $(pk,\rho_{evk})$ and outputs a quantum state on $\mathcal{M} \otimes \cal E$.
\item For $r\in \{0,1\}$, let   $\Xi_{\QHE}^{{\sf cpa},r}: D(\mathcal{M}) \rightarrow D(\mathcal{C})$ be:
$\Xi_{\QHE}^{{\sf cpa},0}(\rho)=  {\sf QHE.Enc}_{pk}(\egoketbra{\textbf{0}})$ and $\Xi_{\QHE}^{{\sf cpa},1}(\rho)=  {\sf QHE.Enc}_{pk}(\rho)$.
A random bit $r \in \{0,1\}$ is chosen and $\Xi_{\QHE}^{{\sf cpa},r}$ is applied to the state in  $\mathcal{M}$ (the output being a state in $\mathcal{C}$).
\item Adversary $\advA_2$ obtains the system in $\mathcal{C} \otimes \mathcal{E}$ and outputs a bit $r'$.
\item The output of the experiment is defined to be~1 if~$r'=r$ and~$0$~otherwise.  In case $r=r'$, we say that $\advA$ \emph{wins} the experiment.
\end{enumerate}

\iftoggle{crypto}{\vspace{-5pt}}{}
\begin{definition}[Quantum Indistinguishability under Chosen Plaintext Attack (q-IND-CPA)] \label{def:q-IND-CPA1}
A quantum homomorphic encryption scheme $\sf{QHE}$ is \emph{q-IND-CPA} secure if for any quantum \iftoggle{crypto}{polynomial-time}{poly-nomial-time} adversary $\advA = (\advA_1, \advA_2)$ there exists a negligible function $\eta$ such that\iftoggle{crypto}{ $\Pr[\mathsf{PubK^{cpa}_{\advA, QHE}} (\kappa) =1] \leq \frac{1}{2} +  \eta(\kappa)$.}{:
\begin{equation*}
\Pr[\mathsf{PubK^{cpa}_{\advA, QHE}} (\kappa) =1] \leq \frac{1}{2} +  \eta(\kappa)\,.
\end{equation*}}
\end{definition}

\iftoggle{crypto}{
In the case of classical cryptosystems, it is known that IND-CPA security, the classical analogue of Def.~\ref{defn:classical-cipher-q-IND-CPA}, implies a seemingly stronger security against an adversary who can send multiple messages to a challenger. In the quantum case, we can analogously define an experiment similar to $\mathsf{PubK}_{\advA,\QHE}^{\sf cpa}$, but where the adversary prepares a state in $\mathcal{M}^{\otimes t}\otimes \mathcal{M}^{\otimes t}$ and sends it to the challenger, who traces out either the first half or the second half of the system, before applying an encryption map to each of the remaining subspaces. The adversary must then decide which system was traced out.
In the full version \cite{BJ15}, we give a formal definition of this notion of security, which we call q-IND-CPA-mult, and prove the equivalence of q-IND-CPA and q-IND-CPA-mult. This strengthens our results since security in the most general case (q-IND-CPA-mult) follows from security for the simplest definition (q-IND-CPA).
}
{\paragraph{CPA-mult security} The CPA-mult indistinguishability experiment is similar to the CPA scenario above, but in this case the adversary chooses two $t$-tuples of messages, for any $t\geq 1$, and the challenger returns encryptions corresponding to one of the $t$-tuples. The adversary's task is then to guess which of the two $t$-tuples of messages has been encrypted. The experiment is given below\iftoggle{crypto}{}{; the illustration follows closely the one in Fig.~\ref{fig:q-IND-CPA-2} of App.~\ref{sec:appendix-equivalent-q-IND-CPA} (but with single messages replaced by $t$-fold messages)}.

\noindent\textbf{The quantum CPA-mult indistinguishability experiment $\mathsf{PubK^{cpa\text{-}mult}_{\advA, QHE}} (\kappa)$}
\begin{enumerate}
\item $\mathsf{KeyGen}(1^\kappa)$ is run to obtain keys $(pk,sk,\rho_{evk})$.
\item For $r \in \{0,1\}$, and $t\in O(\mathrm{poly}(\kappa))$, let $\mathcal{M}_r =\mathcal{M}_r^1\otimes  \cdots \otimes \mathcal{M}_r^t$, where  $\mathcal{M}_0^i \equiv\mathcal{M}_1^i\equiv\mathcal{M}$ (for all~$i$).
 Adversary $\advA_1$ is given $(pk,\rho_{evk})$ and outputs a quantum state $\rho$ in $\mathcal{M}_0 \otimes \mathcal{M}_1  \otimes \mathcal{E}$.
\item For $r\in \{0,1\}$, let   $\Xi^{{\sf cpa\text{-}mult},r}_{\QHE}: D(\mathcal{M}_0 \otimes \mathcal{M}_1) \rightarrow D(\mathcal{C}^1 \otimes \cdots \otimes \mathcal{C}^t)$ be given by
$\Xi_{\QHE}^{{\sf cpa\text{-}mult},0}(\rho)=  \Tr_{\mathcal{M}_1}({\sf Enc}_{pk}^{\otimes t} \otimes \id_{\mathcal{M}_1})(\rho)$ and $\Xi_{\QHE}^{\textsf{cpa-mult},1}(\rho)=  \Tr_{\mathcal{M}_0}(\id_{\mathcal{M}_0} \otimes {\sf Enc}_{pk}^{\otimes t})(\rho)$.
A random bit $r \in \{0,1\}$ is chosen and $(\Xi_{\QHE}^{\textsf{cpa-mult},r} \otimes \id_{\cal E})$ is applied to $\rho$ (the output being a state in $\mathcal{C}^{\otimes t} \otimes \mathcal{E}$).
\item Adversary $\advA_2$ obtains the system in $\mathcal{C}^{\otimes t} \otimes \mathcal{E}$ and outputs a bit $r'$.
\item The output of the experiment is defined to be~1 if~$r'=r$ and~$0$~otherwise. In case $r=r'$, we say that $\advA$ \emph{wins} the experiment.
\end{enumerate}

\begin{definition}[Quantum Indistinguishability under Multiple Chosen Plaintext Attack]\label{def:q-IND-CPA-mult}
A quantum homomorphic scheme $\sf{QHE}$ is \emph{q-IND-CPA-mult} secure if for all quantum polynomial-time adversaries \iftoggle{crypto}{$\advA$}{$\advA = (\advA_1, \advA_2)$} there exists a negligible function $\eta$ such that\iftoggle{crypto}{ $\Pr[\mathsf{PubK^{{cpa\text{-}mult}}_{\advA, QHE}} (\kappa) =1] \leq \frac{1}{2} +  \eta(\kappa)$.}{:
\begin{equation*}
\Pr[\mathsf{PubK^{{cpa\text{-}mult}}_{\advA, QHE}} (\kappa) =1] \leq \frac{1}{2} +  \eta(\kappa)\,.
\end{equation*}}
\end{definition}

\begin{theorem}[Equivalence of q-IND-CPA and q-IND-CPA-mult] \label{thm:equiv-IND-CPA}
Let $\QHE$ be a quantum homomorphic encryption scheme. Then $\QHE$ is q-IND-CPA if and only if $\QHE$ is q-IND-CPA-mult.
\end{theorem}

\noindent The proof of Thm.~\ref{thm:equiv-IND-CPA} is given in App.~\ref{sec:appendix-equivalent-q-IND-CPA}.
} 

\vspace{-10pt}
\subsection{Correctness and Compactness of \iftoggle{crypto}{QHE}{Quantum Homomorphic Encryption}}
\label{sec:QHE-correctness}
\vspace{-5pt}
Next, we give\iftoggle{crypto}{}{, in Def.~\ref{defn:C-homomorphism},} a notion that encapsulates correctness of both encryption and evaluation, with respect to a class
$\mathscr{S}$ of quantum circuits\iftoggle{crypto}{}{ (when $\mathscr{S}$ is a strict subset of all computations, the literature  sometimes  refers to this as a ``somewhat homomorphic'' scheme)}.   In the classical context, it is common to restrict attention to circuits that output a single bit, since any deterministic string can be computed bit-by-bit. We cannot do this quantumly, as a quantum state cannot be described\iftoggle{crypto}{}{, or generated,} qubit-by-qubit. We therefore consider correctness
as a global property of the output.
 Furthermore, as quantum data can be entangled, we require that a correct scheme preserve this entanglement and thus explicitly include an auxiliary space in the definition below.

\iftoggle{crypto}{\vspace{-5pt}}{}
\begin{definition}[$\mathscr{S}$-homomorphic]
 \label{defn:C-homomorphism}
  Let $\mathscr{S} = \{\classS_\kappa\}_{\kappa \in \mathbb{N}}$ be a class of quantum circuits.
A quantum encryption scheme  \textsf{QHE} is $\mathscr{S}$-homomorphic (or homomorphic for \iftoggle{crypto}{}{the class }$\mathscr{S}$) if
for any sequence of circuits $\{\mathsf{C}_\kappa \in \mathscr{S}_\kappa\}_{\kappa}$ with induced channels $\Phi_{\mathsf{C}_\kappa}:{\cal M}^{\otimes n(\kappa)}\rightarrow {\cal M}^{\otimes m(\kappa)}$, and  input $\rho\in D({\mathcal{M}^{\otimes n(\kappa)}\otimes \regE})$, there exists a negligible function $\eta$ such that for $(pk,  sk, \rho_\evk) \leftarrow \mathsf{QHE.Keygen(1^\kappa)}$:
\iftoggle{crypto}{\vspace{-5pt}}{}
\begin{equation}
\label{eqn:C-homomorphism}
\Delta\left( \mathsf{QHE.Dec}^{\otimes m(\kappa)}_{sk} \(\mathsf{QHE.Eval}^{\mathsf{C}_\kappa} \(\rho_{evk},  \mathsf{QHE.Enc}^{\otimes n}_{pk} (\rho)\)\) , \Phi_{\mathsf{C}_\kappa} (\rho)\right) = \eta(\kappa)\,.
\end{equation}

\end{definition}

We point out two properties of the above definition. First, we do not require that ciphertexts be  decryptable themselves, only that they become
decryptable after homomorphic evaluation, however, as long as $\QHE$ is homomorphic for the class of identity circuits, we can effectively decrypt a ciphertext by first homomorphically evaluating the identity.  Second, we do not require that the output of $\QHE.\Eval$
be able to undergo additional homomorphic evaluations; indeed, \iftoggle{crypto}{if}{in the case that} the evaluation key $\rho_{evk}$ is quantum, it will in general be ``consumed'' by the $\QHE.\Eval$ process, rendering any future applications of $\QHE.\Eval$ impossible.

Analogously to the classical case, we define compactness\iftoggle{crypto}{}{ (also parametrized by a class of circuits~$\mathscr{S}$)}, which requires that the complexity of $\QHE.\Dec$ be independent of the evaluated circuit, ruling out \iftoggle{crypto}{}{trivial quantum fully homomorphic encryption }schemes where applying the circuit is delayed until after decryption\iftoggle{crypto}{}{ (see the text following Def.~\ref{defn:compact-indivisible} for an informal description of the trivial scheme, $\TRIV$)}.

\iftoggle{crypto}{\vspace{-5pt}}{}
\begin{definition}[$\mathscr{S}$-compactness]
Let $\mathscr{S}=\{\mathscr{S}_{\kappa}\}_{\kappa\in\mathbb{N}}$ be a class of quantum circuits. A quantum encryption scheme $\sf{QHE}$ is $\mathscr{S}$-\emph{compact} if there exists a polynomial $p$ such that for any sequence of circuits $\{\mathsf{C}_\kappa \in \mathscr{S}_\kappa\}_\kappa$,
the circuit complexity of applying $\QHE.\Dec$ to the output of $\QHE.\Eval^{\mathsf{C}_{\kappa}}$ is at most $p(\kappa)$.
\iftoggle{crypto}{}{(That is, the circuit complexity of decryption does not depend on the circuit complexity of $\mathsf{C}_{\kappa}$).}

If $\QHE$ is $\mathscr{S}$-compact for $\mathscr{S}$ the class of all quantum circuits over some universal gate set, then we simply say that $\QHE$ is \emph{compact}.
\end{definition}
\iftoggle{crypto}{\vspace{-5pt}}{}

Although this work leaves open the \iftoggle{crypto}{question}{central problem} of quantum fully homomorphic encryption, we have established all the machinery relevant for a formal definition\iftoggle{crypto}{:}{, which we include below.}

\iftoggle{crypto}{\vspace{-5pt}}{}
\begin{definition}[Quantum Fully Homomorphic Encryption] \label{defn:QFHE}
 A scheme is \emph{a quantum fully homomorphic encryption scheme} if it is both compact and homomorphic for the class of all quantum circuits over some universal gate set.
\end{definition}
\iftoggle{crypto}{\vspace{-5pt}}{}

\vspace{-10pt}
\subsection{Indivisible Schemes}
\label{sec:indivisible}
\vspace{-5pt}

In general, a quantum system is not equal to the sum of its parts. Because of this, for one of our schemes (as given in Sec.~\ref{sec:scheme-EPR}), it is convenient (if not necessary, by the no-cloning theorem~\cite{WZ82}) to define the output of $\QHE.\Eval$ as containing, in addition to a series of cipherstates corresponding to each qubit, some auxiliary quantum register, possibly entangled with each cipherstate. Then the decryption operation,  $\QHE.\Dec$ must operate on the entire quantum system, rather than qubit-by-qubit. This is in contrast to a classical scheme, in which we could make a copy of the auxiliary register for each encrypted bit, enabling the decryption of individual bits, without decrypting the entire system.

\iftoggle{crypto}{\vspace{-5pt}}{}
\begin{definition}An \emph{indivisible} quantum homomorphic encryption scheme is a \iftoggle{crypto}{QHE}{quantum homomorphic encryption} scheme with  $\QHE.\Eval$ and $\QHE.\Dec$ re-defined as:
\begin{description}
\item[{Homomorphic Evaluation.}] $\QHE.\Eval^{\mathsf{C}}: D(\regR_{evk}\otimes  \mathcal{C}^{\otimes n}) \rightarrow D(\regR_{aux} \otimes \mathcal{C'}^{\otimes m})$.
 Compared to $\QHE.\Eval$ in a standard QHE, this algorithm outputs an additional auxiliary quantum register $\regR_{aux}$. This extra information is used in the decryption phase. Since the state of $\mathcal{R}_{aux}$ may be entangled with the state of each $\cal C'$, the system in $\mathcal{R}_{aux}\otimes {\cal C'}^{\otimes  m}$ can no longer be considered subsystem-by-subsystem.

\item[{Decryption.}] $\mathsf{QHE.Dec}_{sk} : D(\regR_{aux} \otimes  \mathcal{C'}^{\otimes m}) \rightarrow D(\mathcal{M}^{\otimes m})$.
  For every possible value of $sk$, $\mathsf{Dec}_{sk}$ is a quantum channel that maps an auxiliary register, together with an $m$-fold cipherstate, to an $m$-fold message in $D(\mathcal{M}^{\otimes m})$.

\end{description}
\end{definition}
\iftoggle{crypto}{\vspace{-5pt}}{}

\noindent We need to define compactness for an indivisible scheme\iftoggle{crypto}{}{ (recall that here, there is no notion of separating the individual output systems)}.

\iftoggle{crypto}{\vspace{-5pt}}{}
\begin{definition}[$\mathscr{S}$-compactness for an indivisible scheme]
 \label{defn:compact-indivisible}
Fix a class of quantum circuits, $\mathscr{S}=\{\mathscr{S}_{\kappa}\}_{\kappa\in\mathbb{N}}$. An indivisible \iftoggle{crypto}{QHE}{quantum homomorphic encryption} scheme $\sf{QHE}$ is $\mathscr{S}$-\emph{compact} if there exists a polynomial $p$ such that for any  sequence of circuits $\{\mathsf{C}_{\kappa} \in \mathscr{S}_{\kappa}\}_{\kappa}$
with \iftoggle{crypto}{}{induced }channels $\Phi_{\mathsf{C}_\kappa}:\mathcal{M}^{\otimes n(\kappa)}\rightarrow \mathcal{M}^{\otimes m(\kappa)}$,
the circuit complexity of applying $\QHE.\Dec^{\otimes m(\kappa)}$ to the output of $\QHE.\Eval^{\mathsf{C}_\kappa}$ is at most $p(\kappa,m(\kappa))$. \iftoggle{crypto}{}{(That is, the circuit complexity of decryption does not depend on the circuit complexity of $\mathsf{C}_{\kappa}$).}
\end{definition}
\iftoggle{crypto}{\vspace{-5pt}}{}

The trivial quantum fully homomorphic encryption scheme, $\TRIV$, is easily phrased as an indivisible scheme. \iftoggle{crypto}{Informally, $\TRIV$ is defined by taking $\TRIV.\KeyGen$ and $\TRIV.\Enc$ from any public-key encryption scheme, letting $\TRIV.\Eval^{\mathsf{C}}$ append a description of $\mathsf{C}$ to the cipherstate, and $\TRIV.\Dec$ decode the cipherstate, and then apply $\mathsf{C}$.}{Informally, $\TRIV$ is the following:
\begin{enumerate}
\item The algorithms $\TRIV.\KeyGen$ and $\TRIV.\Enc$ are taken from any quantum public-key encryption scheme.
\item  The algorithm $\TRIV.\Eval$ simply sets $\regR_{aux}$ to be the target circuit, $\mathsf{C}$, and otherwise outputs the cipherstates corresponding to the encrypted inputs.
\item The algorithm $\TRIV.\Dec$ first decrypts the cipherstates, then applies $\mathsf{C}$ and outputs the result.
\end{enumerate}}
Clearly, $\TRIV$ is homomorphic, but it is not compact, since $\TRIV.\Dec$ must evaluate the quantum circuit $\mathsf{C}$, and so its complexity scales with $G(\mathsf{C})$, the number of gates in $\mathsf{C}$.

Although  a decryption procedure with any dependence on $G$, or any other property of $\mathsf{C}$, is not compact, it is still interesting  to consider schemes whose decryption procedure has complexity that scales sublinearly in $G$ (such schemes are called \emph{quasi-compact} schemes~\cite{GentryThesis}). We give a formal definition that quantifies this notion for indivisible quantum homomorphic encryption schemes.

\iftoggle{crypto}{\vspace{-5pt}}{}
\begin{definition}[quasi-compactness]
Let $\classS=\{\classS_{\kappa}\}_{\kappa}$ be the set of all quantum circuits over some fixed universal gate set. \iftoggle{crypto}{For any $f:\classS\rightarrow\mathbb{R}_{\geq 0}$, an indivisible QHE}{Let $f:\classS\rightarrow \mathbb{R}_{\geq 0}$ be some function on the circuits in $\classS$. An indivisible quantum homomorphic encryption} scheme $\QHE$ is \emph{$f$-quasi-compact} if there exists a polynomial $p$ such that for any sequence of circuits $\{\mathsf{C}_{\kappa}\in\classS_{\kappa}\}_{\kappa}$ with induced channels
$\Phi_{\mathsf{C}_{\kappa}}:\mathcal{M}^{\otimes n(\kappa)}\rightarrow \mathcal{M}^{\otimes m(\kappa)}$, the circuit complexity of decrypting the output of $\QHE.\Eval^{\mathsf{C}_{\kappa}}$ is at most $f(\mathsf{C}_{\kappa})p(\kappa,m(\kappa))$.
\end{definition}
\iftoggle{crypto}{\vspace{-5pt}}{}
This definition allows us to consider schemes whose decryption complexity scales with some property of the evaluated circuit. We consider such a scaling non-trivial when it is smaller than~$G(\mathsf{C})$, the number of gates in $\mathsf{C}$.

\vspace{-10pt}
\iftoggle{crypto}{\subsection{Symmetric-Key Quantum Homomorphic Encryption}}
{\subsection{Quantum Homomorphic Encryption in the Symmetric-Key Setting}}
\label{sec:symmetric-key}
\vspace{-5pt}

We have defined quantum homomorphic encryption as a \emph{public-key} encryption scheme.  For technical reasons, our final scheme, $\AUX$ is given in the symmetric-key setting, so in this section we define \iftoggle{crypto}{\emph{symmetric-key}}{functionality and security for symmetric-key} quantum homomorphic encryption. In the case of classical \iftoggle{crypto}{FHE}{fully homomorphic encryption}, symmetric-key encryption is known to be \emph{equivalent} to public-key encryption~\cite{R11}. In the quantum case, this is not known.
\iftoggle{crypto}{}{
}This section also contains the definition of a \emph{bounded} QHE scheme, which we again require for technical reasons in our symmetric-key scheme, $\AUX$.

\iftoggle{crypto}{\vspace{-5pt}}{}
\begin{definition}
A \emph{symmetric-key} \iftoggle{crypto}{QHE}{quantum homomorphic encryption} scheme is a quantum homomorphic encryption scheme  with  $\QHE.\KeyGen$ and $\QHE.\Enc$ re-defined as:
\begin{description}

\item[{Key Generation.}]
 $\QHE.\KeyGen: 1^\kappa \rightarrow (sk, \rho_{evk})$.
 This algorithm takes a unary representation of the security parameter as input and outputs  a secret encryption/decryption key~$\sk$ and a quantum evaluation key $\rho_{evk}\in D(\regR_{evk})$.

\item[{Encryption.}]
$\QHE.\Enc_{sk}: D(\mathcal{M}) \rightarrow D(\mathcal{C})$. For every possible value of $sk$, the quantum channel $\mathsf{Dec}_{sk}$ maps a state in the message space $\cal M$ to a state (the \emph{cipherstate}) in the cipherspace ${\cal C}$.

\end{description}
\end{definition}
\iftoggle{crypto}{\vspace{-5pt}}{}

 Next, we define a quantum homomorphic encryption scheme that
is \emph{bounded} by $n$, which forces the number of ciphertexts encrypted by $sk$ to be  at most~$n$. Furthermore, the scheme maintains a counter, $d$, of the number of previous encryptions, which can be thought of as allowing the scheme to avoid key reuse.

\iftoggle{crypto}{\vspace{-5pt}}{}
\begin{definition}A \emph{bounded} symmetric-key \iftoggle{crypto}{QHE}{quantum homomorphic encryption} scheme is a \iftoggle{crypto}{symmetric-key}{sym-metric-key} \iftoggle{crypto}{QHE}{quantum homomorphic encryption} scheme with  $\QHE.\KeyGen$, $\QHE.\Enc$, and $\QHE.\Dec$ re-defined as:
\begin{description}

\item[{Key Generation.}]
 $\QHE.\KeyGen: (1^\kappa, 1^n) \rightarrow (sk, \rho_{evk})$.

\item[{Encryption.}]
$\QHE.\Enc_{sk,d}: D(\mathcal{M}) \rightarrow D(\regC )$. Every time $\QHE.\Enc_{sk,d}$ is called, the register containing $d$ is incremented: $d\leftarrow d+1$. If $d>n$, $\QHE.\Enc_{sk,d}$ outputs $\bot$, indicating an error.

\item[{Decryption.}]
$\QHE.\Dec_{sk,d}:D(\mathcal{C}')\rightarrow D(\mathcal{M})$.
\end{description}
\end{definition}
\iftoggle{crypto}{\vspace{-5pt}}{}

\iftoggle{crypto}{
We can define q-IND-CPA security for the symmetric-key setting by allowing the adversary access to an encryption oracle $\mathsf{Enc}_{sk}(\cdot)$. We give details in \cite{BJ15}.
}
{\paragraph{Security of Symmetric Key Schemes}
In order to define \emph{indistinguishability under chosen plaintext attacks} in the symmetric-key setting,  we must equip the adversary with an encryption oracle $\Enc_{sk}(\cdot)$.
An \emph{adversary with access to an encryption oracle}, $\advA$ is a tuple of quantum channels $(\advA^{(1)}, \dots, \advA^{(q+1)})$, such that $\advA^{(1)}:D({\cal X})\rightarrow D(\mathcal{M}\otimes\mathcal{E})$ for some space $\cal X$, for $i=2,\dots,q$, $\advA^{(i)}:D(\mathcal{C}\otimes\mathcal{E})\rightarrow D(\mathcal{M}\otimes\mathcal{E})$, and $\advA^{(q+1)}:D(\mathcal{C}\otimes\mathcal{E})\rightarrow D(\mathcal{Y})$ for some space $\mathcal{Y}$. The interaction of the adversary and the encryption oracle is shown in Fig.~\ref{fig:adv-with-oracle}, and for the case of a bounded encryption scheme, in which the oracle also updates a counter, in Fig.~\ref{fig:adv-with-oracle-bounded}.

\iftoggle{crypto}{}{
\begin{figure}[h]
\centering
\input{fig-adv-with-oracle.tex}
\end{figure}
}

\begin{figure}[h]
\centering
\input{fig-adv-with-oracle-bounded.tex}
\end{figure}

Just as in the public-key setting, we can define a quantum CPA indistinguishability experiment for the symmetric-key setting, $\mathsf{SymK}_{\advA,\QHE}^{\sf cpa}(\kappa)$. An adversary for $\mathsf{SymK}_{\advA,\QHE}^{\sf cpa}(\kappa)$ is a pair of adversaries with access to an encryption oracle $\advA=(\advA_1,\advA_2)=(\advA_1^{(1)},\dots,\advA_1^{(q+1)},\advA_2^{(1)},\dots,\advA_2^{(q'+1)})$ ($q$~is the number of oracle calls before the challenger is called, and $q'$ is the number of oracle calls after the challenger is called). The experiment $\mathsf{SymK}_{\advA,\QHE}^{\sf cpa}(\kappa)$ is defined below, and shown in Fig.~\ref{fig:symmetric-cpa}.\looseness=-1

\noindent\textbf{The quantum symmetric-key CPA indistinguishability experiment} $\mathsf{SymK^{cpa}_{\advA,\QHE}}(\kappa)$
\begin{enumerate}
\item $\KeyGen(1^\kappa)$ is run to obtain keys $(sk,\rho_{evk})$.
\item $\advA_1$ is given $\rho_{evk}$, and may make a polynomial number of calls to an encryption oracle $\QHE.\Enc_{sk}$ before outputting a quantum state in message space $\cal M$ and environment register $\cal E$.
\item A random bit $r\in\{0,1\}$ is chosen and $\Xi_{\QHE}^{{\sf cpa},r}$ is applied to the state in $\cal M$ (the output being a state in $\cal C$).
\item Adversary $\advA_2$ obtains the system $\mathcal{C}\otimes \mathcal{E}$ and may make a polynomial number of calls to an encryption oracle $\QHE.\Enc_{sk}$ before outputting a bit $r'$.
\item The output of the experiment is defined to be 1 if $r=r'$ and 0 otherwise. In case $r=r'$, we say that $\advA$ \emph{wins} the experiment.
\end{enumerate}

\begin{figure}[h]
\centering
\begin{tikzpicture}
\node at (0,0){
\begin{tikzpicture}
\draw (0,2.54)--(5.5,2.54);
\draw (0,2.46)--(5.5,2.46);
\draw (0,.75)--(5,.75);
		\draw (5,.415)--(5.5,.415);
		\draw (5,.335)--(5.5,.335);
\draw (1,0)--(5,0);

\filldraw[fill=white] (-.5,.5) rectangle (0,2.75);
\node[rotate=90,align=center] at (-.25,1.625) {\small ${\sf QHE.KeyGen}$};

\node at (.25,2.7) {$sk$};
\node at (.4,.95) {\small $\mathcal{R}_{evk}$};

\filldraw[fill=white] (1,1.25) rectangle (1.5,2.75);
\node[rotate=90] at (1.25,2) {\small $\QHE.\Enc$};
\draw[->] (1.1,1)--(1.1,1.25);
\draw[->] (1.4,1.25)--(1.4,1);
\filldraw[fill=white] (1,-.25) rectangle (1.5,1);
\node at (1.25,.375) {$\advA_1$};

\node at (1.75,2.7) {$sk$};
\node at (1.75,.95) {\small $\cal M$};
\node at (1.75,.2) {\small $\cal E$};

\filldraw[fill=white] (2.5,.5) rectangle (3.5,2.75);
\node at (3,1.625) {$\Xi_{\QHE}^{\textsf{cpa},r}$};

\node at (3.75,2.7) {$sk$};
\node at (3.75,.95) {\small $\cal C$};

\filldraw[fill=white] (4.5,1.25) rectangle (5,2.75);
\node[rotate=90] at (4.75,2) {\small $\QHE.\Enc$};
\draw[->] (4.6,1)--(4.6,1.25);
\draw[->] (4.9,1.25)--(4.9,1);
\filldraw[fill=white] (4.5,-.25) rectangle (5,1);
\node at (4.75,.375) {$\advA_2$};

\node at (5.25,2.7) {$sk$};

\node at (5.75,.375) {$r'$};

\end{tikzpicture}
};

\node at (8,0){
\begin{tikzpicture}
\draw (0,2.54)--(5.5,2.54);
\draw (0,2.46)--(5.5,2.46);
\draw (.6,1.54)--(5.5,1.54);
\draw (.6,1.46)--(5.5,1.46);
\draw (0,.75)--(5,.75);
		\draw (5,.415)--(5.5,.415);
		\draw (5,.335)--(5.5,.335);
\draw (1,0)--(5,0);

\filldraw[fill=white] (-.5,.5) rectangle (0,2.75);
\node[rotate=90,align=center] at (-.25,1.625) {\small ${\sf QHE.KeyGen}$};

\node at (.25,2.7) {$sk$};
\node at (.5,1.5) {$1$};
\node at (.4,.95) {\small $\mathcal{R}_{evk}$};

\filldraw[fill=white] (1,1.25) rectangle (1.5,2.75);
\node[rotate=90] at (1.25,2) {\small $\QHE.\Enc$};
\draw[->] (1.1,1)--(1.1,1.25);
\draw[->] (1.4,1.25)--(1.4,1);
\filldraw[fill=white] (1,-.25) rectangle (1.5,1);
\node at (1.25,.375) {$\advA_1$};

\node at (1.75,2.7) {$sk$};
\node at (2,1.7) {\small$q+1$};
\node at (1.75,.95) {\small $\cal M$};
\node at (1.75,.2) {\small $\cal E$};

\filldraw[fill=white] (2.5,.5) rectangle (3.5,2.75);
\node at (3,1.625) {$\Xi_{\QHE}^{\textsf{cpa},r}$};

\node at (3.75,2.7) {$sk$};
\node at (4,1.7) {\small$q+2$};
\node at (3.75,.95) {\small $\cal C$};

\filldraw[fill=white] (4.5,1.25) rectangle (5,2.75);
\node[rotate=90] at (4.75,2) {\small $\QHE.\Enc$};
\draw[->] (4.6,1)--(4.6,1.25);
\draw[->] (4.9,1.25)--(4.9,1);
\filldraw[fill=white] (4.5,-.25) rectangle (5,1);
\node at (4.75,.375) {$\advA_2$};

\node at (5.25,2.7) {$sk$};
\node at (5.85,1.7) {\small$q'+q+2$};

\node at (5.75,.375) {$r'$};

\end{tikzpicture}
};
\end{tikzpicture}
\caption{The quantum CPA experiment for symmetric-key systems (left) and bounded symmetric-key systems (right).
}\label{fig:symmetric-cpa}
\end{figure}

\begin{definition}[Quantum Indistinguishability under Chosen Plaintext Attack (q-IND-CPA) for Symmetric Key Schemes]
\label{def:q-ind-cpa-Symmetric}
A symmetric-key quantum homomorphic encryption scheme $\QHE$ is q-IND-CPA secure if for all quantum polynomial-time adversaries with oracle access, \iftoggle{crypto}{$\advA$}{$\advA=(\advA_1^{(1)},\dots,\advA_1^{(q+1)},\advA_2^{(1)},\dots,\advA_2^{(q'+1)})$}, there exists a negligible function $\eta$ such that: \iftoggle{crypto}{$\Pr[\mathsf{SymK_{\advA,\QHE}^{cpa}}(\kappa)=1]\leq \frac{1}{2}+\eta(\kappa)$.}{
$$\Pr[\mathsf{SymK_{\advA,\QHE}^{cpa}}(\kappa)=1]\leq \frac{1}{2}+\eta(\kappa).$$}
\end{definition}

Similar to the case of public-key encryption (Sec.~\ref{sec:security-QHE}), it is straightforward to give the seemingly stronger variant of q-IND-CPA, \emph{q-IND-CPA-mult}, which is defined identically to the public-key case (Def.~\ref{def:q-IND-CPA-mult}) but with an adversary having access to an encryption oracle. However, just as in the public-key case, it turns out that these definitions are equivalent.

\begin{theorem}[Equivalence of q-IND-CPA and q-IND-CPA-mult in symmetric-key schemes] \label{thm:equiv-IND-CPA-sym}
Let $\QHE$ be a symmetric-key quantum homomorphic scheme. Then $\QHE$ is q-IND-CPA if and only if $\QHE$ is q-IND-CPA-mult.
\end{theorem}
\noindent The proof of Thm.~\ref{thm:equiv-IND-CPA-sym} is virtually identical to that of Thm.~\ref{thm:equiv-IND-CPA}, given in App.~\ref{sec:appendix-equivalent-q-IND-CPA}.
}

\vspace{-5pt}
\section{Main Contributions}
\label{sec:main-contributions}
\vspace{-5pt}
We now formally state our main results\iftoggle{crypto}{}{ (formal schemes and proofs are given in Sec.~\ref{sec:scheme-Clifford}--\ref{sec:scheme-AUX})}.
Our first theorem, Thm.~\ref{thm:main-Clifford}, establishes quantum homomorphic encryption for Clifford circuits.

\iftoggle{crypto}{\vspace{-5pt}}{}
\begin{theorem}(Clifford scheme, $\CL$)
\label{thm:main-Clifford}
Let $\classS$ be the class of Clifford circuits. Then assuming the existence of a classical fully homomorphic encryption scheme that is q-IND-CPA secure, there exists a quantum homomorphic encryption scheme that is q-IND-CPA, compact and $\mathscr{S}$-homomorphic.
\end{theorem}
\iftoggle{crypto}{\vspace{-5pt}}{}

Next, we consider two variants of the scheme given by Thm.~\ref{thm:main-Clifford}. Each variant deals with non-Clifford \iftoggle{crypto}{$\tgate$-gates}{group gates --- in our case $\tgate$-gates ---} in a different way. The first scheme, described in Thm.~\ref{thm:main-EPR} and formally defined in Sec.~\ref{sec:scheme-EPR}, uses entanglement to implement $\tgate$-gates, resulting in a \iftoggle{crypto}{QHE}{quantum homomorphic encryption} scheme in which the complexity of decryption scales with the number of $\tgate$-gates in the homomorphically evaluated circuit.

\iftoggle{crypto}{\vspace{-5pt}}{}
\begin{theorem}(entanglement-based scheme, $\EPR$) \label{thm:main-EPR}
Let $\classS$ be the set of all quantum circuits over the universal gate set $\{\xgate,\zgate,\pgate,\hgate,\cnot,\tgate\}$\iftoggle{crypto}{}{ (as well as single-qubit preparation and measurement)}. Then assuming {the existence of} a classical fully homomorphic encryption scheme that is q-IND-CPA secure, there exists an indivisible quantum homomorphic encryption scheme that is q-IND-CPA, $\classS$-homomorphic and $R^2$-quasi-compact, where $R(\mathsf{C})$ is the number of $\tgate$-gates in a circuit~$\mathsf{C}$.
\end{theorem}
\iftoggle{crypto}{\vspace{-5pt}}{}

The compactness of the scheme $\EPR$ is nontrival for all circuits in which $R^2\ll G$, where $G$ is the number of gates.

Our second scheme, formally defined in Sec.~\ref{sec:scheme-AUX}, is based on the use of auxiliary qubits to implement $\tgate$-gates, resulting in a \iftoggle{crypto}{QHE}{quantum homomorphic encryption} scheme that is homomorphic for circuits with constant $\tgate$-depth, as described in the following theorem:

\iftoggle{crypto}{\vspace{-5pt}}{}
\begin{theorem}(auxiliary-qubit scheme, $\AUX$) \label{thm:main-AUX}
Fix a constant $L$. Let $\classS$ be the set of quantum circuits over the universal gate set $\{\xgate,\zgate,\pgate,\hgate,\cnot,\tgate\}$\iftoggle{crypto}{}{ (as well as single-qubit preparation and measurement)} with \mbox{$\tgate$-depth} at most~$L$. Then assuming the existence of a classical fully homomorphic encryption scheme that is q-IND-CPA secure, there exists a bounded symmetric-key quantum homomorphic encryption scheme that is q-IND-CPA, $\classS$-homomorphic and compact.
\end{theorem}
\iftoggle{crypto}{\vspace{-5pt}}{}

The QHE scheme in Thm.~\ref{thm:main-AUX} can be seen as somewhat analogous to an important building block in classical fully homomorphic encryption: a \emph{levelled} fully homomorphic scheme, which is a scheme that takes a parameter $L$, which is an a-priori bound on the \emph{depth} of the circuit that can be evaluated. However, we note that in contrast to a levelled fully homomorphic scheme, in which operations are polynomial in $L$, the complexity of our scheme is a polynomial of degree exponential in $L$, so we really require $L$ to be constant.

As previously noted, Thm.~\ref{thm:main-EPR} and~\ref{thm:main-AUX} are complementary: the scheme $\EPR$ becomes less compact as the \emph{number} of $\tgate$-gates increases, while the scheme~$\AUX$ becomes inefficient as the \emph{depth} of $\tgate$-gates increases.

\vspace{-10pt}
\iftoggle{crypto}{\section{Homomorphic Encryption for Clifford Circuits: $\CL$}}
{\section{Scheme $\CL$: Homomorphic Encryption for Clifford Circuits}}
\label{sec:scheme-Clifford}
\vspace{-5pt}

In this section, we present $\CL$, a compact quantum homomorphic encryption scheme for \iftoggle{crypto}{Clifford circuits}{stabilizer circuits, which consist of Clifford circuits combined with measurements and single-qubit preparation}. This is a building block for the schemes that follow in Sec.~\ref{sec:scheme-EPR} and~\ref{sec:scheme-AUX}. \iftoggle{crypto}{In the full version \cite{BJ15}, we prove that $\CL$ is q-IND-CPA secure, and homomorphic for Clifford circuits, hence proving Thm.~\ref{thm:main-Clifford}.}{The main theorem we prove is Thm.~\ref{thm:main-Clifford}, which follows directly from Thm.~\ref{thm:correctness:Clifford},~\ref{thm:compactness:Clifford} and~\ref{thm:security:Clifford}.}

By definition, Clifford circuits  conjugate Pauli operators to Pauli operators~\cite{Got98}. In other words, for any Clifford $\mathsf{C}$, and any Pauli, $\mathsf{Q}$, there exists a Pauli $\mathsf{Q}'$ such that $\mathsf{C}\mathsf{Q}=\mathsf{Q}'\mathsf{C}$.
Furthermore, applying a random Pauli operator is a perfectly secure symmetric-key quantum encryption scheme: the quantum one-time pad\iftoggle{crypto}{}{ (see Sec.~\ref{sec:prelim-QOTP})}. \iftoggle{crypto}{Thus,}{Combining these observations, we see that} it is possible to perform any Clifford circuit on quantum data that is encrypted using the quantum one-time pad. We can apply the desired Clifford, $\mathsf{C}$, to the encrypted state $\mathsf{Q}\ket{\psi}$ to get $\mathsf{Q}'(\mathsf{C}\ket{\psi})$. Now decrypting the state requires applying the Pauli $\mathsf{Q}'$. If $\mathsf{Q}$ can be described by the encryption key $(a_1,\dots,a_n,b_1,\dots,b_n)$ --- that is, $\mathsf{Q}=\xgate^{a_1}\zgate^{b_1}\otimes \dots\otimes \xgate^{a_n}\zgate^{b_n}$ --- then $\mathsf{Q}'$ can be described by some key $(a_1',\dots,a_n',b_1',\dots,b_n')$ depending on $\mathsf{C}$ and $(a_1,\dots,a_n,b_1,\dots,b_n)$. We describe this dependence by a function $f^{\mathsf{C}}:\mathbb{F}_2^{2n}\rightarrow \mathbb{F}_2^{2n}$, which we call a \emph{key update rule}.
We need only consider key update rules for each gate in our gate set, which consists of \iftoggle{crypto}{}{single-qubit measurement, single-qubit preparation, and }the one- and two-qubit gates in $\{\xgate,\zgate,\pgate,\cnot,\hgate\}$. For a single-qubit gate $\mathsf{C}$, since the only keys that are affected are those corresponding to the wire to which $\mathsf{C}$ is applied, an update rule can be more succinctly described by a pair of functions $f_a^{\mathsf{C}},f_b^{\mathsf{C}}:\mathbb{F}_2^2\rightarrow \mathbb{F}_2$ such that when $\mathsf{C}$ is applied to the \th{i} wire, $a_i'=f_a^{\mathsf{C}}(a_i,b_i)$ and $b_i'=f_b^\mathsf{C}(a_i,b_i)$:
\begin{center}
\noindent\begin{tikzpicture}
\node at (-1.3,0) {$\xgate^{a_i}\zgate^{b_i}\ket{\psi}$};
\draw (-.45,0)--(.75,0);
\filldraw[fill = white] (-.1,.25) rectangle (.4,-.25);
\node at (.15,0) {$\mathsf{C}$};
\node at (1.75,0) {$\xgate^{a_i'}\zgate^{b_i'}\mathsf{C}\ket{\psi}$};

\node at (6.5,0) {$a_i\leftarrow a_i'=f_{a}^{\mathsf{C}}(a_i,b_i),\;\; b_i\leftarrow b_i'=f_{b}^{\mathsf{C}}(a_i,b_i)$};
\end{tikzpicture}
\end{center}
\noindent For the \iftoggle{crypto}{}{two-qubit }$\cnot$-gate, the update rule is described by a 4-tuple of functions, since $\cnot$ acts on two wires. We give the key update rules for all gates in \iftoggle{crypto}{the full version \cite[App.~C]{BJ15}}{App.~\ref{appendix:Key-update-rules-stabilizer}}.\iftoggle{crypto}{ (We also give key update rules for single-qubit measurement and qubit preparation, so that our scheme is actually homomorphic for stabilizer circuits.) }{ }By applying these rules after each gate, we can update the key so that the output is correctly decrypted\iftoggle{crypto}{ (since we are actually carrying out computations on encrypted quantum data---in contrast to merely simulating a quantum computation---we note that all gates except the Pauli gates require  quantum operations)}{}. Such a technique was already used, \emph{e.g.} in~\cite{Chi05,QCED,B15}.

This solution, however, requires that the key updates be executed by the party holding the encryption keys:  an ``easy''  classical computation, but nevertheless a computation that is polynomial in the \emph{size} of the circuit. In the context of quantum homomorphic encryption, the challenge is therefore to allow the execution of \emph{arbitrary} Clifford circuits, while maintaining the compactness condition. Here, we present a quantum public-key encryption scheme which is a hybrid of the quantum one-time pad and of a classical fully homomorphic encryption scheme. This encryption scheme is used to perform key updates on encrypted quantum one-time pad keys, enabling the computation of arbitrary Clifford group circuits on the encrypted quantum states, while maintaining the compactness condition. More precisely, to homomorphically evaluate a Clifford circuit consisting of a sequence of gates $\mathsf{c}_1,\dots,\mathsf{c}_G$, we apply the gates to the quantum one-time pad encrypted message, and homomorphically evaluate the function $f^{{\sf c}_1}\circ\dots\circ f^{{\sf c}_G}$ on the encrypted one-time pad keys $a_1,\dots,a_n,b_1,\dots,b_n$, where $\circ$ denotes function composition. To accomplish this, we keep track of functions for each bit of the quantum one-time pad encryption key, $\{f_{a,i},f_{b,i}\}_{i=1}^n$. Since each of the key update rules \iftoggle{crypto}{(see \cite{BJ15})}{presented in App.~\ref{appendix:Key-update-rules-stabilizer}} is linear, each $f_{a,i}$ and $f_{b,i}$ is a linear polynomial in $\mathbb{F}_2[a_1,\dots,a_n,b_1,\dots,b_n]$ (from the perspective of the evaluation procedure, $a_1,\dots,a_n,b_1,\dots,b_n$ are unknowns), so we refer to them as \emph{key-polynomials}. Before we begin to evaluate the circuit, the key polynomials are the monomials $f_{a,i}=a_i$ and $f_{b,i}=b_i$.
As we evaluate each gate~$\mathsf{c}_j$, we update the {key-polynomials} corresponding to the affected wires by composing them with the key update rules. To compute the new encrypted one-time pad keys once the circuit is complete, we homomorphically evaluate each key-polynomial on the old encrypted one-time pad keys.
\iftoggle{crypto}{We}{
It is interesting to} note that since the key update rules (\iftoggle{crypto}{see \cite{BJ15}}{App.~\ref{appendix:Key-update-rules-stabilizer}})\iftoggle{crypto}{}{ for stabilizer circuit elements} are all linear, for the scheme~$\CL$, the underlying classical fully homomorphic scheme  only needs to be additively homomorphic.

We define our scheme $\CL$ as a QHE scheme. Here and throughout, we assume $\mathsf{HE}$ to be a classical \iftoggle{crypto}{FHE}{fully homomorphic encryption} scheme that is q-IND-CPA secure (see Def.~\ref{defn:classical-cipher-q-IND-CPA}\iftoggle{crypto}{}{ and App.~\ref{appendix:Classical-FHE}}). As noted, such a scheme \iftoggle{crypto}{}{(based on the LWE assumption) }could be derived  from~\cite{BV11}. All of our schemes operate on qubit circuits, and encrypt qubit-by-qubit. Thus we fix $\mathcal{M}=\mathbb{C}^{\{0,1\}}$. Ciphertexts  consist of quantum states in $\mathbb{C}^{\{0,1\}}$, combined with classical strings. Specifically, if $C$ is the output space of $\HE.\Enc$, and $C'$ is the output space of $\HE.\Eval$, then we define $\mathcal{C}= \mathbb{C}^{C\times C}\otimes {\cal X}$, where ${\cal X}\equiv\mathbb{C}^{\{0,1\}}$, and $\mathcal{C}'=\mathbb{C}^{C'\times C'}\otimes \cal X$.
\begin{description}
\item[{Key Generation}.] $\CL.\KeyGen(1^\kappa)$. For key generation, execute $(\pk, \sk, \evk) \leftarrow \mathsf{HE.Keygen(1^\kappa)}$. Output the obtained secret key, $\sk$, and public key, $\pk$. The evaluation key $\rho_\evk$ takes the value of the classical state $\rho({\evk})$.

\item[{Encryption}.] $\mathsf{CL.Enc}_{pk}: D(\mathcal{M})\rightarrow D(\mathcal{C})$.
Encryption is defined as\iftoggle{crypto}{}{ the quantum channel that outputs the classical-quantum state}:
\iftoggle{crypto}{\\$\displaystyle\mathsf{CL.Enc}_{pk}(\rho^\mathcal{M})=\sum_{a,b\in\{0,1\}}\frac{1}{4}\rho(\mathsf{HE.Enc}_{pk}(a),\mathsf{HE.Enc}_{pk}(b))\otimes
\QEnc_{a,b}(\rho^{\cal M}).$}
{\begin{equation*}
\mathsf{CL.Enc}_{pk}(\rho^\mathcal{M})=\sum_{a,b\in\{0,1\}}\frac{1}{4}\rho(\mathsf{HE.Enc}_{pk}(a),\mathsf{HE.Enc}_{pk}(b))\otimes
\QEnc_{a,b}(\rho^{\cal M}).
\end{equation*}}

\item[{Homomorphic Evaluation}.]
    $\mathsf{CL.Eval}^{\mathsf{C}}:D(\mathcal{R}_{evk}\otimes \mathcal{C}^{\otimes n})\rightarrow D(\mathcal{C'}^{\otimes m})$.

Suppose $\mathsf{C}=\mathsf{c}_1,\dots,\mathsf{c}_G$ is a Clifford circuit. \iftoggle{crypto}{}{For every $j=1,\dots, {G}$ such that $\mathsf{c}_j$ initializes a fresh qubit, we initialize a new qubit $\CL.\Enc_{pk}(\ket{0}\bra{0})$ and append it to the system. Let $\rho \in D(\mathcal{X}_1\otimes\dots\otimes\mathcal{X}_{m})$, be the composite system consisting of the input quantum system and the  initialized qubits.
\\}
\begin{enumerate}
\item For all $i\in [n]$, \iftoggle{crypto}{set}{set $f_{a,i},f_{b,i}\in\mathbb{F}_2[a_1,\dots,a_n,b_1,\dots,b_n]$ as} $f_{a,i}\leftarrow a_i$, $f_{b,i}\leftarrow b_i$.
\item For $j=1,\dots,G$ such that $\mathsf{c}_j$ is a gate or a measurement:
\begin{enumerate}
\item Apply the gate $\mathsf{c}_j$ to the state: $\rho\leftarrow \mathsf{c}_j\rho\mathsf{c}_j^{-1}$.
\item Compose the key update rules with the key-polynomials of the affected wires:
if $\mathsf{c}_j$ is a single qubit gate or measurement acting on the \th{i} wire, update as $(f_{a,i},f_{b,i})\leftarrow ( f_{a,i}\circ f_a^{\mathsf{c}_j}, f_{b,i}\circ f_b^{\mathsf{c}_j})$. \iftoggle{crypto}{If}{Otherwise, if} $\mathsf{c}_j$ is a $\cnot$-gate acting on wires $i$ and $i'$, update $(f_{a,i},f_{a,i'},f_{b,i},f_{b,i'})$\iftoggle{crypto}{}{ analogously}.
\end{enumerate}
\item Update the classical encryptions by computing
$$c_i = (\HE.\Eval_{evk}^{f_{a,i}}(\tilde{a}_i), \HE.\Eval_{evk}^{f_{b,i}}(\tilde{b}_i)).$$
\item Output $(c_1,\dots,c_m,\rho)$\iftoggle{crypto}{}{ (with registers permuted to fit the prescribed form)}.
\end{enumerate}

\item[Decryption.] $\mathsf{CL.Dec}_{sk}: D (\mathcal{C'})\rightarrow D(\mathcal{M})$.
For $\tilde{a},\tilde{b} \in C'$, \iftoggle{crypto}{decryption is defined:}{the output space of $\HE.\Eval$, decryption is given by the conditional quantum channel:}
\iftoggle{crypto}{\\$\displaystyle\CL.\Dec_{sk}:\ket{\tilde{a}}\bra{\tilde{a}}\otimes \ket{\tilde{b}}\bra{\tilde{b}}\otimes \rho^{\cal X}\mapsto \QDec_{\HE.\Dec_{sk}(\tilde{a}),\HE.\Dec_{sk}(\tilde{b})}(\rho^{\cal X})$.}
{\begin{equation*}
\CL.\Dec_{sk}:\ket{\tilde{a}}\bra{\tilde{a}}\otimes \ket{\tilde{b}}\bra{\tilde{b}}\otimes \rho^{\cal X}\mapsto \QDec_{\HE.\Dec_{sk}(\tilde{a}),\HE.\Dec_{sk}(\tilde{b})}(\rho^{\cal X}),
\end{equation*}
which can be implemented by first decoding the classical registers to obtain $a=\HE.\Enc_{sk}(\tilde{a})$ and $b=\HE.\Enc_{sk}(\tilde{b})$, applying $\QDec_{a,b}$, and then tracing out $\mathbb{C}^{C'\times C'}$.}
\end{description}
\iftoggle{crypto}{}{We have chosen to present  $\mathsf{CL.Enc}_{pk}$ and $\mathsf{CL.Dec}_{sk}$ as quantum channels, since they are easily seen to be polynomial-time implementable. Note, however, that for more complicated quantum channels such as $\CL.\Eval$ we will generally prefer their description in terms of a high-level algorithmic description.}

\iftoggle{crypto}{
\noindent We prove the homomorphic and security properties of $\CL$ in \cite{BJ15}.
}{
\subsection{Analysis of $\CL$}
\label{sec:analysis-CL}

\iftoggle{crypto}{}{We now analyse the various properties of $\CL$.}

\begin{theorem}
\label{thm:correctness:Clifford}
Let $\mathscr{S}$ be the class of Clifford circuits. Then $\CL$ is  $\mathscr{S}$-homomorphic.
\end{theorem}
\begin{proof}
This follows from the circuits in App.~\ref{appendix:Key-update-rules-stabilizer}, as well as the homomorphic property of $\mathsf{HE}$. In particular, since the decrypted values of the ciphertexts are correct (except with exponentially small probability), then \iftoggle{crypto}{}{Equation~}\eqref{eqn:C-homomorphism} is satisfied.
\end{proof}

\begin{theorem}
\label{thm:compactness:Clifford}
$\CL$ is compact.
\end{theorem}
\begin{proof}
Let $p$ be a polynomial such that the complexity of applying $\HE.\Dec$ to the output of $\HE.\Eval$ is at most $p(\kappa)$ --- such a polynomial exists by the compactness of $\HE$. Then decrypting a single qubit of the output of $\CL.\Eval$ has complexity at most $2p(\kappa)+2$, since we must decrypt two keys $a$ and $b$ and then apply $\xgate^a$ and $\zgate^b$, so $\CL$ is also compact.
\end{proof}

\begin{theorem}
\label{thm:security:Clifford}
\iftoggle{crypto}{If $\HE$ is q-IND-CPA then $\CL$ is q-IND-CPA.}{Assuming a classical fully homomorphic encryption scheme $\HE$ that is q-IND-CPA secure, the quantum homomorphic scheme $\CL$  is q-IND-CPA secure.}
\end{theorem}
\begin{proof}
The main part of this proof will be to show that the classical ciphertexts $\HE.\Enc_{pk}(a)$ and $\HE.\Enc_{pk}(b)$ give at most a negligible advantage. We will then see that without these classical ciphertexts, the quantum CPA Indistinguishability experiment is independent of $r$ from the  perspective of the adversary.

Let $\CL'$ be the quantum homomorphic encryption scheme with $\CL'.\KeyGen=\CL.\KeyGen$, $\CL'.\Eval=\CL.\Eval$, $\CL'.\Dec=\CL.\Dec$, and
\begin{align*}
\CL'.\Enc_{pk}(\rho) &=\sum_{a,b\in\{0,1\}}\frac{1}{4}\rho(\HE.\Enc_{pk}(0),\HE.\Enc_{pk}(0))\otimes (\xgate^a\zgate^b\rho\zgate^b\xgate^a)\\
&=\rho(\HE.\Enc_{pk}(0),\HE.\Enc_{pk}(0))\otimes \frac{1}{2}\mathbb{I}_2.
\end{align*}

Let $\advA=(\advA_1,\advA_2)$ be an adversary for $\mathsf{PubK}_{\advA,\CL}^{\textsf{cpa}}(\kappa)$.
We will define an adversary $\advA'=(\advA_1',\advA_2')$ for $\mathsf{PubK_{\advA',\HE}^{cpa\text{-}mult}}(\kappa)$. Essentially, $\advA'$ will simulate $\mathsf{PubK}_{\advA,\CL}^{\textsf{cpa-mult}}(\kappa)$, except that when it simulates $\Xi_{\CL}^{{\sf cpa},r}$, it will use $\Xi_{\HE}^{\textsf{cpa-mult},s}$ in place of $\HE.\Enc$, so that it will actually be running either $\Xi_{\CL}^{{\sf cpa},r}$ (if $s=1$) or $\Xi_{\CL'}^{{\sf cpa},r}$ (if $s=0$) (see Fig.~\ref{fig:cl-security}).

\begin{figure}[H]
\centering
\input{fig-cl-security.tex}
\end{figure}

\begin{description}
\item[$\advA_1'(pk,evk)$:] Run $\advA_1(pk,evk)$ to get a state $\rho^{\cal ME}$. Choose a uniform random bit $r$. If $r=0$, discard the $\mathcal{M}$ subsystem and replace it with the state $\ket{0}\bra{0}$. Choose uniform random bits $a$ and $b$, and apply $\QEnc_{a,b}$, the quantum one-time pad, to $\mathcal{M}$, relabelling the resulting system by $\cal X$. Input $(a,b)$ and $(0,0)$ to $\Xi_{\HE}^{\textsf{cpa-mult},s}$.
\item[$\advA_2'$:] Run $\advA_2$ to get a bit $r'$. Output $1$ if $r=r'$ and $0$ otherwise.
\end{description}
We now compute the probability that $\advA'$ correctly guesses $s$, which we know must be at most $\frac{1}{2}+\eta(\kappa)$ for some negligible function, since $\HE$ is q-IND-CPA. If $s=1$, then $\advA'$ is simulating $\mathsf{PubK}_{\advA,\CL}^{\mathsf{cpa}}$, so the probability that $r'=r$ (and thus that $s'=1=s$) is $\Pr[\mathsf{PubK_{\advA,\CL}^{cpa}}(\kappa)=1]$.

On the other hand, if $s=0$, $\advA_2$ gets encryptions of $0$ rather than $\HE.\Enc(a),\HE.\Enc(b)$,~so $\advA'$ is simulating $\mathsf{PubK}_{\advA,\CL'}^{\mathsf{cpa}}$, so the probability that $r\neq r'$, and thus $s'=0=s$, is \mbox{$\Pr[\mathsf{PubK_{\advA,\CL'}^{cpa}}(\kappa)=0]$}.\looseness = -1

Then since the total probability that $s=s'$ is at most $\frac{1}{2}+\eta(\kappa)$, we have:
\begin{align}
\frac{1}{2}\Pr[\mathsf{PubK_{\advA,\CL}^{cpa}}(\kappa)=1]+\frac{1}{2}\Pr[\mathsf{PubK_{\advA,\CL'}^{cpa}}(\kappa)=0] &\leq \frac{1}{2}+\eta(\kappa)\nonumber\\
\Pr[\mathsf{PubK_{\advA,\CL}^{cpa}}(\kappa)=1]+1-\Pr[\mathsf{PubK_{\advA,\CL'}^{cpa}}(\kappa)=1] &\leq 1+2\eta(\kappa)\nonumber\\
\Pr[\mathsf{PubK_{\advA,\CL}^{cpa}}(\kappa)=1]-\Pr[\mathsf{PubK_{\advA,\CL'}^{cpa}}(\kappa)=1] &\leq 2\eta(\kappa).\label{eq:negl}
\end{align}

We complete the proof by noting that when $s=0$, since $c=(\HE.\Enc_{pk}(0),\HE.\Enc_{pk}(0))$, it is independent of $a,b$ (see
 Fig.~\ref{fig:cl-security-s0}).
\begin{figure}[H]
\centering
\input{fig-cl-security-s0.tex}
\end{figure}

\noindent Then from the perspective of $\advA_2$, since $a,b$ is uniform random, the system $\cal X$ just contains the completely mixed state $\maxmix$ (see Fig.~\ref{fig:cl-security-last}).

\begin{figure}[H]
\centering
\input{fig-cl-security-last.tex}
\end{figure}

\noindent Since the experiment $\mathsf{PubK_{\advA,\CL'}^{cpa}}$ is independent of $r$ from the perspective of $\advA$, it follows that $\Pr[\mathsf{PubK_{advA,\CL'}^{cpa}}(\kappa)=1]=\frac{1}{2}$. Combining this with Equation \eqref{eq:negl}, we get
\begin{equation*}
\Pr[\mathsf{PubK_{\advA,\CL}^{cpa}}(\kappa)=1]\leq \frac{1}{2}+2\eta(\kappa),
\end{equation*}
which completes the proof, since $2\eta$ is still a negligible function.
\end{proof}

}

\vspace{-5pt}
\iftoggle{crypto}{\section{$\tgate$-gate Computation Using Entanglement: $\EPR$}}
{\section{Scheme $\EPR$: $\tgate$-gate Computation Using Entanglement}}
\label{sec:scheme-EPR}
\vspace{-5pt}
In order to achieve universality for quantum circuits, we need to add a non-Clifford group gate, such as the $\tgate$-gate. As noted in Sec.~\ref{sec:intro-summary}, if we apply the same technique as in Sec.~\ref{sec:scheme-Clifford} (\emph{i.e.} to apply the $\tgate$-gate on the encrypted quantum data)  we run into a problem,
\iftoggle{crypto}
{ since $\tgate \xgate^a \zgate^b = \xgate^a \zgate^{a \oplus b} \pgate^a \tgate$.}
{ since: \begin{equation}
\label{eqn:T-gate-P-error}
\tgate \xgate^a \zgate^b = \xgate^a \zgate^{a \oplus b} \pgate^a \tgate.
\end{equation}}
That is, conditioned on $a$, the output picks up an undesirable
$\pgate$ error, which cannot be corrected by applying Pauli
corrections. In~\cite{Chi05}, Childs arrives at the same conclusion,
and makes the observation that, in the case where $a=1$,
the evaluation algorithm could be made to \emph{correct} this erroneous
$\pgate$-gate\iftoggle{crypto}{}{ by executing a correction (which consists of
$\zgate\pgate$)}. As long as the evaluation algorithm does not find out if this
correction is being executed or not, security holds. The solution in~\cite{Chi05} involves quantum interaction; this was recently improved to a single auxiliary qubit, coupled with classical interaction~\cite{QCED,B15}. \iftoggle{crypto}{}{In this section, we base the evaluation of the $\tgate$-gate on a modification of this technique, as presented in Fig.~\ref{fig:T-gate-EPR}. The modification is that we allow the auxiliary qubit to be prepared in a state dependent on the $\xgate$-encryption key, whereas~\cite{QCED,B15} explicitly avoids  this since it requires the auxiliary qubits to be prepared independently of the computation. Correctness of Fig.~\ref{fig:T-gate-EPR} is proven in App.~\ref{appendix:Correctness-T-gate}.

}As a proof technique (for establishing security), \cite{QCED, B15} considers an equivalent, entanglement-based protocol. Here, we use the idea of exploiting entanglement in order to \emph{delay} the correction required for the evaluation of the $\tgate$-gate on encrypted data.
The protocol is illustrated in Fig.~\ref{fig:T-gadget-EPR}. \iftoggle{crypto}{Correctness of Fig.~\ref{fig:T-gadget-EPR} is proven in the full version \cite{BJ15}.}{}

\iftoggle{crypto}{}{
\begin{figure}
\centering
\begin{tikzpicture}

\draw[dashed] (0,0)--(8,0);
\draw (2.5,1.5)--(7,1.5);
\draw (5,.5)--(8,.5);
\draw (5,.5)--(5,-1);
\draw (2.5,-1)--(5,-1);
\draw (7,1.54)--(7.5,1.54);
\draw (7,1.46)--(7.5,1.46);
\draw (7.54,1.5)--(7.54,-.5);
\draw (7.46,1.5)--(7.46,-.5);
\draw (7.5,-.46)--(8,-.46);
\draw (7.5,-.54)--(8,-.54);

\filldraw[fill=white] (4,1.25) rectangle (4.5,1.75);
\node at (4.25,1.5) {$\tgate$};

\node at (1.6,1.5) {$\xgate^a\zgate^b\ket{\psi}$};

\filldraw[fill=white] (3,-1.25) rectangle (3.5,-.75);
\node at (3.25,-1) {$\pgate^a$};
\filldraw[fill=white] (4,-1.25) rectangle (4.5,-.75);
\node at (4.25,-1) {$\zgate^k$};

\node at (2,-1) {$\ket{+}$};

\node at (6.75,1.5) {\meas};
\node at (7.5,1.5) {\cntrl};
\node at (7.5,-.5) {\cntrl};

\node at (5.5,.5) {\cntrl};
\node at (5.5,1.5) {\target};
\draw (5.5,.5)--(5.5,1.5);

\node at (6.5,-1) {$(k\in_R\{0,1\})$};
\node at (8.4,-.5) {$c$};
\node at (10,.5) {$\xgate^{a\oplus c}\zgate^{a\oplus b\oplus k\oplus a\cdot c}\tgate\ket{\psi}$};

\end{tikzpicture}

 \caption{\label{fig:T-gate-EPR}Functionality of the $\tgate$-gate gadget.}
\end{figure}
}

\begin{figure}
\centering
\begin{tikzpicture}
\draw (0,5)--(2,5);
\draw (2,5.04)--(3,5.04);
\draw (2.5,4.96)--(3,4.96);
\node at (-1,5) {$\xgate^{f_{a,i}}\zgate^{f_{b,i}}\ket{\psi}$};
\node at (.25,5.2) {\small $\mathcal{X}_i$};

\filldraw[fill=white] (.5,5.25) rectangle (1,4.75);
\node at (.75,5) {$\tgate$};

\node at (1.5,5) {\target};
\node at (1.5,4.5) {\cntrl};
\draw (1.5,5)--(1.5,4.5);

\node at (2.25,5) {\meas};

\node at (3.2,5) {$c$};

\draw (.25,4.5)--(3,4.5);
\node at (5.25,4.5) {$\xgate^{f_{a,i}\oplus c}\zgate^{f_{a,i}\oplus f_{b,i}\oplus k_t\oplus cf_{a,i}}\tgate\ket{\psi}$};
\node at (2.75,4.3) {\small $\mathcal{X}_i$};
\node at (-.5,4.25) {$\ket{\Phi^+}$};
\draw (0,4.25)--(.25,4.5);
\draw (0,4.25)--(.25,4);
\draw (.25,4)--(3.5,4);
\node at (2.75,3.8) {{\small $\mathcal{R}_t$}};

\iftoggle{crypto}{
\node at (6.2,4.4) {\parbox{3in}{\raggedleft $f_t\leftarrow f_{a,i}$\\[5pt]
$V\leftarrow V\cup\{k_t\}$\\[5pt]
$f_{a,i}\leftarrow f_{a,i}\oplus c$\\[5pt]
$f_{b,i}\leftarrow (1\oplus c)f_{a,i}\oplus f_{b,i} \oplus k_t$}};
}{
\node at (9.5,3.75) {\parbox{3in}{\raggedleft $f_t\leftarrow f_{a,i}$\\[5pt]
$V\leftarrow V\cup\{k_t\}$\\[5pt]
$f_{a,i}\leftarrow f_{a,i}\oplus c$\\[5pt]
$f_{b,i}\leftarrow (1\oplus c)f_{a,i}\oplus f_{b,i} \oplus k_t$}};
}

\draw (3.5,4)--(3.5,3)--(6,3);
\filldraw[fill=white] (3.9,3.25) rectangle (4.6,2.75);
\node at (4.25,3) {$\pgate^{f_t}$};

\filldraw[fill=white] (5,3.25) rectangle (5.5,2.75);
\node at (5.25,3) {$\hgate$};

\node at (6.25,3) {\meas};

\draw (6.6,3.04)-- (7.25,3.04);
\draw (6.6,2.96)-- (7.25,2.96);
\node at (7.5,3) {$k_t$};

\draw[dashed] (3.25,3.35) rectangle (7.9,2.25);
\node at (5.4,2.5) {\small (Part of decryption)};

\end{tikzpicture}

  \caption{\label{fig:T-gadget-EPR}Evaluation protocol for the \th{t} $\tgate$-gate, \iftoggle{crypto}{on}{applied to} the \th{i} wire. The key-polynomials $f_{a,i}$ and $f_{b,i}$ are in $\mathbb{F}_2[V]$. After the protocol, $V$ gains a new variable corresponding to the unknown measurement result $k_t$. The dashed box shows part of the decryption\iftoggle{crypto}{}{ procedure}, which happens at some point in the future, after the complete evaluation is finished.}
\end{figure}

Fig.~\ref{fig:T-gadget-EPR} shows that, using the \iftoggle{crypto}{}{entangled }state~$\ket{\Phi^+} = \frac{1}{\sqrt{2}}(\ket{00} + \ket{11})$, the conditional~$\pgate$ correction can be delayed. The cost of this is that the value of the measurement result,~$k_t$, on auxiliary register $\regR_t$, is undetermined until later, when it is  measured as part of the decryption\iftoggle{crypto}{}{ algorithm}. Thus we view the key updates as a symbolic computation: each time a $\tgate$-gate is applied, an extra variable, $k_t$, is introduced.

For the first $\tgate$-gate evaluation ($t=1$), the evaluation procedure does not have the knowledge to evaluate  $f_1=f_{a,i}$, where $i$ is the wire upon which the gate is performed, in order to perform the correction. It is possible (using the classical scheme $\HE$), to compute a classical ciphertext $\widetilde{f_1}$ that decrypts to $f_1(a_1,b_1,\dots,a_n,b_n)$. Thus, for this $\tgate$-gate, the output  part of the auxiliary system contains both $\widetilde{f_1}$ and the register~$\regR_1$. As part of the decryption operation,  compute $f_1 \leftarrow \HE.\Dec(\widetilde{f_1})$, and apply $\pgate^{f_1}$ on $\regR_1$ before measuring in the Hadamard basis and obtaining~$k_1$. From the point of view of the evaluation procedure,~$k_1$ is unknown and so it becomes an \emph{unknown} part of the encryption key (in contrast with the previous keys, which are also ``unknown'', but to a lesser degree, since we have access to the classical encrypted values of these keys). The algorithm $\Eval$ continues in this fashion for values of $t$ up to~$R$; each time, the set of unknown variables increasing by~one.
Note that, according to Fig.~\ref{fig:T-gadget-EPR}, as well as the linearity of the key update rules,  for all $t$, $f_t\in \mathbb{F}_2[a_1,\dots,a_n,b_1,\dots,b_n,k_1,\dots,k_{t-1}]$ is linear (since $c$ is a known constant), so we can write $f_t=f_t^k+f_t^{ab}$ for $f_t^k\in \mathbb{F}_2[k_1,\dots,k_{t-1}]$ and $f_t^{ab}\in \mathbb{F}_2[a_1,\dots,a_n,b_1,\dots,b_n]$.

The cost of this construction is that each $\tgate$-gate adds to the complexity of the decryption procedure, since, in particular, for each $\tgate$-gate, we must perform a possible $\pgate$-correction and a measurement on an auxiliary qubit. In addition, we cannot evaluate the key-polynomials, nor the $f_t$, until the variables $k_t$ have been measured, so this evaluation must take place in the decryption phase, increasing the dependence on $R$, the number of $\tgate$-gates, to $O(R^2)$\iftoggle{crypto}{ (see full version \cite{BJ15}).}{} \iftoggle{crypto}{}{We make this dependence precise in Thm.~\ref{cor:EPR-dec-complexity}.}\looseness=-1

We now formally define the indivisible \iftoggle{crypto}{QHE}{quantum homomorphic encryption} scheme,~$\EPR$. As in $\CL$, we have message space ${\cal M}=\mathbb{C}^{\{0,1\}}$ and cipherspace ${\cal C}=\mathbb{C}^{C\times C}\otimes\cal X$, where $C$ is the output space of $\HE.\Enc$ and ${\cal X}\equiv \mathbb{C}^{\{0,1\}}$. Since $\EPR$ is indivisible, the output space of $\EPR.\Eval^{\mathsf{C}}$ has the form $\mathcal{R}_{aux}\otimes {\cal C'}^{\otimes m}$. \iftoggle{crypto}{}{We require an indivisible scheme, because decryption of any one of the output qubits requires access to the auxiliary system. }In our case, we have $\mathcal{R}_{aux}=\mathcal{R}_1\otimes\dots\otimes \mathcal{R}_R\otimes(\mathbb{C}^{\{0,1\}^{R+1}})^{\otimes R}\otimes (\mathbb{C}^{C'})^{\otimes R}$, where $R$ is the number of $\tgate$-gates\iftoggle{crypto}{}{ in $\mathsf{C}$}, $C'$ is the output space of $\HE.\Eval$, and $\mathcal{R}_t\equiv \mathbb{C}^{\{0,1\}}$\iftoggle{crypto}{}{ for each $t$}. The classical parts of the auxiliary space allow us to output $R$ linear polynomials in $\mathbb{F}_2[k_1,\dots, k_R]$ corresponding to $\{f_t^k\}_{t=1}^R$, each of which can be represented with $R+1$ bits; as well as $R$ $\HE.\Eval$ outputs, corresponding to encryptions of $\{f_t^{ab}(a_1,\dots,a_n,b_1,\dots,b_n)\}_{t=1}^R$.
Similarly, we have $\mathcal{C}'=(\mathbb{C}^{\{0,1\}^{R+1}})^{\otimes 2}\otimes\mathbb{C}^{C'\times C'} \otimes\cal X$.

\iftoggle{crypto}{
The key generation, $\EPR.\KeyGen$, and encryption, $\EPR.\Enc$, are defined exactly as $\CL.\KeyGen$ and $\CL.\Enc$. We now define $\EPR.\Eval$ and $\EPR.\Dec$.
}{
\noindent\textbf{Key Generation.} $\EPR.\KeyGen(1^\kappa)$. The key generation procedure is the same as $\mathsf{CL.KeyGen}(1^\kappa)$.

\noindent\textbf{Encryption.} $\EPR.\Enc_{pk}:D({\mathcal{M}})\rightarrow D({\cal C})$. The encryption procedure is the same as $\mathsf{CL.Enc}_{pk}$.
}

\noindent\textbf{Evaluation.} $\EPR.\Eval_{evk}$. As in $\mathsf{CL}$, apply gates in $\{\xgate, \zgate, \pgate, \hgate, \cnot\}$ directly on the encrypted quantum registers.
For the $\tgate$-gate, use the gadget defined in Fig.~\ref{fig:T-gadget-EPR}. This gadget differs from previous gadgets in that it uses an auxiliary Bell state,~$\ket{\Phi^+}$. After the system of the \th{i} wire, $\mathcal{X}_i$, is measured,  relabel half of the Bell state as $\mathcal{X}_i$, and the other half as~$\regR_t$, which is returned as part of~$\mathcal{R}_{aux}$.
The full evaluation procedure is as follows.

\begin{enumerate}
\item Set $V\leftarrow \{a_i,b_i\}_{i\in [n]}$, and
$\forall\,i\in [n]$, \iftoggle{crypto}{}{set $f_{a,i},f_{b,i}\in\mathbb{F}_2[V]$ as} $f_{a,i}\leftarrow a_i$, $f_{b,i}\leftarrow b_i$.
\item Let $\mathsf{g}_1,\dots,\mathsf{g}_G$ be a topological ordering of the gates in~$\mathsf{C}$. For $j=1,\dots,G$, evaluate $\mathsf{g}_j$ using the appropriate gadget.
\item Let $S$ be the set of output wires.
     Let $\mathcal{L}$ be the set of labels $\mathcal{L}=\{(a,i),(b,i):i\in S\}\cup\{1,\dots, R\}$. For each $\alpha\in \mathcal{L}$, we want to homomorphically evaluate $f_\alpha$ to  obtain the actual (encrypted) key, but we can only actually evaluate the part of $f_\alpha$ that is in the variables $\{a_i,b_i\}_i$ --- the $\{k_t\}_t$ are still unknown.  Recall that we can write  $f_\alpha=f_\alpha^k+f_\alpha^{ab}$ for $f_\alpha^k\in \mathbb{F}[k_1,\dots,k_R]$ and $f_\alpha^{ab}\in \mathbb{F}_2[a_1,\dots,a_n,b_1,\dots,b_n]$.
      Compute
 $\widetilde{f_\alpha^{ab}}\leftarrow\mathsf{HE.Eval}_{evk}^{f_\alpha^{ab}}(\tilde{a}_1,\dots,\tilde{a}_n,\tilde{b}_1,\dots,\tilde{b}_n)$.
\item Output: \iftoggle{crypto}{the $m=|S|$ qubit registers $\{\mathcal{X}_i:i\in S\}$ corresponding to the encrypted output of the circuit;
the $R$ qubit registers $\mathcal{R}_1,\dots,\mathcal{R}_R$ corresponding to auxiliary states created by $\tgate$-gadgets;
the polynomials $\{f_\alpha^k\}_{\alpha\in \mathcal{L}}\subset\mathbb{F}_2[k_1,\dots,k_R]$ and the homomorphically evaluated polynomials $\{\widetilde{f_\alpha^{ab}}\}_{\alpha\in \mathcal{L}}$.}{\begin{itemize}
\item The $m=|S|$ qubit registers $\{\mathcal{X}_i:i\in S\}$ corresponding to the encrypted output of the circuit;
\item The $R$ qubit registers $\mathcal{R}_1,\dots,\mathcal{R}_R$ corresponding to auxiliary states created by $\tgate$-gadgets;
\item The polynomials $\{f_\alpha^k\}_{\alpha\in \mathcal{L}}\subset\mathbb{F}_2[k_1,\dots,k_R]$ and the homomorphically evaluated polynomials $\{\widetilde{f_\alpha^{ab}}\}_{\alpha\in \mathcal{L}}$.
\end{itemize}}
\end{enumerate}

\noindent\textbf{Decryption.} $\EPR.\Dec_{sk}$. In order to decrypt,
measure the $\mathcal{R}_t$ in order from 1 to~$R$, computing $f_t(k_1,\dots,k_{t-1})$ as required.
Formally:
\begin{enumerate}
\item For $t=1,\dots, R$:
  \begin{enumerate}
  \item Decrypt $f^{ab}_t \leftarrow \mathsf{HE.Dec}_{sk}(\widetilde{f^{ab}_t})$.
  \item Compute $a \leftarrow f_t^k(k_1,\dots,k_{t-1})\oplus f^{ab}_t$ and apply $\hgate\pgate^a$ to $\mathcal{R}_t$.
  \item Measure $\mathcal{R}_t$ to get $k_t$.
  \end{enumerate}
\item Let $S$ be the set of indices of the output qubit registers. For $i\in S$:
  \begin{enumerate}
  \item Decrypt $f^{ab}_{a,i} \leftarrow \mathsf{HE.Dec}_{sk}(\widetilde{f^{ab}_{a,i}})$ and $f^{ab}_{b,i} \leftarrow \mathsf{HE.Dec}_{sk}(\widetilde{f^{ab}_{b,i}})$.
  \item Compute $a_i \leftarrow  f^k_{a,i}(k_1,\dots,k_t)\oplus f^{ab}_{a,i}$ and $b_i \leftarrow f^k_{b,i}(k_1,\dots,k_t)\oplus f^{ab}_{b,i}$.
  \end{enumerate}
\item  To each register $\mathcal{X}_i$, apply the map $\QDec_{a_i, b_i}$.
Output registers $\mathcal{X}_1, \ldots, \mathcal{X}_{m}$.
\end{enumerate}

\iftoggle{crypto}{
We prove that $\EPR$ is homomorphic for all quantum circuits in the universal gate set $\{\xgate,\zgate,\pgate,\cnot,\hgate,\tgate\}$, $R^2$-quasi-compact, and q-IND-CPA, in \cite{BJ15}.
}{
\subsection{Analysis of $\EPR$}

\iftoggle{crypto}{}{We now analyse the various properties of $\EPR$.}
Since the scheme $\EPR$ uses the same $\KeyGen$ and $\Enc$ procedures as~$\CL$, the following theorem follows from Thm.~\ref{thm:security:Clifford}.

\begin{theorem}
\label{thm:security:QHE}
If $\HE$ is q-IND-CPA secure, then $\EPR$ is q-IND-CPA secure.
\end{theorem}

The next theorem shows the  homomorphic property for all circuits (recall that this property is independent of compactness).

\begin{theorem}
\label{thm:correctness:QHE}
Let $\classS$ be the class of all quantum circuits. Then $\EPR$ is  $\classS$-homomorphic.
\end{theorem}
The proof follows from the circuits in App.~\ref{appendix:Key-update-rules-stabilizer}, Fig.~\ref{fig:T-gadget-EPR}, as well as the homomorphic property of~$\mathsf{HE}$.

Since the complexity of the decryption procedure depends on $R$, the number of $\tgate$-gates in the circuit, it is clear that the scheme $\EPR$ is not compact. However, by analysing the circuit's dependence on $R$, we can see that for a very large class of quantum circuits, $\EPR$ is non-trivially quasi-compact. The following theorem is immediate from the decryption procedure.

\begin{theorem}\label{cor:EPR-dec-complexity}
Let $p$ be a polynomial such that $\HE.\Dec$ has complexity $O(p(\kappa))$. Then \iftoggle{crypto}{}{the decryption procedure} $\EPR.\Dec$ has complexity $O(R^2+Rp(\kappa)+mp(\kappa)+mR)$.
\end{theorem}

\noindent Thus, the dependence of the complexity of $\EPR.\Dec$ on the evaluated circuit $\mathsf{C}$ is $R^2$:

\begin{corollary}
Let $R(\mathsf{C})$ denote the number of $\tgate$-gates in a circuit $\mathsf{C}$. Then $\EPR$ is $R^2$-quasi-compact.
\end{corollary}

This beats the compactness of the trivial scheme for all circuits $\mathsf{C}$ such that the number of $\tgate$-gates is less than the squareroot of the number of gates; that is $R\ll \sqrt{G}$.

}

\vspace{-5pt}
\iftoggle{crypto}{
\section{$\tgate$-gate Computation Using Auxiliary States: $\AUX$}
}{
\section{Scheme $\AUX$: $\tgate$-gate Computation Using Auxiliary States}
}
\label{sec:scheme-AUX}
\vspace{-5pt}
In the previous QHE scheme, we solved the problem of performing the $\pgate$ correction \iftoggle{crypto}{}{(Eq.~\eqref{eqn:T-gate-P-error})} by \emph{delaying} the correction via entanglement.
In this section, we present a quantum homomorphic encryption scheme, $\AUX$, that takes a more proactive approach to dealing with the $\pgate$ correction. At a high level, $\AUX$ can be understood as the following:  as part of the evaluation key, $\textsf{AUX.Keygen}$ outputs a number of auxiliary states. These states ``encode'' parts of the original encryption key, and are used to correct for the errors induced by the straightforward application of the $\tgate$-gate on the cipherstates. In more details, the auxiliary states encode hidden versions of $\pgate$ corrections, such as $\ket{+_{a,k}}:=\zgate^{k}\pgate^{a} \ket{+}$ (where $k$ is a random bit and~$a$ is an encryption key) that are useful for the evaluation of the $\tgate$-gate (see Fig.~\ref{fig:T-gate-AUX}). In general (after having applied prior gates), the exact auxiliary state will not be available; instead, the $\Eval$ procedure combines a number of auxiliary states in order to create a single copy of a state that is useful for performing the correction. This combination operation, however, is expensive as it introduces new unknowns (in terms of new variables as well as ``cross-terms''), that need to be corrected in any future $\tgate$-gate. Thus the size of the evaluation key grows rapidly, as  a polynomial whose degree is exponential in the $\tgate$-depth.  We can thus tolerate only a \emph{constant} $\tgate$-gate depth for this scheme to be efficient.
\begin{figure}[t]
\centering
\begin{tikzpicture}
\draw (0.5,1)--(3,1);
\node at (.75,.75) {\small $\mathcal{X}_i$};
\draw (3,1.04) -- (3.75,1.04);
\draw (3,.96) -- (3.75,.96);
\node at (-.3,1) {$\xgate^{f_{a,i}}\zgate^{f_{b,i}}$};
\filldraw[fill=white] (1,1.25) rectangle (1.5,.75);
\node at (1.25,1) {$\tgate$};
\node at (2.075,1) {\target};
\draw (2.075,1) -- (2.075,0);
\node at (2.075,0) {\cntrl};
\node at (3,1) {\meas};
\node at (4,1) {$c$};

\draw (0.5,0)--(3.75,0);
\node at (-.2,0) {$\ket{+_{f_{a,i},k}}$};
\node at (3.5,-.25) {\small $\mathcal{X}_i$};
\node at (5.75,0) {$\xgate^{f_{a,i}\oplus c}\zgate^{f_{a,i}\oplus f_{b,i}\oplus k\oplus cf_{a,i}}\tgate$};

\node at (8.1,.5) {\parbox{2in}{\raggedleft $f_{a,i}\leftarrow f_{a,i}\oplus c$\\[5pt]
$f_{b,i}\leftarrow f_{a,i}\oplus f_{b,i}\oplus k\oplus cf_{a,i}$\\[5pt]
$V\leftarrow V\cup \mathrm{var}(k)$}};
\end{tikzpicture}
\caption{A $\tgate$-gadget for the scheme $\AUX$ consists of the above circuit and key-update rules. We use $\mathrm{var}(k)$ to denote the set of variables in the polynomial $k$, which depends on the construction of the auxiliary state $\ket{+_{f_{a,i},k}}$, described below.}\label{fig:T-gate-AUX}
\end{figure}

We further specify  that $\AUX$ is a symmetric-key encryption scheme. This is because $\AUX.\KeyGen$  generates auxiliary qubits that depend on the quantum one-time pad encryption keys.  Also, $\textsf{KeyGen}$ takes an extra parameter $1^n$, where $n$ is  an upper bound on the total number of qubits that can be encrypted ($\AUX$ acts much like a classical one-time pad scheme that picks a fixed-length encryption key ahead of time).  After this bound on the number of encryptions has been attained, no further qubits can be encrypted. We will suppose without loss of generality that a circuit being homomorphically evaluated is on $n$ wires.
Furthermore, the number and type of auxiliary qubits will depend on the $\tgate$-depth of the circuit to be evaluated, $L$. The scheme will not be able to homomorphically evaluate circuits with $\tgate$-depth greater than $L$. \iftoggle{crypto}{}{We will see that the number of required auxiliary states grows super-exponentially in $L$, so we will require that $L$ be a constant. }Fix a constant $L$. We will now define a scheme $\AUX=\AUX_L$ that is homomorphic for all circuits with $\tgate$-depth at most $L$.

\iftoggle{crypto}{}{\noindent\textbf{Auxiliary Qubit Construction.}} \iftoggle{crypto}{P}{In general, p}roviding the necessary auxiliary states for each $\tgate$-gate  would  require advance knowledge of  the key $f_{a,i}$ at the time a $\tgate$-gate is applied to the \th{i} wire. Since this depends on both the circuit being applied and \iftoggle{crypto}{}{on the }prior measurement results, we appear to be at an impasse. The key observation that allows us to continue with this approach is that, given auxiliary states $\ket{+_{f_1,k_1}}$ and $\ket{+_{f_2,k_2}}$, we can combine them to get $\ket{+_{f_1\oplus f_2,k}}$, for some $k$, using the following circuit:
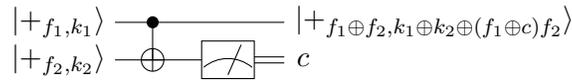
\vspace{-10pt}
{\begin{figure}[H]}
\centering
\begin{tikzpicture}
\node at (-.25,.5) {$\ket{+_{f_1,k_1}}$};
\node at (-.25,0) {$\ket{+_{f_2,k_2}}$};
\draw (.5,.5)--(2.75,.5);
\draw (.5,0)--(2,0);
\draw (2,0.04)--(2.75,0.04);
\draw (2,-0.04)--(2.75,-0.04);
\node at (4.75,.5) {$\ket{+_{f_1\oplus f_2,k_1\oplus k_2\oplus (f_1\oplus c)f_2}}$};
\node at (1,.5) {\cntrl};
\node at (1,0) {\target};
\draw (1,.5)--(1,0);
\node at (2,0) {\meas};
\node at (3,0) {$c$};
\end{tikzpicture}
{\end{figure}}
\vspace{-10pt}
\noindent By iterating this procedure, given auxiliary states $\ket{+_{f_1,k_1}},\dots,\ket{+_{f_r,k_r}}$, we can construct\iftoggle{crypto}{}{ the auxiliary state} $\ket{+_{f_1\oplus\dots\oplus f_r,k}}$, where $k=\bigoplus_{i=1}^m k_i\oplus \bigoplus_{i=2}^r c_if_i\oplus \bigoplus_{i=1}^r\bigoplus_{j=1}^{i-1}f_if_j$ for known values $c_i$. Thus, if we give many initial auxiliary states of the form $\{\ket{+_{a_i,k_{a,i}}},\ket{+_{b_i,k_{b,i}}}\}_i$ (with different keys for different copies), we can construct $\ket{+_{f,k}}$ for $f$ a linear function of $\{a_i,b_i\}_{i\in [n]}$. However, using an auxiliary state $\ket{+_{f_{a,i},k}}$ to facilitate a $\tgate$-gate on the \th{i} wire introduces the unknown $k$ into $f_{b,i}$. In particular, suppose $f_{a,i}=\bigoplus_{j=1}^r t_j$ for some monomial terms $t_j\in\mathbb{F}_2[V]$. Then we will need to construct it from auxiliary states $\ket{+_{t_1,k_1}},\dots,\ket{+_{t_r,k_r}}$, to get $\ket{+_{f_{a,i},k}}$ for $k=\bigoplus_{i=1}^m k_i\oplus \bigoplus_{i=2}^r c_it_i\oplus \bigoplus_{i=1}^r\bigoplus_{j=1}^{i-1}t_it_j$. Thus, after the $\tgate$-gadget, the new keys $f_{a,i}',f_{b,i}'$ are in  unknowns $V\cup\{k_1,\dots,k_r\}$. Furthermore, because of the cross terms $t_it_j$, the degree of the key-polynomials increases, so we can no longer assume they are linear. Since we can't produce $\ket{+_{f_1f_2,k}}$ from $\ket{+_{f_1,k_1}}$ and $\ket{+_{f_2,k_2}}$, we need to provide additional auxiliary states for every possible term. We discuss this more formally below and in \iftoggle{crypto}{the full version \cite{BJ15}.}{Sec.~\ref{sec:aux-analysis}.}

\iftoggle{crypto}{}{\paragraph{Spaces}} As in $\CL$ and $\EPR$, we work with qubits: $\mathcal{M}\equiv \mathbb{C}^{\{0,1\}}$. In contrast to our previous schemes, the classical encryptions of quantum one-time pad keys is part of the evaluation key (for convenience only), so we  have $\mathcal{C}\equiv \mathbb{C}^{\{0,1\}}$. However, after evaluation, the classical encryption of the new one-time pad keys is needed for decryption, so as in $\CL$, we  have $\mathcal{C}'\equiv \mathbb{C}^{C'\times C'}\otimes\mathcal{X}$, where $C'$ is the output space of $\HE.\Eval$, and $\mathcal{X}\equiv\mathbb{C}^{\{0,1\}}$.

\noindent\textbf{Key Generation.} $\mathsf{\AUX.Keygen}(1^\kappa, 1^n)$.
The evaluation key  contains auxiliary states that allow each of $L$ layers of $\tgate$-gates to be implemented. Thus, for each layer, since every wire must have the possibility to implement a $\tgate$-gate, for each wire, we need to be able to construct an auxiliary state $\ket{+_{f_{a,i},k}}$ for some $k$. Since we can add auxiliary states, we can construct this auxiliary state if we have an auxiliary state for each term in $f_{a,i}$. Since $f_{a,i}$ depends on the circuit, which we do not know in advance, we need to provide an auxiliary state for every term that could possibly be in $f_{a,i}$ at the \th{\ell} layer of $\tgate$-gates, for $\ell=1,\dots,L$.

We now define sets of monomials $T_1,\dots,T_L$ such that the keys in the \th{\ell} layer consist of sums of terms from $T_\ell$\iftoggle{crypto}{}{ (as proven in Lemma \ref{lem:terms})}. Let $V_1:={\{a_i,b_i\}_{i\in[n]}}$, and define $T_1\subset \mathbb{F}_2[V_1]$ by \iftoggle{crypto}{$T_1:=\{a_i,b_i\}_{i\in[n]}$.}{
$$T_1:=\{a_1,\dots,a_n,b_1,\dots,b_n\}.$$}
The monomials in $T_1$ represent the possible terms in the key-polynomials before the first layer of $\tgate$-gates. Each of the up to $n$ $\tgate$-gates in the first layer requires a copy of each of $\{\ket{+_{t,k^{(1)}_t}}\}_{t\in T_1}$, with independent random keys for each, for a total of $n|T_1|$ auxiliary states. More generally, for the \th{\ell} layer of $\tgate$-gates, we let $T_{\ell}$ be the set of possible terms in the key-polynomials before applying the \th{\ell} layer of $\tgate$-gates. We can see from the $\tgate$-gadget, as well as the construction for adding auxiliary states that the keys from the previous layer's auxiliary states, $\{k^{(\ell-1)}_{1,i},\dots,k^{(\ell-1)}_{|T_{\ell-1}|,i}\}_{i=1}^n$, may now be variables in the key-polynomials, and that products of terms from the previous layer may now be terms in the key-polynomials of the current layer. (This is caused by auxiliary state addition. See \iftoggle{crypto}{\cite{BJ15}}{Lemma \ref{lem:terms}} for details). Thus,
for $\ell> 1$, we can define $T_{\ell}\subset \mathbb{F}_2[V_{\ell}]$, where \iftoggle{crypto}{$V_{\ell}:=V_{\ell-1}\cup \{k^{(\ell-1)}_{1,i},\dots,k^{(\ell-1)}_{|T_{\ell-1}|,i}\}_{i=1}^n$,}{
$$V_{\ell}:=V_{\ell-1}\cup \left\{k^{(\ell-1)}_{1,i},\dots,k^{(\ell-1)}_{|T_{\ell-1}|,i}\right\}_{i=1}^n,$$}
by
$$T_{\ell}:=T_{\ell-1}\cup \{tt': t,t'\in T_{\ell-1}, t\neq t'\}\cup \left\{k^{(\ell-1)}_{1,i},\dots,k^{(\ell-1)}_{|T_{\ell-1}|,i}\right\}_{i=1}^n.$$
We then provide each of the $n$ wires with an auxiliary state for each term in $T_{\ell}$, for $\ell=1,\dots,L$. We now make this more precise.

To each $T_\ell$, we associate a family of strings $\{s^{(\ell)}(x)\}_{x\in\{0,1\}^{V_\ell}}$ in $\{0,1\}^{T_{\ell}}$, defined so that for every $f\in T_{\ell}$, the $f$-entry of $s^{(\ell)}(x)$ is $s^{(\ell)}_f(x)=f(x).$ That is, $s^{(\ell)}(x)$ represents evaluating every monomial in $T_{\ell}$ at $x$. For instance, we have, for any strings $a,b\in \{0,1\}^n$, $s^{(1)}(a,b)=(a_1,\dots,a_n,b_1,\dots,b_n)$.

For any strings $s,k\in \{0,1\}^n$, define \iftoggle{crypto}{$\sigma(s,k):=\bigotimes_{i=1}^n\ket{+_{s_i,k_i}}\bra{+_{s_i,k_i}}$.}{
$$\sigma(s,k):=\bigotimes_{i=1}^n\ket{+_{s_i,k_i}}\bra{+_{s_i,k_i}}.$$}

For any string $s$, let $s^{*n}$ denote the concatenation of $n$ copies of $s$. For any $a,b\in\{0,1\}^n$ and $k=(k^{(1)},\dots,k^{(L)})\in \{0,1\}^{n|T_1|}\times\dots\times\{0,1\}^{n|T_L|}$, define
$$\sigma_{aux}^{a,b,k}:=\sigma(s^{(1)}(a,b)^{*n},k^{(1)})\otimes \dots \otimes \sigma(s^{(L)}(a,b,k^{(1)},\dots,k^{(L-1)})^{*n},k^{(L)}).$$

\noindent We can now define the procedure $\AUX.\KeyGen(1^\kappa,1^n)$:
\begin{enumerate}
\item Execute $(pk,sk,evk)\leftarrow \HE.\KeyGen(1^{\kappa+n})$.
\item Choose uniform random $a,b\in\{0,1\}^n$ and $k=(k^{(1)},\dots,k^{(L)})\in \{0,1\}^{n|T_1|}\times\dots\times\{0,1\}^{n|T_L|}$.
\item Output secret key $(sk,a,b,k)$.
\item Output evaluation key: $pk$, $evk$, $\tilde{a}_1=\HE.\Enc_{pk}(a_1),\dots,\tilde{a}_n=\HE.\Enc_{pk}(a_n)$,\\ $\tilde{b}_1=\HE.\Enc_{pk}(b_1),\dots,\tilde{b}_n=\HE.\Enc_{pk}(b_n)$, $\(\tilde{k}^{(\ell)}_i=\HE.\Enc_{pk}\(k^{(\ell)}_{j,i}\)\)_{\substack{\ell\in[L]\\ i\in[n] \\ j\in[|T_\ell|]}}$, and $\sigma_{aux}^{a,b,k}$.
\end{enumerate}

\noindent\textbf{Encryption.} $\AUX.\Enc_{(sk,a,b,k),d}: D(\mathcal{M}) \rightarrow D(\mathcal{C})$. The encryption procedure takes an extra parameter~$d$ that keeps track of the number of qubits already encrypted (we assume $d$ is initially~$1$ and not modified outside of $\AUX.\Enc$). If $d \leq n $, \iftoggle{crypto}{}{for a single-qubit register $\mathcal{M}$, }it applies the quantum one-time pad channel $\QEnc_{a_d, b_d}: D(\mathcal{M}) \rightarrow D(\mathcal{C})$.
The output is the cipherstate in register $\mathcal{C}$; the parameter $d$ is updated as $d \leftarrow d+1$.
If $d >n $, then output $\bot$ to indicate an error.

\noindent\textbf{Decryption.} $\AUX.\Dec_{(sk,a,b,k),d}: D(\regC') \rightarrow D(\mathcal{M})$. The decryption is defined the same as $\CL.\Dec_{sk}$. 

\noindent\textbf{Homomorphic Evaluation.} $\AUX.\Eval^{\mathsf{C}}: D(\regR_{evk} \otimes \regC^{\otimes n}) \rightarrow  D(\regC'^{\otimes m})$.
For Clifford group gates, we apply the gadgets as in $\CL.\Eval$. For $\tgate$-gates, we apply the gadget in Fig.~\ref{fig:T-gate-AUX}. The full evaluation procedure is as follows:
\begin{enumerate}
\item Set $V\leftarrow \{a_i,b_i\}_{i\in [n]}$, and
$\forall\,i\in [n]$, \iftoggle{crypto}{}{set $f_{a,i},f_{b,i}\in\mathbb{F}_2[V]$ as} $f_{a,i}\leftarrow a_i$, $f_{b,i}\leftarrow b_i$.
\item Let $\mathsf{g}_1,\dots,\mathsf{g}_G$ be a topological ordering of the gates in $\mathsf{C}$. For $i=1,\dots,G$, evaluate $\mathsf{g}_i$ using the appropriate gadget.
\item Let $S$ be the set of output wire labels. For each $i\in S$:
  \begin{enumerate}
  \item Homomorphically evaluate $f_{a,i}$ and $f_{b,i}$ to obtain updated (encrypted) keys: $\tilde{a}_i\leftarrow \HE.\Eval_{evk}^{f_{a,i}}(\tilde{v}:v\in V)$ and $\tilde{b}_i\leftarrow \HE.\Eval_{evk}^{f_{b,i}}(\tilde{v}:v\in V)$.
  \end{enumerate}
\item Output in ${\cal C'}_i$ the classical-quantum system given by:
  \begin{itemize}
    \item The encrypted keys $\{\tilde{a}_i,\tilde{b}_i\}_{i\in S}$.
  \item The output\iftoggle{crypto}{}{ register}
                     corresponding to the encrypted output qubit $i$ of the circuit.

  \end{itemize}
\end{enumerate}
The correctness of this scheme depends on two facts, which we prove in \iftoggle{crypto}{\cite{BJ15}}{Sec.~\ref{sec:aux-analysis}}. First, for every unknown $v\in V$, we have an encrypted copy of $\tilde{v}$, encrypted using $\HE.\Enc$. We need these to compute the final keys $\{\tilde{a}_i,\tilde{b}_i\}$ using $f_{a,i},f_{b,i}\in\mathbb{F}_2[V]$. Finally, for each level $\ell$, for each wire label $i$, we need an auxiliary state $\ket{+_{t,k}}$ for every term that may appear in the key $f_{a,i}$ going into the \th{\ell} level. This allows us to construct the auxiliary qubit required to execute each $\tgate$-gadget.
\iftoggle{crypto}{In the full version \cite{BJ15}, we prove that $\AUX$ requires $O(n^{2^{L-1}+1})$ auxiliary qubits, from which it follows that $\AUX$ is homomorphic for quantum circuits with $\tgate$-depth $L$. We further show that $\AUX$ is q-IND-CPA and compact.}
{\vskip10pt}

We remark that if we only had a classical encryption scheme that was homomorphic over linear circuits, and not fully homomorphic, then we could get the same functionality from a slightly modified version of this scheme, in which we include with every auxiliary qubit $\ket{+_{s,k}}\bra{+_{s,k}}$, $\HE.\Enc_{pk}(s)$ --- at the moment we only include some of these, but not those auxiliary states arising from \emph{products} of terms, since we can compute products homomorphically. Since we have classical fully homomorphic encryption, we use this to slightly simplify the scheme, however the observation that the fully homomorphic property is not fully taken advantage of strengthens the idea that Clifford circuits are analogous to classical linear circuits in the context of QHE.

\iftoggle{crypto}{
}{
\subsection{Analysis of $\AUX$}\label{sec:aux-analysis}

We now analyse the various properties of $\AUX$. Consider a layered quantum circuit $\mathsf{C}$ with $L$ layers of $\tgate$-gates.
To simplify the analysis, we assume that the ordering of gates $\mathsf{g}_1,\dots,\mathsf{g}_G$ has the property that if $\mathsf{g}_i$ is a $\tgate$-gate in level $\ell$, and $\mathsf{g}_j$ is a $\tgate$-gate in level $\ell+1$, then $i<j$; that is, we completely evaluate level $\ell$ before we begin to evaluate level $\ell+1$.

\begin{lemma}\label{lem:terms}
Let $f_{a,i}$ be a key-polynomial going into the \th{\ell} layer of $\tgate$-gates. Then $f_{a,i}$ is a sum of terms in $T_\ell$.
\end{lemma}
\iftoggle{crypto}{}{
\begin{proof}
We prove this statement by induction on $\ell$. Before any gates have been applied, the key-polynomial are $f_{a,i}=a_i$ and $f_{b,i}=b_i$ for $i=1,\dots,n$. We can easily see from the update rules that applying Clifford gates results in keys of the form $f$ or $f+f'$, where $f$ and $f'$ were previous keys. Thus, after a Clifford circuit has been applied, all key-polynomial are sums of terms from $\{a_1,\dots,a_n,b_1,\dots,b_n\}=T_1$.

Let $f_{a,1},\dots,f_{a,n},f_{b,1},\dots,f_{b,n}$ be the key-polynomials going into the \th{\ell} layer, and suppose they are sums of terms in $T_{\ell}$. Let $f'_{a,1},\dots,f'_{a,n},f'_{b,1},\dots,f'_{b,n}$ be the key-polynomials right after the \th{\ell} layer of $\tgate$-gates has been applied. If no $\tgate$ is applied on the \th{i} wire, then $f_{a,i}'=f_{a,i}$ and $f_{b,i}'=f_{b,i}$, so $f_{a,i}',f_{b,i}'$ are both sums of terms in $T_{\ell}\subset T_{\ell+1}$. Suppose on the other hand that we apply a $\tgate$-gate to the \th{i} wire at level $\ell$.
From the $\tgate$-gadget (Fig.~\ref{fig:T-gate-AUX}), we  see that after applying a $\tgate$ to the \th{i} wire, we have new keys $f_{a,i}'=f_{a,i}\oplus c$ for a known constant $c$, so $f_{a,i}'$ is a sum of terms in $T_\ell\subset T_{\ell+1}$; and $f_{b,i}'=(1\oplus c)f_{a,i}\oplus f_{b,i}\oplus k$, where $k$ is the auxiliary state key of the auxiliary state used to implement the gadget. If $f_{a,i}=t_1\oplus \dots\oplus t_r$, for $t_1,\dots,t_r\in T_\ell$, then we construct $\ket{+_{f_{a,i},k}}$ from auxiliary states $\ket{+_{t_1,k_1}},\dots,\ket{+_{t_r,k_r}}$ for some $k_1,\dots,k_r\in\{k^{(\ell)}_{q,i}\}_{q=1}^{|T_{\ell}|}\subset T_{\ell+1}$, so we have
$k=\bigoplus_{j=1}^rk_j\oplus \bigoplus_{j=2}^rc_jt_j\oplus \bigoplus_{j=1}^r\bigoplus_{j'=1}^{j-1}t_jt_{j'}$ for known $c_2,\dots,c_r$, which is the sum of terms in $T_{\ell+1}$, since $t_1,\dots,t_r\in T_{\ell}$. Thus, $f_{b,i}'$ is the sum of terms in $T_{\ell+1}$.

Thus, after applying the \th{\ell} layer of $\tgate$-gates, all key-polynomials are sums of terms from $T_{\ell+1}$. To complete the proof, we simply observe again that Clifford circuits act additively on the keys, and so do not introduce new terms, so just before the \th{(\ell+1)} layer of $\tgate$-gates, the key-polynomials are still sums of terms in $T_{\ell+1}$.
\end{proof}}

\noindent The bottleneck in this scheme is the number of auxiliary states required:
\begin{lemma}
The number of auxiliary qubits output by $\AUX.\KeyGen(1^\kappa,1^n)$ grows as $O(n^{2^{L-1}+1})$ in $n$.
\end{lemma}
\iftoggle{crypto}{}{
\begin{proof}
The number of qubits encoded in $\sigma_{aux}^{a,b,k}$ is
$$|k^{(1)}|+|k^{(2)}|+\dots+|k^{(L)}|=n|T_1|+n|T_2|+\dots+n|T_L|=n\sum_{\ell=1}^{L}|T_\ell|.$$
From the definition of $T_{\ell}$, we  see that:
$$|T_1|=2n,\quad\mbox{and for $\ell>1$,}\quad |T_\ell|=|T_{\ell-1}|+\binom{|T_{\ell-1}|}{2}+n|T_{\ell-1}|.$$
So certainly for all $\ell>1$, $|T_{\ell}|\leq c|T_{\ell-1}|^2$ for some constant $c$, and thus $|T_{\ell}|\leq c^{\ell-1}(2n)^{2^{\ell-1}}\in O(n^{2^{\ell-1}})$. Thus $n\sum_{\ell=1}^L|T_{\ell}|\in O(n^{2^{L-1}+1})$.
\end{proof}}
\noindent We thus have the following theorem:
\begin{theorem}
Let $\classS_n$ be the class of all quantum circuits on $n$ wires with $\tgate$-depth at most~$L$, and let $\classS=\{\classS_n\}_{n\in\mathbb{N}}$. Then $\AUX$ is $\classS$-homomorphic and compact.
\end{theorem}

\noindent We now consider the security of the scheme.

\begin{theorem}
\label{thm:security:AUX}
If $\HE$ is q-IND-CPA secure, then $\AUX$ is q-IND-CPA secure.
\end{theorem}

We will prove Thm.~\ref{thm:security:AUX} in several parts. To begin, we will show that an adversary that interacts with $\AUX.\KeyGen$ can't do much better than an adversary that interacts instead with an altered version of $\AUX.\KeyGen$, $\KeyGen'$, in which every classical encryption has been replaced with $\HE.\Enc_{pk}(0)$ (Lemma \ref{lem:aux-security-1}). Then we will be able to complete the proof by showing that an adversary interacting with $\KeyGen'$ instead of $\AUX.\KeyGen$ can't win the q-IND-CPA experiment for $\AUX$ with probability better than $\frac{1}{2}$.

\begin{lemma}\label{lem:aux-security-1}
Define a QHE scheme $\AUX'$ such that $\AUX'.\KeyGen(1^{\kappa},1^n)=\KeyGen'(1^\kappa,1^n)$, where $\KeyGen'$ behaves identically to $\AUX.\KeyGen$, except it replaces every classical encryption $\HE.\Enc_{pk}(x)$ with $\HE.\Enc_{pk}(0)$. Let $\AUX'.\Enc=\AUX.\Enc$, $\AUX'.\Dec=\AUX.\Dec$ and $\AUX'.\Eval=\AUX.\Eval$. Then for any quantum polynomial-time adversary $\advA=(\advA_1,\advA_2)$ with encryption oracle access, there exists a negligible function $\eta$ such that:
$$\Pr[\mathsf{SymK_{\advA,\AUX}^{cpa}}(\kappa)=1]-\Pr[\mathsf{SymK_{\advA,\AUX'}^{cpa}}(\kappa)=1]\leq \eta(\kappa).$$
Thus, we can restrict our attention to adversaries that make no use of the classical encryptions, since they add at most a negligible advantage.
\end{lemma}
\iftoggle{crypto}{}{
\begin{proof}
We will define an adversary $\advA'=(\advA_1',\advA_2')$ for the quantum CPA-mult indistinguishability experiment for $\HE$, $\mathsf{PubK_{\advA',\HE}^{cpa\text{-}mult}}(\kappa)$.
Essentially, $\advA'$ will run $\AUX.\KeyGen$, except it will use the challenger $\Xi_{\HE}^{\textsf{cpa-mult}}$ in place of $\HE.\Enc_{pk}$, so that it is either running $\AUX.\KeyGen$ or $\KeyGen'$. It will then simulate the $\sf SymK$ experiment, and if $\advA$ wins, it will guess that it ran the original version of $\AUX.\KeyGen$, and otherwise it will guess that it ran $\KeyGen'$.

\begin{description}
\item[$\advA_1'(pk,evk)$:] $\advA_1'$ chooses uniform random bit strings $a,b\in\{0,1\}^n$ and $k\in\{0,1\}^N$, where $N=n|T_1|+\dots+n|T_L|$, and gives $m_0=\mathbf{0}=0^{2n+N}$ and $m_1=(a,b,k)$ to the challenger $\Xi_{\HE}^{\textsf{cpa-mult}}$, which outputs either $c_1=\HE.\Enc_{pk}(a,b,k)$, or $c_0=\HE.\Enc_{pk}(\mathbf{0})$.
\item[$\advA_2'(c)$:] $\advA_2'$ computes $\sigma_{aux}^{a,b,k}$ and gives $\sigma_{aux}^{a,b,k},pk,evk,c$ to $\advA_1$. $\advA_1$ may make several oracle calls, which $\advA_2'$ can simulate, because it has $a,b$ and so can run $\AUX.\Enc$. When $\advA_1$ outputs a message to the challenger, $\advA_2'$ samples a random bit $r$, and runs $\Xi_{\AUX}^{\mathsf{cpa},r}$, which it can simulate, since it has $a,b$, and so can run $\AUX.\Enc$. $\advA_2'$ then gives the challenge to $\advA_2$, and if $\advA_2$ outputs $r$, $\advA_2'$ outputs $1$, and otherwise, $\advA_2'$ outputs $0$.
\end{description}
We now calculate the probability that $\advA'$ correctly guesses which of $c_0$ and $c_1$ it received from the challenger, which we know must be less than $\frac{1}{2}+\eta(\kappa+n)$ for some negligible function, since $\HE$ is q-IND-CPA, $\kappa+n$ is the security parameter given to $\HE.\KeyGen$, and $|m_0|=|m_1|=2n+N = O(\mathrm{poly}(n))=O(\mathrm{poly}(n+\kappa))$. If $\advA'$ received $c_0$, then it acted as $\KeyGen'$, whereas if it received $c_1$, it acted as $\AUX.\KeyGen$. In the former case, the probability that $\advA'$ correctly guesses 0 is the probability that $\advA$ loses the $\sf SymK$ experiment when it interacts with $\AUX'$, $\Pr[\mathsf{SymK_{\advA,\AUX'}^{cpa}}(\kappa)=0]$. In the latter case, the probability that $\advA'$ correctly guesses 1 is the probability that $\advA$ wins the $\sf SymK$ experiment when it interacts with $\AUX$, $\Pr[\mathsf{SymK_{\advA,\AUX}^{cpa}}(\kappa)=1]$. Thus, since $\HE$ is q-IND-CPA, there exists a negligible function $\eta'$ such that
\begin{align*}
\frac{1}{2}\Pr[\mathsf{SymK_{\advA,\AUX'}^{cpa}}(\kappa)=0]+\frac{1}{2}\Pr[\mathsf{SymK_{\advA,\AUX}^{cpa}}(\kappa)=1] & \leq \frac{1}{2}+\eta'(\kappa)\\
1-\Pr[\mathsf{SymK_{\advA,\AUX'}^{cpa}}(\kappa)=1]+\Pr[\mathsf{SymK_{\advA,\AUX}^{cpa}}(\kappa)=1] &\leq 1+2\eta'(\kappa)\\
\Pr[\mathsf{SymK_{\advA,\AUX}^{cpa}}(\kappa)=1]-\Pr[\mathsf{SymK_{\advA,\AUX'}^{cpa}}(\kappa)=1] &\leq 2\eta'(\kappa).
\end{align*}
Setting $\eta=2\eta'$ completes the proof.
\end{proof}
}

The next lemma shows that the output of $\KeyGen'$ is actually $(pk,evk,\maxmix)$, which is independent of $a,b,k$. The proof of Lemma \ref{lem:aux-security-2} is mainly computational, and provides little insight, so we relegate it to App.~\ref{app:misc}.

\begin{lemma}\label{lem:aux-security-2}
Let $N=n|T_1|+\dots+n|T_L|$. For any $a,b\in\{0,1\}^{n}$, \iftoggle{crypto}{$\sum_{k\in\{0,1\}^{N}}\sigma_{aux}^{a,b,k}=\frac{1}{2^{N}}\mathbb{I}_{2^{N}}$.}{
$\sum_{k\in\{0,1\}^{N}}\sigma_{aux}^{a,b,k}=\frac{1}{2^{N}}\mathbb{I}_{2^{N}}.$}
\end{lemma}

To complete the proof of Thm.~\ref{thm:security:AUX}, we show that no adversary interacting with $\AUX.\KeyGen'$ can win the experiment $\sf SymK^{cpa}$ with probability better than $\frac{1}{2}$.

\begin{lemma}\label{lem:aux-security-3}
For any adversary $\advA=(\advA_1,\advA_2)$ with access to an encryption oracle,
$$\Pr[\mathsf{SymK_{\advA,\AUX'}^{cpa}}(\kappa)=1]= \frac{1}{2}.$$
\end{lemma}
\iftoggle{crypto}{}{
\begin{proof}
Let $q$ be the number of oracle calls made by $\advA_1$, and write $\advA_1=(\advA_1^{(1)},\dots,\advA_1^{(q+1)})$. Let $q'$ be the number of oracle calls made by $\advA_2$, and write $\advA_2=(\advA_2^{(1)},\dots,\advA_2^{(q'+1)})$. If $q\geq n$, then the challenger just outputs $\bot$, independent of $r$, so certainly in that case $\advA$ cannot win with probability more than $\frac{1}{2}$, so suppose $q<n$. If $q+q'+1>n$, then the last $q+q'+1-n$ oracle calls made by $\advA_2$ simply return $\bot$, which $\advA$ could simulate without actually making these oracle calls, so suppose without loss of generality that $q+q'+1\leq n$.

The output of $\AUX'.\KeyGen=\KeyGen'$ to $\advA_1^{(1)}$ is $(pk,evk,\sigma_{aux}^{a,b,k})$, and
by Lemma \ref{lem:aux-security-2}, for any $a,b$, this is equal to $(pk,evk,\frac{1}{2^{N}}\mathbb{I}_{2^{N}})$. Thus, the interaction of $\KeyGen'$ with the experiment is shown in part (a) of Fig.~\ref{fig:aux-security}. $\KeyGen'$ chooses random bits $a_1,b_1,\dots,a_{q+q'+1},b_{q+q'+1}$, for use in oracle calls and the challenge itself, but these are independent of the information given to $\advA$ by $\KeyGen'$. (The other random bits selected by $\KeyGen'$, $a_{q+q'+2},b_{q+q'+2},\dots,a_n,b_n$ and the string $k$, are independent of the interaction with the adversary, so we ignore them.)

It is then easy to see from Fig.~\ref{fig:aux-security} that every call to the encryption oracle can be replaced by a channel that discards the input and returns a completely mixed state, since for any input $\rho^{\mathcal{M}}$, the encryption oracle returns
$$\Tr_1\(\frac{1}{4}\sum_{a,b\in\{0,1\}}\ket{a,b}\bra{a,b}_1\otimes\xgate^{a}\zgate^{b}\rho^{\cal M}\zgate^{b}\xgate^{a}\)=\frac{1}{4}\sum_{a,b\in\{0,1\}}\xgate^{a}\zgate^{b}\rho^{\cal M}\zgate^{b}\xgate^{a}=\frac{1}{2}\mathbb{I}_{2}.$$
\newpage
In other words, we have:
\begin{figure}[H]
\centering
\begin{tikzpicture}
\node at (0,0){
\begin{tikzpicture}
\draw (0,.79)--(.75,.79);
\draw (0,.71)--(.75,.71);
\draw (0,0)--(2.75,0);

\filldraw[fill=white] (-.5,.5) rectangle (0,1);
\node at (-.25,.75) {$\maxmix$};

\node at (.4,.95) {$a,b$};
\node at (.5,.2) {\small $\cal M$};

\filldraw[fill=white] (.75,-.25) rectangle (2,1);
\node at (1.375,0) {\small $\QEnc_{a,b}$};

\node at (2.25,.2) {\small $\cal C$};
\end{tikzpicture}};
\node at (2.5,0) {$\equiv$};
\node at (5,0) {
\begin{tikzpicture}
\draw (0,0) -- (1,0); \draw (2,0) -- (2.75,0);

\node at (.4,.2) {\small $\cal M$};

\node at (1,0) {\meas};

\filldraw[fill=white] (1.7,-.25) rectangle (2.2,.25);
\node at (1.95,0) {$\maxmix$};

\node at (2.45,.2) {\small $\cal C$};
\end{tikzpicture}
};
\end{tikzpicture}
\end{figure}
\noindent Here $\maxmix$ denotes the channel that outputs a completely  mixed state, or equivalently, a uniform random variable.

For the same reason, the call to the challenger $\Xi_{\AUX}^{{\sf cpa},r}$ can also be replaced with the channel that discards the input and returns $\maxmix$, since $\Xi_{\AUX}^{{\sf cpa},r}$ applies a quantum one-time pad using random keys $a_{q+1},b_{q+1}$ to the input or to $\ket{0}\bra{0}$, and in either case, the resulting state is the completely mixed state. Thus, from the perspective of $\advA$, the experiment is independent of $r$, as shown in part (b) of Fig.~\ref{fig:aux-security}. Thus, an adversary cannot win with probability better than $\frac{1}{2}$.
\end{proof}
}

\begin{figure}
\centering
\begin{tikzpicture}
\node at (0,0) {
\begin{tikzpicture}
\draw (1,4.25)--(13,4.25);
\draw (1,3.25)--(13,3.25);
\draw (1,2.5)--(10.5,2.5);
\draw (1,1.5)--(3.5,1.5);	
\draw (0,.75)--(13,.75);
\draw (1,0) -- (13,0);

\filldraw[fill=white] (-.5,-.25) rectangle (0,2.75);
\node[rotate=90] at (-.25,1.25) {$\HE.\KeyGen(1^{\kappa+n})$};
\node at (.75,.95) {$pk,evk$};

\filldraw[fill=white] (.8,4) rectangle (1.3,4.5);
\node at (1.05,4.25) {$\maxmix$};
\node at (1.05,3.85) {$\vdots$};
\filldraw[fill=white] (.8,3) rectangle (1.3,3.5);
\node at (1.05,3.25) {$\maxmix$};
\filldraw[fill=white] (.8,2.25) rectangle (1.3,2.75);
\node at (1.05,2.5) {$\maxmix$};
\node at (1.05,2.1) {$\vdots$};
\filldraw[fill=white] (.8,1.25) rectangle (1.3,1.75);
\node at (1.05,1.5) {$\maxmix$};

\filldraw[fill=white] (.8,-.25) rectangle (1.3,.25);
\node at (1.05,0) {$\maxmix$};

\filldraw[fill=white] (1.8,-.25) rectangle (2.7,1);
\node at (2.25,.375) {$\advA_1^{(1)}$};

\node at (2.8,4.45) {\small $a_{q+q'+1},b_{q+q'+1}$};
\node at (2.35,3.45) {\small $a_{q+2},b_{q+2}$};
\node at (2.35,2.7) {\small $a_{q+1},b_{q+1}$};
\node at (2,1.7) {\small $a_1,b_1$};
\node at (2.95,.95) {\small $\cal M$};
\node at (2.95,.2) {\small $\cal E$};

\filldraw[fill=white] (3.25,.5) rectangle (4.75,1.75);
\node at (4,.75) {\small$\QEnc_{a_1,b_1}$};

\node at (6,2) {\small (encryption oracle)};
\draw[dashed,->] (6,1.8) -- (4.75,1.25);

\node at (4.95,.95) {\small $\cal C$};

\filldraw[fill=white] (5.3,-.25) rectangle (6.2,1);
\node at (5.75,.375) {$\advA_1^{(2)}$};

\node at (6.45,.95) {\small $\cal M$};
\node at (6.45,.2) {\small $\cal E$};

\fill[fill=white] (6.9,-.25) rectangle (7.85,1.75);
\node at (7.375,.75) {$\dots$};

\node at (8.325,.95) {\small $\cal C$};
\node at (8.325,.2) {\small $\cal E$};

\filldraw[fill=white] (8.55,-.25) rectangle (9.7,1);
\node at (9.125,.375) {$\advA_1^{(q+1)}$};

\node at (9.95,.95) {\small $\cal M$};
\node at (9.95,.2) {\small $\cal E$};

\filldraw[fill=white] (10.4,.5) rectangle (11.6,2.75);
\node at (11,1.625) {$\Xi_{\AUX}^{{\sf cpa},r}$};

\node at (11.85,.95) {\small $\cal C$};

\filldraw[fill=white] (12.3,1.25) rectangle (13.2,4.5);
\node[rotate=90] at (12.75,2.875) {$\AUX.\Enc$};
\draw[->] (12.4,1) -- (12.4,1.25);
\draw[->] (13.1,1.25) -- (13.1,1);
\filldraw[fill=white] (12.3,-.25) rectangle (13.2,1);
\node at (12.75,.375) {$\advA_2$};

\draw (13.2,.415)--(13.7,.415);
\draw (13.2,.335)--(13.7,.335);
\node at (13.95,.375) {$r'$};

\draw[dashed] (-.75, -.5) rectangle (1.5,5.25);
\node at (.375,5) {$\KeyGen'$};

\end{tikzpicture}
};

\node at (0,-3.25) {(a)};

\node at (0,-5){
\begin{tikzpicture}
\draw (0,.75)--(13,.75);
\draw (1,0) -- (13,0);

\filldraw[fill=white] (-.5,-.25) rectangle (0,2.5);
\node[rotate=90] at (-.25,1.125) {\small$\HE.\KeyGen(1^{\kappa+n})$};
\node at (.75,.95) {$pk,evk$};

\filldraw[fill=white] (.8,-.25) rectangle (1.3,.25);
\node at (1.05,0) {$\maxmix$};

\filldraw[fill=white] (1.8,-.25) rectangle (2.7,1);
\node at (2.25,.375) {$\advA_1^{(1)}$};

\node at (2.95,.95) {\small $\cal M$};
\node at (2.95,.2) {\small $\cal E$};

\fill[fill=white] (3.4,.5) rectangle (4.6,1.75);
\node at (3.6,.75) {\meas};
\filldraw[fill=white] (4.25,.5) rectangle (4.75,1);
\node at (4.5,.75) {$\maxmix$};

\node at (5,.95) {\small $\cal C$};

\filldraw[fill=white] (5.3,-.25) rectangle (6.2,1);
\node at (5.75,.375) {$\advA_1^{(2)}$};

\node at (6.45,.95) {\small $\cal M$};
\node at (6.45,.2) {\small $\cal E$};

\fill[fill=white] (6.9,-.25) rectangle (7.85,1.75);
\node at (7.375,.75) {$\dots$};

\node at (8.325,.95) {\small $\cal C$};
\node at (8.325,.2) {\small $\cal E$};

\filldraw[fill=white] (8.55,-.25) rectangle (9.7,1);
\node at (9.125,.375) {$\advA_1^{(q+1)}$};

\node at (9.95,.95) {\small $\cal M$};
\node at (9.95,.2) {\small $\cal E$};

\fill[fill=white] (10.4,.5) rectangle (11.6,1.75);
\node at (10.6,.75) {\meas};
\filldraw[fill=white] (11.25,.5) rectangle (11.75,1);
\node at (11.5,.75) {$\maxmix$};

\node at (12,.95) {\small $\cal C$};

\node at (12.3,1.5) {\meas};
\filldraw[fill=white] (12.85,1.25)rectangle(13.35,1.75);
\node at (13.1,1.5) {$\maxmix$};
\draw[->] (12.4,1)--(12.4,1.25);
\draw[->] (13.1,1.25)--(13.1,1);
\filldraw[fill=white] (12.3,-.25) rectangle (13.2,1);
\node at (12.75,.375) {$\advA_2$};

\draw (13.2,.415)--(13.7,.415);
\draw (13.2,.335)--(13.7,.335);
\node at (13.95,.375) {$r'$};

\end{tikzpicture}
};
\node at (0,-6.75) {(b)};
\end{tikzpicture}

\caption{Proof of Lemma \ref{lem:aux-security-3}. Part (a) shows how $\KeyGen'$ interacts with the experiment. The channel $\maxmix$ outputs a completely  mixed state, or equivalently, a uniform random variable. Since the random bits $a_i,b_i$ are independent of the other outputs of $\KeyGen'$, for each $i$, we can replace each of the oracle calls as well as the challenger with a channel that discards the input and returns a completely mixed state, as shown in part (b). Thus, the experiment is independent of $r$ from the perspective of $\advA$, and so $\advA$ can do no better than guessing $r$. }\label{fig:aux-security}
\end{figure}

\noindent Combining Lemma \ref{lem:aux-security-1} and Lemma \ref{lem:aux-security-3} proves Thm.~\ref{thm:security:AUX} immediately.
}

\iftoggle{crypto}{}{
\section{Conclusions and Open Problems}

In this work, we have presented three quantum homomorphic encryption schemes. The first, $\CL$, is a stepping stone to the other two, and is homomorphic and compact for the class of stabilizer circuits. The second, $\EPR$, is homomorphic for all quantum circuits, but the compactness property degrades with the number of $\tgate$-gates. In the third scheme, $\AUX$, the complexity of the evaluation key and the evaluation procedure scale doubly exponentially with the $\tgate$-depth, so that it is only homomorphic for circuits with constant $\tgate$-depth, but it is also compact.

The clear central open problem in this work is to come up with a quantum fully homomorphic encryption scheme satisfying Def.~\ref{defn:QFHE}, which must be homomorphic for all quantum circuits and compact. Our schemes $\EPR$ and $\AUX$ make progress towards this goal from two directions, but still leave open a full solution to this problem.

Our work can be seen as analogous to a number of classical results leading up to fully homomorphic encryption, including classical encryption schemes that were homomorphic for some limited classes of circuits, including limits in the multiplicative depth, as well as quasi-compact homomorphic schemes. In addition, we have attempted, in our security definitions and the theorems in App.~\ref{sec:appendix-equivalent-q-IND-CPA}, to set the groundwork for a rigorous treatment of quantum homomorphic encryption, hopefully leading, eventually to quantum fully homomorphic encryption.
}

\iftoggle{crypto}{}{
\section*{Acknowledgements}
\addcontentsline{toc}{section}{Acknowledgements}

The authors would like to thank Fang Song for helpful discussions about this paper, in particular regarding security definitions.

A.\,B.\ acknowledges support from the Canadian Institute for Advanced Research (\textsc{Cifar}); A.\,B.\ and S.\,J.\ acknowledge support from the Natural Sciences and Engineering Research Council of Canada (\textsc{Nserc}).
Part of this work was done while the authors were visitors at the Simons Institute for the Theory of Computing.
}

\addcontentsline{toc}{section}{References}
\iftoggle{crypto}{\bibliographystyle{plain}}
{\bibliographystyle{alpha}}
\small
\bibliography{references}

\iftoggle{crypto}{\end{document}}{}

\appendix

\section{Classical Fully Homomorphic Encryption}
\label{appendix:Classical-FHE}
Here we present the definitions from the full version of \cite{BV11}.

\begin{definition}A homomorphic encryption scheme is a $4$-tuple $(\sf{HE.KeyGen},\sf{HE.Enc},\sf{HE.Dec},\sf{HE.Eval})$ of PPT algorithms such that:

\begin{description}
\item[{Key Generation.}] The algorithm $(pk, evk, sk) \leftarrow \mathsf{HE.KeyGen}(1^\kappa)$ takes a unary representation of the security parameter and outputs a public encryption key $pk$, a public evaluation key $evk$ and a secret decryption key $sk$.

\item[{Encryption.}] The algorithm $c \leftarrow \mathsf{HE.Enc}_{pk}(\mu)$ takes the public key $pk$ and a single bit message $\mu \in \{0,1\}$ and outputs a ciphertext~$c$.

\item[{Decryption.}] The algorithm $\mu^* \leftarrow \mathsf{HE.Dec}_{sk}(c)$ takes the secret key $sk$ and a ciphertext $c$ and outputs a message $\mu^* \in \{0,1\}$.

 \item[{Homomorphic Evaluation.}] The algorithm $c_f \leftarrow \mathsf{HE.Eval}_{evk}^{f}(c_1, \ldots, c_\ell)$ takes the evaluation key $evk$, a classical circuit $f: \{0,1\}^\ell \rightarrow \{0,1\}$ and a set of $\ell$ ciphertexts  $c_1, \ldots, c_\ell$ and outputs a ciphertext $c_f$.

\end{description}

\end{definition}

We  define $\mathscr{S}$-homomorphic, which is homomorphism with respect to a specified class
$\mathscr{S}$ of circuits. This notion is sometimes also referred to as ``somewhat homomorphic''.

\begin{definition}[$\mathscr{S}$-homomorphic] Let $\mathscr{S} = \{\classS_\kappa\}_{\kappa \in \mathbb{N}}$ be a class of classical circuits. A scheme $\HE$ is $\mathscr{S}$-homomorphic (or, homomorphic for the class $\mathscr{S}$) if
for any sequence of circuits $\{f_\kappa \in \mathscr{S}_\kappa\}_{\kappa\in \mathbb{N}}$ and respective inputs $\mu_1, \ldots , \mu_\ell \in \{0,1\}$ (where $\ell = \ell(\kappa))$, there exists a negligible function $\eta$ such that
\begin{equation}
\mathrm{Pr}[\mathsf{HE.Dec}_{sk}(\mathsf{HE.Eval}_{evk}^{f}(c_1, \ldots ,c_\ell)) \neq  f(\mu_1, \ldots, \mu_\ell)] = \eta(\kappa)\,,
\end{equation}
where $(pk, evk, sk) \leftarrow \mathsf{HE.Keygen}(1^\kappa)$ and $c_i \leftarrow \mathsf{HE.Enc_{pk}}(\mu_i)$.
\end{definition}

\begin{definition}[compactness] A homomorphic scheme $\mathsf{HE}$ is \emph{compact} if there exists a polynomial
$p$ such that the circuit complexity of decrypting the output of $\mathsf{HE.Eval}^f(\cdots)$ is at most $p(\kappa)$ (regardless of $f$).
\end{definition}

\begin{definition}[fully homomorphic encryption] A scheme $\mathsf{HE}$ is fully homomorphic if it is both
compact and homomorphic for the class of all arithmetic circuits over~$\mathbb{F}_2$.
\end{definition}

\iftoggle{crypto}{
\section{Quantum Semantic Security}}
{\section{Equivalence of Definitions for q-IND-CPA}}
\label{sec:appendix-equivalent-q-IND-CPA}

\iftoggle{crypto}{
In this section we give more details about our notion of quantum semantic security, q-IND-CPA. Specifically, we show that it is equivalent to a seemingly stronger security definition involving multiple-message attacks, q-IND-CPA-mult, and give analogous definitions for the case of symmetric-key encryption.

\paragraph{CPA-mult security} The CPA-mult indistinguishability experiment is similar to the CPA scenario above, but in this case the adversary chooses two $t$-tuples of messages, for any $t\geq 1$, and the challenger returns encryptions corresponding to one of the $t$-tuples. The adversary's task is then to guess which of the two $t$-tuples of messages has been encrypted. The experiment is given below; the illustration follows closely the one in Fig.~\ref{fig:q-IND-CPA-2} (but with single messages replaced by $t$-fold messages).

\noindent\textbf{The quantum CPA-mult indistinguishability experiment $\mathsf{PubK^{cpa\text{-}mult}_{\advA, QHE}} (\kappa)$}
\begin{enumerate}
\item $\mathsf{KeyGen}(1^\kappa)$ is run to obtain keys $(pk,sk,\rho_{evk})$.
\item For $r \in \{0,1\}$, and $t\in O(\mathrm{poly}(\kappa))$, let $\mathcal{M}_r =\mathcal{M}_r^1\otimes  \cdots \otimes \mathcal{M}_r^t$, where  $\mathcal{M}_0^i \equiv\mathcal{M}_1^i\equiv\mathcal{M}$ (for all~$i$).
 Adversary $\advA_1$ is given $(pk,\rho_{evk})$ and outputs a quantum state $\rho$ in $\mathcal{M}_0 \otimes \mathcal{M}_1  \otimes \mathcal{E}$.
\item For $r\in \{0,1\}$, let   $\Xi^{{\sf cpa\text{-}mult},r}_{\QHE}: D(\mathcal{M}_0 \otimes \mathcal{M}_1) \rightarrow D(\mathcal{C}^1 \otimes \cdots \otimes \mathcal{C}^t)$ be given by
$\Xi_{\QHE}^{{\sf cpa\text{-}mult},0}(\rho)=  \Tr_{\mathcal{M}_1}({\sf Enc}_{pk}^{\otimes t} \otimes \id_{\mathcal{M}_1})(\rho)$ and $\Xi_{\QHE}^{\textsf{cpa-mult},1}(\rho)=  \Tr_{\mathcal{M}_0}(\id_{\mathcal{M}_0} \otimes {\sf Enc}_{pk}^{\otimes t})(\rho)$.
A random bit $r \in \{0,1\}$ is chosen and $(\Xi_{\QHE}^{\textsf{cpa-mult},r} \otimes \id_{\cal E})$ is applied to $\rho$ (the output being a state in $\mathcal{C}^{\otimes t} \otimes \mathcal{E}$).
\item Adversary $\advA_2$ obtains the system in $\mathcal{C}^{\otimes t} \otimes \mathcal{E}$ and outputs a bit $r'$.
\item The output of the experiment is defined to be~1 if~$r'=r$ and~$0$~otherwise. In case $r=r'$, we say that $\advA$ \emph{wins} the experiment.
\end{enumerate}

\begin{definition}[Quantum Indistinguishability under Multiple Chosen Plaintext Attack \mbox{(q-IND-CPA-mult)}]\label{def:q-IND-CPA-mult}
A quantum homomorphic scheme $\sf{QHE}$ is \emph{q-IND-CPA-mult} secure if for all quantum polynomial-time adversaries {$\advA = (\advA_1, \advA_2)$} there exists a negligible function $\eta$ such that:
\begin{equation*}
\Pr[\mathsf{PubK^{{cpa\text{-}mult}}_{\advA, QHE}} (\kappa) =1] \leq \frac{1}{2} +  \eta(\kappa)\,.
\end{equation*}
\end{definition}

}{In this section, we prove Thm.~\ref{thm:equiv-IND-CPA} of Sec.~\ref{sec:security-QHE}, restated below:}

\iftoggle{crypto}{
Clearly if a QHE is q-IND-CPA-mult it is also q-IND-CPA, but it is less obvious that the other direction holds. In order to prove that q-IND-CPA and q-IND-CPA-mult are actually equivalent security definitions,
}{
\vskip5pt
\noindent\textbf{Theorem~\ref{thm:equiv-IND-CPA}} (Equivalence of q-IND-CPA and q-IND-CPA-mult).\emph{
Let $\QHE$ be a quantum homomorphic encryption scheme. Then $\QHE$ is q-IND-CPA if and only if $\QHE$ is q-IND-CPA-mult.
}
\vskip5pt

In order to prove Thm.~\ref{thm:equiv-IND-CPA},}
 we will first introduce an intermediate security definition, q-IND-CPA-2, which is like q-IND-CPA-mult when we restrict to $t=1$. We will show that q-IND-CPA is equivalent to q-IND-CPA-2, and then that q-IND-CPA-2 is equivalent to q-IND-CPA-mult. \iftoggle{crypto}{}{We note that the proof given in this section is easily modified to a proof of Thm.~\ref{thm:equiv-IND-CPA-sym} --- a similar statement for the symmetric-key setting. The only difference in the symmetric-key case is that adversaries can make calls to an encryption oracle. The main techniques in this section involve constructing an adversary $\advA'$ that runs some other adversary $\advA$. If $\advA$ is an adversary with oracle access, then another adversary with oracle access, $\advA'$, can easily run $\advA$, so the same ideas go through in an identical manner in the symmetric-key case.}

\paragraph{CPA-2 security} The CPA-2 indistinguishability experiment is given below and illustrated in Fig.~\ref{fig:q-IND-CPA-2}.

\begin{figure}[h]
\centering
\begin{tikzpicture}
\node at (0,0) {
\begin{tikzpicture}
\draw (-.25,1.79) -- (2,1.79);
\draw (-.25,1.71) -- (2,1.71);
\draw (.5,0)--(4,0);
\draw (.5,.75)--(2,.75);
\draw (.5,1.25) -- (3.75,1.25);
\draw (-.25,0.04)--(.5,0.04);
\draw (-.25,-.04)--(.5,-.04);
\draw (-.25,1.25)--(.5,1.25);

\draw (4.25,.665)--(4.75,.665);
\draw (4.25,.585)--(4.75,.585);

\filldraw[fill=white] (-.75,2) rectangle (-.25,-.25);
\node[rotate=90] at (-.5,.875) {\small $\KeyGen(1^\kappa)$};
\node at (0,.2) {$pk$};
\node at (.1,1.45) {\small $\mathcal{R}_{evk}$};
\node at (0,1.95) {$pk$};

\filldraw[fill=white] (.5,1.5) rectangle (1,-.25);
\node at (.75,.625) {$\advA_1$};

\node at (1.25,.2) {\small $\cal E$};
\node at (1.3,.95) {\small $\mathcal {M}_1$};
\node at (1.3,1.45) {\small $\mathcal {M}_0$};

\filldraw[fill=white] (1.75,.5) rectangle (3.25,2);
\node at (2.5,1.25) {$\Xi_{\QHE}^{\textsf{cpa-2},r}$};

\node at (3.5,1.45) {\small $\cal C$};

\filldraw[fill=white] (3.75,1.5) rectangle (4.25,-.25);
\node at (4,.625) {$\advA_2$};

\node at (5,.625) {$b'$};

\end{tikzpicture}
};

\node at (5.5,0) {
\begin{tikzpicture}
\draw[dashed] (-1.25,-1) rectangle (2.65,1);

\node at (-.5,0) {$\Xi_{\QHE}^{\textsf{cpa-2},0}$:};
\draw (.5,-.5) -- (1.5,-.5);
\draw (.5,.5) -- (2.5,.5);
\node at (.65,.7) {\small $\mathcal {M}_0$};
\node at (.65,-.3) {\small $\mathcal {M}_1$};

\filldraw[fill=white] (1,.25) rectangle (2,.75);
\node at (1.5,.5) {${\sf Enc}_{pk}$};
\node at (2.25,.7) {\small $\cal C$};
\node at (1.5,-.5) {\meas};
\end{tikzpicture}
};
\node at (9.5,0) {
\begin{tikzpicture}
\draw[dashed] (-1.25,-1) rectangle (2.65,1);

\node at (-.5,0) {$\Xi_{\QHE}^{\textsf{cpa-2},1}$:};
\draw (.5,-.5) -- (2.5,-.5);
\draw (.5,.5) -- (1.5,.5);
\node at (.65,.7) {\small $\mathcal {M}_0$};
\node at (.65,-.3) {\small $\mathcal {M}_1$};

\node at (1.5,.5) {\meas};
\filldraw[fill=white] (1,-.25) rectangle (2,-.75);
\node at (1.5,-.5) {${\sf Enc}_{pk}$};
\node at (2.25,-.3) {\small $\cal C$};
\end{tikzpicture}
};
\end{tikzpicture}
\caption{The quantum IND-CPA-2 indistinguishability experiment.}\label{fig:q-IND-CPA-2}
\end{figure}

\noindent\textbf{The quantum IND-CPA-2 indistinguishability experiment $\mathsf{PubK^{cpa\text{-}2}_{\advA, QHE}} (\kappa)$}
\begin{enumerate}
\item $\mathsf{KeyGen}(1^\kappa)$ is run to obtain keys $(pk,sk,\rho_{evk})$.
\item Adversary $\advA_1$ is given $(pk,\rho_{evk})$ and outputs a quantum state $\rho$ in $\mathcal{M}_0\otimes \mathcal{M}_1 \otimes \mathcal{E}$, where $\mathcal{M}_0\equiv\mathcal{M}_1\equiv\mathcal{M} $
\item For $r\in \{0,1\}$, let   $\Xi_{\QHE}^{\textsf{cpa-2},r}: D(\mathcal{M}_0 \otimes \mathcal{M}_1) \rightarrow D(\mathcal{C})$ be given by
$\Xi_{\QHE}^{\textsf{cpa-2},0}(\rho)=  \Tr_{\mathcal{M}_1}\(({\sf Enc}_{pk}^{\mathcal{M}_0}\otimes \id_{{\cal M}_1})(\rho)\)$ and $\Xi_{\QHE}^{\textsf{cpa-2},1}(\rho)=  \Tr_{\mathcal{M}_0}\((\id_{{\cal M}_0} \otimes {\sf Enc}_{pk}^{\mathcal{M}_1})(\rho)\)$.
A random bit $r \in \{0,1\}$ is chosen and $\Xi_{\QHE}^{\textsf{cpa-2},r} \otimes \id_{\cal E}$ is applied to~$\rho$ (the output being a state in $\mathcal{C} \otimes \mathcal{E}$).
\item Adversary $\advA_2$ obtains the system in $\mathcal{C} \otimes \mathcal{E}$ and outputs a bit $r'$.
\item The output of the experiment is defined to be~1 if~$r'=r$ and~$0$~otherwise. In case $r=r'$, we say that $\advA$ \emph{wins} the experiment.
\end{enumerate}

\begin{definition}[q-IND-CPA-2] \label{def:q-IND-CPA-2}
A quantum homomorphic encryption scheme $\sf{QHE}$ is \emph{q-IND-CPA-2} secure if for all quantum polynomial-time adversaries \iftoggle{crypto}{$\advA$}{$\advA = (\advA_1, \advA_2)$} there exists a negligible function $\eta$ such that\iftoggle{crypto}{ $\Pr[\mathsf{PubK^{\mathsf{cpa\text{-}2}}_{\advA, QHE}} (\kappa) =1] \leq \frac{1}{2} +  \eta(\kappa)$.}{:
\begin{equation*}
\Pr[\mathsf{PubK^{\mathsf{cpa\text{-}2}}_{\advA, QHE}} (\kappa) =1] \leq \frac{1}{2} +  \eta(\kappa)\,.
\end{equation*}}
\end{definition}

\begin{theorem}[Equivalence of q-IND-CPA and q-IND-CPA-2]\label{th:equiv}
A quantum homomorphic encryption scheme is q-IND-CPA if and only if it is q-IND-CPA-2.
\end{theorem}
\begin{proof}
It is trivial to see that if a scheme is q-IND-CPA-2 then it is q-IND-CPA.

Suppose a scheme $\QHE$ is q-IND-CPA. Then let $\advA=(\advA_1,\advA_2)$ be any adversary for the experiment ${\sf PubK}_{\advA,\QHE}^{{\sf cpa\text{-}2},r}$, so $\advA_1$ implements a quantum channel from $D(\mathcal{R}_{evk})$ to $D(\mathcal{M}_0\otimes\mathcal{M}_1\otimes\mathcal{E})$ conditioned on $pk$, and $\advA_2$ implements a quantum channel on $D(\mathcal{C}\otimes\mathcal{E})$ that outputs a bit.
We will use $\advA$ to construct an adversary, $\advA'$ for the experiment $\mathsf{PubK}_{\advA',\QHE}^{{\sf cpa}}$ as shown in Fig.~\ref{fig:equiv}. We define $\advA_1'$ by $\advA'_1(pk,\rho^{\mathcal{R}_{evk}})=\Tr_{\mathcal{M}_0}(\advA_1(pk,\rho^{\mathcal{R}_{evk}}))$, and $\advA_2'$ by $\advA_2'=\advA_2$.

We now consider the probability that $\advA'$ wins the q-IND-CPA experiment: $\Pr[\mathsf{PubK^{cpa}_{\advA',\QHE}}(\kappa)=1]$.
If $r=0$, the probability that $r'=r$ is $\Pr[\advA_2((\mathsf{Enc}(\egoketbra{\mathbf{0}}))\otimes \rho^{\cal E})=0]$. If $r=1$, the probability that $r'=r$ is $\Pr[\advA_2( (\mathsf{Enc}_{pk}\otimes \mathbb{I}_{\cal E})(\rho^{\mathcal{M}_1\otimes \mathcal{E}}))=1]$. Thus, the probability that this adversary correctly predicts $r$ is
\begin{equation}
\frac{1}{2}\Pr\left[\advA_2((\mathsf{Enc}(\ket{0}\bra{0}))\otimes \rho^{\cal E})=0\right]+\frac{1}{2}\Pr\left[\advA_2( \mathsf{Enc}_{pk}\otimes \mathbb{I}_{\cal E}(\rho^{\mathcal{M}_1\otimes \mathcal{E}}))=1\right]\leq \frac{1}{2}+\eta'(\kappa),
\label{eq:first}
\end{equation}
for some negligible function $\eta'$, by the fact that $\QHE$ is q-IND-CPA.

Consider a slightly different strategy, $\advA''=(\advA''_1,\advA''_2)$, for the same experiment $\mathsf{PubK}_{\advA'',\QHE}^{{\sf cpa}}$. This strategy discards the \emph{second} message space, $\mathcal{M}_1$, and inputs the first, $\mathcal{M}_0$, into $\Xi_{\QHE}^{{\sf cpa},r}$, that is $\advA''_1=\Tr_{{\cal M}_1}(\advA_1(pk,\rho^{\mathcal{R}_{evk}}))$. The new adversary $\advA''$ outputs the complement of the output of $\advA$: $\advA''_2(\rho^{\cal CE})=\advA_2(\rho^{\cal CE})\oplus 1$.
We can then see that $\advA''$ correctly predicts $r$ with probability
\begin{equation}
\frac{1}{2}\Pr\left[\advA_2((\mathsf{Enc}(\ket{0}\bra{0}))\otimes \rho^{\cal E})=1\right]+\frac{1}{2}\Pr\left[\advA_2( \mathsf{Enc}_{pk}\otimes \mathbb{I}_{\cal E}(\rho^{\mathcal{M}_0\otimes \mathcal{E}}))=0\right]\leq \frac{1}{2}+\eta''(\kappa).
\label{eq:second}
\end{equation}
for some negligible function $\eta''$.
The addition of  \eqref{eq:first} and \eqref{eq:second}, gives:
\begin{equation*}
\frac{1}{2}+\frac{1}{2}\Pr\left[\advA_2( \mathsf{Enc}_{pk}\otimes \mathbb{I}_{\cal E}(\rho^{\mathcal{M}_1\otimes \mathcal{E}}))=1\right]+ \frac{1}{2}\Pr\left[\advA_2( \mathsf{Enc}_{pk}\otimes \mathbb{I}_{\cal E}(\rho^{\mathcal{M}_0\otimes \mathcal{E}}))=0\right]\leq  1+\eta'(\kappa)+\eta''(\kappa).
\end{equation*}
Since $\eta:=\eta'+\eta''$ is still negligible, we conclude that for some negligible function~$\eta$:
\begin{equation*}
\frac{1}{2}\Pr\left[\advA_2( \mathsf{Enc}_{pk}\otimes \mathbb{I}_{\cal E}(\rho^{\mathcal{M}_1\otimes \mathcal{E}}))=1\right]+ \frac{1}{2}\Pr\left[\advA_2( \mathsf{Enc}_{pk}\otimes \mathbb{I}_{\cal E}(\rho^{\mathcal{M}_0\otimes \mathcal{E}}))=0\right]\leq  \frac{1}{2}+\eta(\kappa). \iftoggle{crypto}{}{\qedhere}
\end{equation*}
\end{proof}

\begin{figure}[H]
\centering
\begin{tikzpicture}
\draw (-.5,2.04)--(3,2.04);
\draw (-.5,1.96)--(3,1.96);
\draw (-.5,0)--(4.75,0);
\draw (.5,.75)--(4.75,.75);
\draw (.5,1.25) -- (2,1.25);
\draw (-.5,1.29)--(.5,1.29);
\draw (-.5,1.21)--(.5,1.21);
	\draw (5.25,.665)--(5.75,.665);
	\draw (5.25,.585)--(5.75,.585);

\filldraw[fill=white] (-1,2.25) rectangle (-.5,-.25);
\node[rotate=90] at (-.75,1) {$\mathsf{KeyGen}(1^\kappa)$};

\node at (-.25,2.2) {$pk$};
\node at (-.25,1.45) {$pk$};
\node at (-.1,.2) {\small $\mathcal{R}_{evk}$};

\filldraw[fill=white] (.5,1.5) rectangle (1,-.25);
\node at (.75,.625) {$\advA_1$};

\node at (1.25,.2) {\small $\cal E$};
\node at (1.35,.95) {\small $\mathcal{M}_1$};
\node at (1.35,1.45) {\small $\mathcal{M}_0$};

\node at (2,1.25) {\meas};

\filldraw[fill=white] (2.75,.5) rectangle (4,2.25);
\node at (3.375,1.25) {$\Xi_{\QHE}^{{\sf cpa},r}$};

\node at (4.25,.95) {\small $\cal C$};

\filldraw[fill=white] (4.75,1.5) rectangle (5.25,-.25);
\node at (5,.625) {$\advA_2$};

\node at (6,.625) {$b'$};

\draw[dashed] (.3,1.7) rectangle (2.55,-.45);
\node at (1.425,-.7) {$\advA_1'$};

\draw[dashed] (4.55,1.7) rectangle (5.45,-.45);
\node at (5,-.7) {$\advA_2'$};

\end{tikzpicture}
\caption{The adversary $\advA'$ described in the proof of Thm.~\ref{th:equiv}.}\label{fig:equiv}
\end{figure}

\begin{theorem}[Equivalence of q-IND-CPA and q-IND-CPA-mult]\label{th:equiv-mult}\iftoggle{crypto}{\label{thm:equiv-IND-CPA}}{}
A quantum homomorphic encryption scheme is q-IND-CPA if and only if it is q-IND-CPA-mult.
\end{theorem}
\begin{proof}
It is trivial to see that if a scheme is q-IND-CPA-mult, then it is q-IND-CPA.

For the other direction, it is simple to adapt a similar classical proof (see for example \cite{KL08}) to the quantum setting. Suppose $\QHE$ is q-IND-CPA, so in particular, it is q-IND-CPA-2. Let $\advA=(\advA_1,\advA_2)$ be an adversary for the q-IND-CPA-mult experiment $\mathsf{PubK}_{\advA,\QHE}^{\textsf{cpa-mult}}$. We will construct an adversary, $\advA'$, for the q-IND-CPA-2 experiment $\mathsf{PubK}_{\advA',\QHE}^{\textsf{cpa-2}}$ from $\advA$.

For any $i\in \{0,\dots,t\}$, define $\Psi_i:D(\mathcal{M}_0^1\otimes \dots\otimes \mathcal{M}_0^t\otimes\mathcal{M}_1^1\otimes\dots\otimes\mathcal{M}_1^t\otimes \mathcal{E})\rightarrow D(\mathcal{C}^1\otimes\dots\otimes\mathcal{C}^t\otimes\mathcal{E})$ as the channel that applies $\mathsf{Enc}_{pk}$ to the systems $\mathcal{M}_0^1,\dots,\mathcal{M}_0^i,\mathcal{M}_1^{i+1},\dots,\mathcal{M}_1^t$, and traces out the systems $\mathcal{M}_1^1,\dots,\mathcal{M}_1^i,\mathcal{M}_0^{i+1},\dots,\mathcal{M}_0^t$.

Let $\advA=(\advA_1,\advA_2)$ be a $t$-message adversary for q-IND-CPA-mult. We define a q-IND-CPA-2 adversary $\advA'=(\advA'_1,\advA'_2)$ as follows:
\begin{description}
\item[$\advA_1'(pk,\rho^{\mathcal{R}_{evk}})$:] Run $\advA_1(pk,\rho^{\mathcal{R}_{evk}})$ to get $\rho\in D(\mathcal{M}_0^1\otimes \dots\otimes \mathcal{M}_0^t\otimes\mathcal{M}_1^1\otimes\dots\otimes\mathcal{M}_1^t\otimes \mathcal{E})$. Choose a random $i\in\{1,\dots,t\}$, and apply $\Xi_{\QHE}^{\textsf{cpa-2},r}$ to the system $\mathcal{M}_0^i,\mathcal{M}_1^i$.  For $j<i$, apply $\mathsf{Enc}_{pk}$ to $\mathcal{M}_0^j$ and label the output as $\mathcal{C}_j$. For $j>i$, apply $\mathsf{Enc}_{pk}$ to $\mathcal{M}_1^j$, and label the output as $\mathcal{C}_j$. Let $\mathcal{E}'=\(\bigotimes_{j=1: j\neq i}^t\mathcal{C}_j\)\otimes\mathcal {E} \otimes \mathbb{C}^{t+1}$. Record $i$ in the last register.

\item[$\advA_2':D(\mathcal{C}\otimes\mathcal{E}')\rightarrow \{0,1\}$:] Label the output of $\Xi_{\QHE}^{\textsf{cpa-2},r}$, $\cal C$, as $\mathcal{C}_i$. Apply $\advA_2$ to the system $\mathcal{C}_1\otimes\dots\otimes\mathcal{C}_t\otimes\mathcal{E}$, to get a bit $r'$. Output $r'$.
\end{description}
We now consider the success probability of $\advA'$ on $\mathsf{PubK}_{\advA',\QHE}^{\textsf{cpa-2}}$. Let $\rho=\advA_1(pk,\rho^{\mathcal{R}_{evk}})$. We first note that if $\advA'_1$ selects $i$, then if $r=0$, the state passed to $\advA_2'$ is $\Psi_i(\rho)$, but if $r=1$, the state passed to $\advA_2'$ is $\Psi_{i+1}(\rho)$.
So if $r=0$, then the success probability is:
\begin{equation*}
\Pr[\mathsf{PubK_{\advA',\QHE}^{cpa\text{-}2}}(\kappa)=1|r=0] = \sum_{i=1}^{t}\frac{1}{t}\Pr[\advA_2(\Psi_i(\rho))=0].
\end{equation*}
And if $r=1$, then the success probability is:
\begin{equation*}
\Pr[\mathsf{PubK_{\advA',\QHE}^{cpa\text{-}2}}(\kappa)=1|r=1] = \sum_{i=1}^t\frac{1}{t}\Pr[\advA_2(\Psi_{i+1}(\rho))=1]=\sum_{i=0}^{t-1}\frac{1}{t}\Pr[\advA_2(\Psi_i(\rho))=1].
\end{equation*}
From these two equations, we can compute:
\begin{align*}
&\Pr[\mathsf{PubK_{\advA',\QHE}^{cpa\text{-}2}}(\kappa)=1]\\
 &= \frac{1}{2}\Pr[\mathsf{PubK_{\advA',\QHE}^{cpa\text{-}2}}(\kappa)=1|r=0]+\frac{1}{2}\Pr[\mathsf{PubK_{\advA',\QHE}^{cpa\text{-}2}}(\kappa)=1|r=1]\\
 &= \frac{1}{2}\(\frac{1}{t}\Pr[\advA_2(\Psi_0(\rho))=1]+\sum_{i=1}^{t-1}\frac{1}{t}\(\Pr[\advA_2(\Psi_i(\rho))=0]+\Pr[\advA_2(\Psi_i(\rho))=1]\)+\frac{1}{t}\Pr[\advA_2(\Psi_t(\rho))=0]\)\\
&=\frac{1}{2}\(\frac{1}{t}\Pr[\advA_2(\Psi_0(\rho))=1]+\frac{t-1}{t}+\frac{1}{t}\Pr[\advA_2(\Psi_t(\rho))=0]\).\\
\end{align*}
Since $\QHE$ is assumed to be q-IND-CPA, there exists a negligible function $\eta'$ such that $\Pr[\mathsf{PubK_{\advA',\QHE}^{cpa\text{-}2}}(\kappa)=1]\leq \frac{1}{2}+\eta'(\kappa)$, so we can compute:
\begin{equation}
\frac{1}{t}\Pr[\advA_2(\Psi_0(\rho))=1]+\frac{t-1}{t}+\frac{1}{t}\Pr[\advA_2(\Psi_t(\rho))=0]  \leq 1+2 \eta'(\kappa).
\label{eq:equiv-mult-3}
\end{equation}
Note that $\Psi_0(\rho)=\(\(\mathsf{Enc}_{pk}\)^{\otimes t}\otimes \mathbb{I}_{\cal E}\)(\rho^{(\bigotimes_{j=1}^t\mathcal{M}_1^j)\otimes\mathcal{E}})$ and $\Psi_t(\rho)=\(\(\mathsf{Enc}_{pk}\)^{\otimes t}\otimes \mathbb{I}_{\cal E}\)(\rho^{(\bigotimes_{j=1}^t\mathcal{M}_0^j)\otimes\mathcal{E}})$, so
$\frac{1}{2}\Pr[\advA_2(\Psi_0(\rho))=1]+\frac{1}{2}\Pr[\advA_2(\Psi_t(\rho))=0]=\Pr[\mathsf{PubK_{\advA,\QHE}^{cpa\text{-}mult}}(\kappa)=1]$. From Equation \eqref{eq:equiv-mult-3}, we get:
\begin{equation*}
\Pr[\mathsf{PubK_{\advA,\QHE}^{cpa\text{-}mult}}(\kappa)=1] \leq  \frac{1}{2}+t\cdot \eta'(\kappa).
\end{equation*}
It must be the case that $t=O(\mathrm{poly}(\kappa))$, since $\advA$ is a QPT algorithm, so $t\cdot \eta'(\kappa)$ is negligible in $\kappa$. Setting $\eta=t\cdot \eta'$ completes the proof.
\end{proof}

\begin{figure}[h]
\centering
\begin{tikzpicture}[scale = 1.25]
\draw (-1.25,2.54)--(3,2.54);		\node at (-1,2.7) {$pk$};
\draw (-1.25,2.46)--(3,2.46);
\draw (-1.25,2.04)--(-.5,2.04);	\node at (-1,2.2) {$pk$};
\draw (-1.25,1.96)--(-.5,1.96);
\draw (-1.25,0.04)--(-.5,0.04);	\node at (-.875,.2) {\small $\mathcal{R}_{evk}$};
\draw (-1.25,-.04)--(-.5,-.04);

\draw (0,2)--(4,2);
\node at (.3,2.15) {\small ${\cal M}_0^1$};
\node at (3.4,2.15) {\small ${\cal C}_1$};
\draw (0,1.5)--(3,1.5);
\node at (.3,1.65) {\small ${\cal M}_0^2$};
\draw (0,1)--(1.5,1);
\node at (.3,1.15) {\small ${\cal M}_1^1$};
\draw (0,.5)--(3,.5);
\node at (.3,.65) {\small ${\cal M}_1^2$};
\draw (0,0)--(4,0);
\node at (.3,.15) {\small ${\cal E}$};
\draw (3,1) -- (4,1);
\node at (3.4,1.15) {\small ${\cal C}_2$};
\draw (4.5,1.04)--(5,1.04);
\draw (4.5,.96)--(5,.96);

\filldraw[fill=white] (-1.75,2.75) rectangle (-1.25,-.25);
\node[rotate=90] at (-1.5,1.25) {$\KeyGen$};

\filldraw[fill=white] (-.5,2.25) rectangle (0,-.25);
\node at (-.25,1) {$\advA_1$};

\node at (1.25,1) {\meas};

\filldraw[fill=white] (2.35,2.75) rectangle (3.15,1.8);
\node at (2.75,2) {$\mathsf{Enc}_{pk}$};

\filldraw[fill=white] (2.1, 1.65) rectangle (3.15, .35);
\node at (2.625,1) {$\Xi_{\QHE}^{\textsf{cpa-2},r}$};

\filldraw[fill=white] (4,2.25) rectangle (4.5,-.25);
\node at (4.25,1) {$\advA_2$};

\node at (5.2,1) {$b'$};
\end{tikzpicture}
\caption{The strategy described in the proof of Thm.~\ref{th:equiv-mult} for the case $t=2$, when $i=2$.}\label{fig:equiv-mult}
\end{figure}

\iftoggle{crypto}{
\paragraph{Security of Symmetric Key Schemes}
In order to define q-IND-CPA security in the symmetric-key setting,  we must equip the adversary with an encryption oracle $\Enc_{sk}(\cdot)$.
An \emph{adversary with access to an encryption oracle}, $\advA$ is a tuple of quantum channels $(\advA^{(1)}, \dots, \advA^{(q+1)})$, such that $\advA^{(1)}:D({\cal X})\rightarrow D(\mathcal{M}\otimes\mathcal{E})$ for some space $\cal X$, for $i=2,\dots,q$, $\advA^{(i)}:D(\mathcal{C}\otimes\mathcal{E})\rightarrow D(\mathcal{M}\otimes\mathcal{E})$, and $\advA^{(q+1)}:D(\mathcal{C}\otimes\mathcal{E})\rightarrow D(\mathcal{Y})$ for some space $\mathcal{Y}$. The interaction of the adversary and the encryption oracle is shown in Fig.~\ref{fig:adv-with-oracle}, and for the case of a bounded encryption scheme, in which the oracle also updates a counter, in Fig.~\ref{fig:adv-with-oracle-bounded}.

\begin{figure}[h]
\centering
\input{fig-adv-with-oracle.tex}
\end{figure}

\begin{figure}[h]
\centering
\input{fig-adv-with-oracle-bounded.tex}
\end{figure}

Just as in the public-key setting, we can define a quantum CPA indistinguishability experiment for the symmetric-key setting, $\mathsf{SymK}_{\advA,\QHE}^{\sf cpa}(\kappa)$. An adversary for $\mathsf{SymK}_{\advA,\QHE}^{\sf cpa}(\kappa)$ is a pair of adversaries with access to an encryption oracle $\advA=(\advA_1,\advA_2)=(\advA_1^{(1)},\dots,\advA_1^{(q+1)},\advA_2^{(1)},\dots,\advA_2^{(q'+1)})$ ($q$~is the number of oracle calls before the challenger is called, and $q'$ is the number of oracle calls after the challenger is called). The experiment $\mathsf{SymK}_{\advA,\QHE}^{\sf cpa}(\kappa)$ is defined below, and shown in Fig.~\ref{fig:symmetric-cpa}.\looseness=-1

\noindent\textbf{The quantum symmetric-key CPA indistinguishability experiment} $\mathsf{SymK^{cpa}_{\advA,\QHE}}(\kappa)$
\begin{enumerate}
\item $\KeyGen(1^\kappa)$ is run to obtain keys $(sk,\rho_{evk})$.
\item $\advA_1$ is given $\rho_{evk}$, and may make a polynomial number of calls to an encryption oracle $\QHE.\Enc_{sk}$ before outputting a quantum state in message space $\cal M$ and environment register $\cal E$.
\item A random bit $r\in\{0,1\}$ is chosen and $\Xi_{\QHE}^{{\sf cpa},r}$ is applied to the state in $\cal M$ (the output being a state in $\cal C$).
\item Adversary $\advA_2$ obtains the system $\mathcal{C}\otimes \mathcal{E}$ and may make a polynomial number of calls to an encryption oracle $\QHE.\Enc_{sk}$ before outputting a bit $r'$.
\item The output of the experiment is defined to be 1 if $r=r'$ and 0 otherwise. In case $r=r'$, we say that $\advA$ \emph{wins} the experiment.
\end{enumerate}

\begin{figure}[h]
\centering
\begin{tikzpicture}
\node at (0,0){
\begin{tikzpicture}
\draw (0,2.54)--(5.5,2.54);
\draw (0,2.46)--(5.5,2.46);
\draw (0,.75)--(5,.75);
		\draw (5,.415)--(5.5,.415);
		\draw (5,.335)--(5.5,.335);
\draw (1,0)--(5,0);

\filldraw[fill=white] (-.5,.5) rectangle (0,2.75);
\node[rotate=90,align=center] at (-.25,1.625) {\small ${\sf QHE.KeyGen}$};

\node at (.25,2.7) {$sk$};
\node at (.4,.95) {\small $\mathcal{R}_{evk}$};

\filldraw[fill=white] (1,1.25) rectangle (1.5,2.75);
\node[rotate=90] at (1.25,2) {\small $\QHE.\Enc$};
\draw[->] (1.1,1)--(1.1,1.25);
\draw[->] (1.4,1.25)--(1.4,1);
\filldraw[fill=white] (1,-.25) rectangle (1.5,1);
\node at (1.25,.375) {$\advA_1$};

\node at (1.75,2.7) {$sk$};
\node at (1.75,.95) {\small $\cal M$};
\node at (1.75,.2) {\small $\cal E$};

\filldraw[fill=white] (2.5,.5) rectangle (3.5,2.75);
\node at (3,1.625) {$\Xi_{\QHE}^{\textsf{cpa},r}$};

\node at (3.75,2.7) {$sk$};
\node at (3.75,.95) {\small $\cal C$};

\filldraw[fill=white] (4.5,1.25) rectangle (5,2.75);
\node[rotate=90] at (4.75,2) {\small $\QHE.\Enc$};
\draw[->] (4.6,1)--(4.6,1.25);
\draw[->] (4.9,1.25)--(4.9,1);
\filldraw[fill=white] (4.5,-.25) rectangle (5,1);
\node at (4.75,.375) {$\advA_2$};

\node at (5.25,2.7) {$sk$};

\node at (5.75,.375) {$r'$};

\end{tikzpicture}
};

\node at (8,0){
\begin{tikzpicture}
\draw (0,2.54)--(5.5,2.54);
\draw (0,2.46)--(5.5,2.46);
\draw (.6,1.54)--(5.5,1.54);
\draw (.6,1.46)--(5.5,1.46);
\draw (0,.75)--(5,.75);
		\draw (5,.415)--(5.5,.415);
		\draw (5,.335)--(5.5,.335);
\draw (1,0)--(5,0);

\filldraw[fill=white] (-.5,.5) rectangle (0,2.75);
\node[rotate=90,align=center] at (-.25,1.625) {\small ${\sf QHE.KeyGen}$};

\node at (.25,2.7) {$sk$};
\node at (.5,1.5) {$1$};
\node at (.4,.95) {\small $\mathcal{R}_{evk}$};

\filldraw[fill=white] (1,1.25) rectangle (1.5,2.75);
\node[rotate=90] at (1.25,2) {\small $\QHE.\Enc$};
\draw[->] (1.1,1)--(1.1,1.25);
\draw[->] (1.4,1.25)--(1.4,1);
\filldraw[fill=white] (1,-.25) rectangle (1.5,1);
\node at (1.25,.375) {$\advA_1$};

\node at (1.75,2.7) {$sk$};
\node at (2,1.7) {\small$q+1$};
\node at (1.75,.95) {\small $\cal M$};
\node at (1.75,.2) {\small $\cal E$};

\filldraw[fill=white] (2.5,.5) rectangle (3.5,2.75);
\node at (3,1.625) {$\Xi_{\QHE}^{\textsf{cpa},r}$};

\node at (3.75,2.7) {$sk$};
\node at (4,1.7) {\small$q+2$};
\node at (3.75,.95) {\small $\cal C$};

\filldraw[fill=white] (4.5,1.25) rectangle (5,2.75);
\node[rotate=90] at (4.75,2) {\small $\QHE.\Enc$};
\draw[->] (4.6,1)--(4.6,1.25);
\draw[->] (4.9,1.25)--(4.9,1);
\filldraw[fill=white] (4.5,-.25) rectangle (5,1);
\node at (4.75,.375) {$\advA_2$};

\node at (5.25,2.7) {$sk$};
\node at (5.85,1.7) {\small$q'+q+2$};

\node at (5.75,.375) {$r'$};

\end{tikzpicture}
};
\end{tikzpicture}
\caption{The quantum CPA experiment for symmetric-key systems (left) and bounded symmetric-key systems (right).
}\label{fig:symmetric-cpa}
\end{figure}

\begin{definition}[Quantum Indistinguishability under Chosen Plaintext Attack (q-IND-CPA) for Symmetric Key Schemes]
\label{def:q-ind-cpa-Symmetric}
A symmetric-key quantum homomorphic encryption scheme $\QHE$ is q-IND-CPA secure if for all quantum polynomial-time adversaries with oracle access, {$\advA=(\advA_1^{(1)},\dots,\advA_1^{(q+1)},\advA_2^{(1)},\dots,\advA_2^{(q'+1)})$}, there exists a negligible function $\eta$ such that: {
$$\Pr[\mathsf{SymK_{\advA,\QHE}^{cpa}}(\kappa)=1]\leq \frac{1}{2}+\eta(\kappa).$$}
\end{definition}

Similar to the case of public-key encryption (Sec.~\ref{sec:security-QHE}), it is straightforward to give the seemingly stronger variant of q-IND-CPA, \emph{q-IND-CPA-mult}, which is defined identically to the public-key case (Def.~\ref{def:q-IND-CPA-mult}) but with an adversary having access to an encryption oracle. However, just as in the public-key case, it turns out that these definitions are equivalent.

\begin{theorem}[Equivalence of q-IND-CPA and q-IND-CPA-mult in symmetric-key schemes] \label{thm:equiv-IND-CPA-sym}
Let $\QHE$ be a symmetric-key quantum homomorphic encryption scheme. Then $\QHE$ is q-IND-CPA if and only if $\QHE$ is q-IND-CPA-mult.
\end{theorem}
\noindent The proof of Thm.~\ref{thm:equiv-IND-CPA-sym} is virtually identical to that of Thm.~\ref{thm:equiv-IND-CPA}; the only difference is that adversaries now have access to an encryption oracle, but all techniques still hold, since an adversary with access to an encryption oracle can simulate another adversary with access to an encryption oracle.
}{}

\section{Key Update Rules for Stabilizer Elements}
\label{appendix:Key-update-rules-stabilizer}
We review here the key update rules for performing stabilizer/Clifford operators on quantum data encrypted with the quantum one-time pad~\cite{Got98}. \iftoggle{crypto}{We use the following set of gates to (redundantly) generate the Clifford group:
$$\xgate = \left[\begin{array}{cc} 0 & 1\\ 1 & 0\end{array}\right],
\,\zgate = \left[\begin{array}{cc} 1 & 0\\ 0 & -1\end{array}\right],
\,\pgate = \left[\begin{array}{cc} 1 & 0\\ 0 & i\end{array}\right],
\,\tgate = \left[\begin{array}{cc} 1 & 0\\ 0 & e^{i\pi/4}\end{array}\right],
\,\hgate = \frac{1}{\sqrt{2}}\left[\begin{array}{cc}1 & 1\\1 & -1\end{array}\right],
\,\cnot = \left[\begin{array}{cccc} 1 & 0 & 0 & 0\\ 0 & 1 & 0 & 0\\ 0 & 0 & 0 & 1\\ 0 & 0 & 1 & 0\end{array}\right].$$

These elements, along with single-qubit measurement and qubit preparation generate the class of stabilizer circuits. The key-update rules follow.

}{}

\begin{figure}[H]
\centering
\begin{tikzpicture}
\node at (-.2,0) {$\xgate^{f_{a,i}} \zgate^{f_{b,i}} \ket{\psi }$};

\draw (1,0)--(2,0);

\draw (2,0.04)--(3,0.04);
\draw (2,-0.04)--(3,-0.04);

\node at (2,0) {\meas};

\node at (3.25,0) {$c$};
\node at (1.25,-.25) {\small ${\cal X}_i$};

\node at (5,0) {$f_{a,i}\leftarrow f_{a,i}$};
\end{tikzpicture}
\caption{\label{fig:measurement} Protocol for measurement on the \th{i} wire:
Simply perform the measurement. The resulting bit, $c$, can be decrypted by applying $\xgate^{f_{a,i}}$ (The key $f_{b,i}$ is no longer relevant).}
\end{figure}

\begin{figure}[H]
\centering
\begin{tikzpicture}

\node at (-.4,0) {$\ket{0}$};
\draw (0,0) -- (1.5,0);
\node at (2.2,0) {$\xgate^0\zgate^0\ket{0}$};
\node at (1.25,-.25) {\small ${\cal X}_i$};

\node at (5,0) {$f_{a,i}\leftarrow 0,\quad f_{b,i}\leftarrow 0$};

\end{tikzpicture}
 \caption{\label{fig:aux-prep} Protocol for auxiliary qubit
preparation on a new wire, $i$: Initialize a new wire labelled ${\cal X}_i$ and new key-polynomials $f_{i,a}=f_{b,i}=0$.}
\end{figure}

\begin{figure}[H]
\centering
\begin{tikzpicture}
\node at (-1.15,0) {$\xgate^{f_{a,i}}\zgate^{f_{b,i}}\ket{\psi}$};
\draw (0,0) -- (2,0);
\node at (1.75,-.25) {\small ${\cal X}_i$};
\filldraw[fill=white] (.75,.25) rectangle (1.25,-.25);
\node at (1,0) {$\xgate$};
\node at (3.3,0) {$\xgate^{f_{a,i}}\zgate^{f_{b,i}}\xgate\ket{\psi}$};

\node at (7,0) {$f_{a,i}\leftarrow f_{a,i},\quad f_{b,i}\leftarrow f_{b,i}$};
\end{tikzpicture}
 \caption{\label{fig:X-gate} Protocol for an \xgate-gate on the \th{i} wire: Simply apply the $\xgate$-gate.}
\end{figure}

\begin{figure}[H]
\centering
\begin{tikzpicture}
\node at (-1.15,0) {$\xgate^{f_{a,i}}\zgate^{f_{b,i}}\ket{\psi}$};
\draw (0,0) -- (2,0);
\node at (1.75,-.25) {\small ${\cal X}_i$};
\filldraw[fill=white] (.75,.25) rectangle (1.25,-.25);
\node at (1,0) {$\zgate$};
\node at (3.3,0) {$\xgate^{f_{a,i}}\zgate^{f_{b,i}}\zgate\ket{\psi}$};

\node at (7,0) {$f_{a,i}\leftarrow f_{a,i},\quad f_{b,i}\leftarrow f_{b,i}$};
\end{tikzpicture}
\caption{\label{fig:Z-gate} Protocol for a \zgate-gate on the \th{i} wire: Simply apply the $\zgate$-gate.}
\end{figure}

\begin{figure}[H]
\centering
\begin{tikzpicture}
\node at (-1.15,0) {$\xgate^{f_{a,i}}\zgate^{f_{b,i}}\ket{\psi}$};
\draw (0,0) -- (2.5,0);
\node at (2.25,-.25) {\small ${\cal X}_i$};
\filldraw[fill=white] (1,.25) rectangle (1.5,-.25);
\node at (1.25,0) {$\hgate$};
\node at (3.8,0) {$\xgate^{f_{b,i}}\zgate^{f_{a,i}}\hgate\ket{\psi}$};

\node at (7.5,0) {$f_{a,i}\leftarrow f_{b,i},\quad f_{b,i}\leftarrow f_{a,i}$};
\end{tikzpicture}
 \caption{\label{fig:H-gate} Protocol for an \hgate-gate on the \th{i} wire: Apply the gate and swap the key-polynomials.}
\end{figure}

\begin{figure}[H]
\centering
\begin{tikzpicture}
\node at (-1.15,0) {$\xgate^{f_{a,i}}\zgate^{f_{b,i}}\ket{\psi}$};
\draw (0,0) -- (2.5,0);
\node at (2.25,-.25) {\small ${\cal X}_i$};
\filldraw[fill=white] (1,.25) rectangle (1.5,-.25);
\node at (1.25,0) {$\pgate$};
\node at (4.15,0) {$\xgate^{f_{a,i}}\zgate^{f_{b,i}\oplus f_{a,i}}\pgate\ket{\psi}$};

\node at (8.5,0) {$f_{a,i}\leftarrow f_{a,i},\quad f_{b,i}\leftarrow f_{b,i}\oplus f_{a,i}$};

\end{tikzpicture}
 \caption{\label{fig:P-gate} Protocol for a \pgate-gate on the \th{i} wire: Apply the gate and update $f_{b,i}$.}
\end{figure}

\begin{figure}[H]
\centering
\begin{tikzpicture}
\node at (-1.4,0) {$(\xgate^{f_{a,i}} \zgate^{f_{b,i}} \otimes \xgate^{f_{a,j}} \zgate^{f_{b,j}})\ket{\psi}$};

\node[yscale=2] at (1,0) {$\Big\{$};

\draw (1.25,.4)--(3,.4);
\draw (1.25,-.4)--(3,-.4);

\node at (2.125,.4) {\cntrl};
\draw (2.125,.4)--(2.125,-.4);
\node at (2.125,-.4) {\target};

\node[yscale=2] at (3.25,0) {$\Big\}$};

\node at (7,0) {$(\xgate^{f_{a,i}}  \zgate^{f_{b,i} \oplus f_{b,j}} \otimes \xgate^{f_{a,i} \oplus f_{a,j}} \zgate^{f_{b,j}})\cnot(\ket{\psi})$};

\node at (3.5,-1.25) {$f_{a,i}\leftarrow f_{a,i},\quad f_{b,i}\leftarrow f_{b,i}\oplus f_{b,j},\quad f_{a,j}\leftarrow f_{a,i}\oplus f_{a,j},\quad f_{b,j}\leftarrow f_{b,j}$};

\node at (1.5,.15) {\small ${\cal X}_i$};
\node at (1.5,-.65) {\small ${\cal X}_j$};
\end{tikzpicture}
\caption{\label{fig:CNOT-gate}Protocol for a \cnot-gate with control wire $i$ and target wire $j$: Apply the gate and update  $f_{b,i}$ and $f_{a,j}$.}
\end{figure}

We remark that an alternative gadget for the $\xgate$ is to update the $\xgate$-key as $f_{a,i}\rightarrow f_{a,i}\oplus 1$, rather than applying $\xgate$ to the quantum state. A similar alternative holds for the $\zgate$-gadget. However, these are the only two gates for which a key update is sufficient to affect the gate. Since we are actually carrying out quantum computations on encrypted quantum data --- in contrast to merely simulating a quantum computation --- all gates except the Pauli gates require actual quantum operations to be applied during evaluation.

\iftoggle{crypto}{
\section{Analysis of $\CL$}
\label{app:analysis-cl}

In this section, we analyze the properties of $\CL$, to prove Theorem \ref{thm:main-Clifford}.

\begin{theorem}
\label{thm:correctness:Clifford}
Let $\mathscr{S}$ be the class of Clifford circuits. Then $\CL$ is  $\mathscr{S}$-homomorphic.
\end{theorem}
\begin{proof}
This follows from the circuits in App.~\ref{appendix:Key-update-rules-stabilizer}, as well as the homomorphic property of $\mathsf{HE}$. In particular, since the decrypted values of the ciphertexts are correct (except with exponentially small probability), then {Eq.~}\eqref{eqn:C-homomorphism} is satisfied.
\end{proof}

\begin{theorem}
\label{thm:compactness:Clifford}
$\CL$ is compact.
\end{theorem}
\begin{proof}
Let $p$ be a polynomial such that the complexity of applying $\HE.\Dec$ to the output of $\HE.\Eval$ is at most $p(\kappa)$ --- such a polynomial exists by the compactness of $\HE$. Then decrypting a single qubit of the output of $\CL.\Eval$ has complexity at most $2p(\kappa)+2$, since we must decrypt two keys $a$ and $b$ and then apply $\xgate^a$ and $\zgate^b$, so $\CL$ is also compact.
\end{proof}

\begin{theorem}
\label{thm:security:Clifford}
{Assuming a classical fully homomorphic encryption scheme $\HE$ that is q-IND-CPA secure, the quantum homomorphic scheme $\CL$  is q-IND-CPA secure.}
\end{theorem}
\begin{proof}
The main part of this proof will be to show that the classical ciphertexts $\HE.\Enc_{pk}(a)$ and $\HE.\Enc_{pk}(b)$ give at most a negligible advantage. We will then see that without these classical ciphertexts, the quantum CPA Indistinguishability experiment is independent of $r$ from the  perspective of the adversary.

Let $\CL'$ be the quantum homomorphic encryption scheme with $\CL'.\KeyGen=\CL.\KeyGen$, $\CL'.\Eval=\CL.\Eval$, $\CL'.\Dec=\CL.\Dec$, and
\begin{align*}
\CL'.\Enc_{pk}(\rho) &=\sum_{a,b\in\{0,1\}}\frac{1}{4}\rho(\HE.\Enc_{pk}(0),\HE.\Enc_{pk}(0))\otimes (\xgate^a\zgate^b\rho\zgate^b\xgate^a)\\
&=\rho(\HE.\Enc_{pk}(0),\HE.\Enc_{pk}(0))\otimes \frac{1}{2}\mathbb{I}_2.
\end{align*}

Let $\advA=(\advA_1,\advA_2)$ be an adversary for $\mathsf{PubK}_{\advA,\CL}^{\textsf{cpa}}(\kappa)$.
We will define an adversary $\advA'=(\advA_1',\advA_2')$ for $\mathsf{PubK_{\advA',\HE}^{cpa\text{-}mult}}(\kappa)$. Essentially, $\advA'$ will simulate $\mathsf{PubK}_{\advA,\CL}^{\textsf{cpa-mult}}(\kappa)$, except that when it simulates $\Xi_{\CL}^{{\sf cpa},r}$, it will use $\Xi_{\HE}^{\textsf{cpa-mult},s}$ in place of $\HE.\Enc$, so that it will actually be running either $\Xi_{\CL}^{{\sf cpa},r}$ (if $s=1$) or $\Xi_{\CL'}^{{\sf cpa},r}$ (if $s=0$) (see Fig.~\ref{fig:cl-security}).

\begin{figure}[t]
\centering
\input{fig-cl-security.tex}
\end{figure}

\begin{description}
\item[$\advA_1'(pk,evk)$:] Run $\advA_1(pk,evk)$ to get a state $\rho^{\cal ME}$. Choose a uniform random bit $r$. If $r=0$, discard the $\mathcal{M}$ subsystem and replace it with the state $\ket{0}\bra{0}$. Choose uniform random bits $a$ and $b$, and apply $\QEnc_{a,b}$, the quantum one-time pad, to $\mathcal{M}$, relabelling the resulting system by $\cal X$. Input $(a,b)$ and $(0,0)$ to $\Xi_{\HE}^{\textsf{cpa-mult},s}$.
\item[$\advA_2'$:] Run $\advA_2$ to get a bit $r'$. Output $1$ if $r=r'$ and $0$ otherwise.
\end{description}
We now compute the probability that $\advA'$ correctly guesses $s$, which we know must be at most $\frac{1}{2}+\eta(\kappa)$ for some negligible function, since $\HE$ is q-IND-CPA. If $s=1$, then $\advA'$ is simulating $\mathsf{PubK}_{\advA,\CL}^{\mathsf{cpa}}$, so the probability that $r'=r$ (and thus that $s'=1=s$) is $\Pr[\mathsf{PubK_{\advA,\CL}^{cpa}}(\kappa)=1]$.

On the other hand, if $s=0$, then $\advA_2$ gets encryptions of $0$ rather than $\HE.\Enc(a),\HE.\Enc(b)$,~so $\advA'$ is simulating $\mathsf{PubK}_{\advA,\CL'}^{\mathsf{cpa}}$, so the probability that $r\neq r'$, and thus $s'=0=s$ is \mbox{$\Pr[\mathsf{PubK_{\advA,\CL'}^{cpa}}(\kappa)=0]$}.\looseness = -1

Then since the total probability that $s=s'$ is at most $\frac{1}{2}+\eta(\kappa)$, we have:
\begin{align}
\frac{1}{2}\Pr[\mathsf{PubK_{\advA,\CL}^{cpa}}(\kappa)=1]+\frac{1}{2}\Pr[\mathsf{PubK_{\advA,\CL'}^{cpa}}(\kappa)=0] &\leq \frac{1}{2}+\eta(\kappa)\nonumber\\
\Pr[\mathsf{PubK_{\advA,\CL}^{cpa}}(\kappa)=1]+1-\Pr[\mathsf{PubK_{\advA,\CL'}^{cpa}}(\kappa)=1] &\leq 1+2\eta(\kappa)\nonumber\\
\Pr[\mathsf{PubK_{\advA,\CL}^{cpa}}(\kappa)=1]-\Pr[\mathsf{PubK_{\advA,\CL'}^{cpa}}(\kappa)=1] &\leq 2\eta(\kappa).\label{eq:negl}
\end{align}

We complete the proof by noting that when $s=0$, since $c=(\HE.\Enc_{pk}(0),\HE.\Enc_{pk}(0))$, it is independent of $a,b$ (see
 Fig.~\ref{fig:cl-security-s0}).
\begin{figure}[H]
\centering
\input{fig-cl-security-s0.tex}
\end{figure}

\noindent Then from the perspective of $\advA_2$, since $a,b$ is uniform random, the system $\cal X$ just contains the completely mixed state $\maxmix$ (see Fig.~\ref{fig:cl-security-last}).

\begin{figure}[H]
\centering
\input{fig-cl-security-last.tex}
\end{figure}

\noindent Since the experiment $\mathsf{PubK_{\advA,\CL'}^{cpa}}$ is independent of $r$ from the perspective of $\advA$, $\Pr[\mathsf{PubK_{advA,\CL'}^{cpa}}(\kappa)=1]=\frac{1}{2}$. Combining this with Equation \eqref{eq:negl}, we get
\begin{equation*}
\Pr[\mathsf{PubK_{\advA,\CL}^{cpa}}(\kappa)=1]\leq \frac{1}{2}+2\eta(\kappa),
\end{equation*}
which completes the proof, since $2\eta$ is still a negligible function.
\end{proof}

}{}

\section{Correctness of the $\tgate$-gate Gadget}
\label{appendix:Correctness-T-gate}

We give below a step-by-step proof of the correctness of the
\tgate-gate protocol\iftoggle{crypto}{s}{ from Fig.~\ref{fig:T-gate-EPR}}. The basic
building block is the circuit identity for an \xgate-teleportation
from~\cite{ZLC00}, which we re-derive here. Also of relevance to
this work are the techniques developed by Childs, Leung, and
Nielsen~\cite{CLN05} to manipulate circuits that produce an output
that is correct \emph{up to known Pauli corrections}.

We will make use of the following identities which all hold up to an
irrelevant global phase: $ \xgate \zgate = \zgate \xgate$, $\pgate
\zgate = \zgate \pgate$, $\pgate \xgate = \xgate \zgate \pgate$,
$\tgate \zgate = \zgate \tgate$, $\tgate \xgate = \xgate \zgate
\pgate \tgate$,  $\pgate^2  = \zgate$ and $\pgate^{a\oplus b} =
\zgate^{a\cdot b} \pgate^{a + b}$ (for $a, b \in \{0,1\}$).

\begin{enumerate}
\item  Our first circuit identity (Fig.~\ref{fig:circuit-proof:1}) swaps a qubit $\ket{\psi}$ with the
state $\ket{+}$ and is easy to verify.

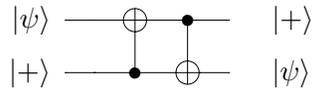
\begin{figure}[H]
\centerline{
 \Qcircuit @C=1em @R=1em {
\lstick{\ket{\psi}} & \qw &\targ      &  \ctrl{1} & \qw   &  \rstick{\ket{+}}     \\
\lstick{\ket{+}}    & \qw &\ctrl{-1}  &   \targ   & \qw &
\rstick{\ket{\psi}}&   }
 }
 \caption{ Circuit identity (easy to verify).}
 \label{fig:circuit-proof:1}
\end{figure}

\item We can measure the top qubit in the above circuit
and classically control the output correction (Fig.~\ref{fig:circuit-proof:2}). We have thus
re-derived the circuit corresponding to the ``\xgate-teleportation''
of~\cite{ZLC00}.

\begin{figure}[H]
 \centerline{
 \Qcircuit @C=1em @R=1em {
\lstick{\ket{\psi}} & \qw &\targ      &  \meter & \cw   &  \rstick{c}      \\
\lstick{\ket{+}}    & \qw &\ctrl{-1}  &   \qw  & \qw &
\rstick{\xgate^c\ket{\psi}}&   }
 }
 \caption{$\xgate$-teleportation}
 \label{fig:circuit-proof:2}
\end{figure}

\item Let the input be $\tgate  \xgate^a \zgate^b \ket{\psi}$, and add two gates on the auxiliary wire, $\pgate^a$ and $\zgate^k$ (Fig.~\ref{fig:circuit-proof:3}). Using the fact that  $\pgate$ and $\zgate$ commute with
control, and applying identities given above, we get as output (using $\tgate\xgate = \xgate\zgate\pgate\tgate$):
\begin{equation}
\pgate^a\zgate^k\xgate^c \tgate \xgate^a \zgate^b  \ket{\psi}  = \pgate^a\zgate^k\xgate^c  \xgate^a \zgate^{a \oplus b} \pgate^a \tgate   \ket{\psi}.
\end{equation}

This is equal to (simplifying, then pushing the first $\pgate$ to the end):
\begin{eqnarray*}
 \pgate^a \xgate^{a \oplus c} \zgate^{a \oplus b \oplus k} \pgate^a \tgate \ket{\psi} &=&  \xgate^{a\oplus c}\zgate^{(a\oplus c)a}\pgate^a\zgate^{a\oplus b\oplus k}\pgate^a\tgate\ket{\psi}\\
&=& \xgate^{a\oplus c}\zgate^{a^2\oplus c\cdot a\oplus a\oplus b\oplus k}\pgate^{2a}\tgate\ket{\psi}\\
&=& \xgate^{a\oplus c}\zgate^{c\cdot a\oplus a\oplus b\oplus k}\tgate\ket{\psi}\quad\mbox{since $a^2=a$ and $\pgate^2=\zgate$}.\\
 \end{eqnarray*}

\begin{figure}[H]
\centerline{
 \Qcircuit @C=1em @R=1em {
\lstick{\xgate^a \zgate^b\ket{\psi}} & \qw & \gate{\tgate} & \qw &\targ     &  \meter &  \cw  &  \rstick{c}      \\
\lstick{\ket{+}}    & \qw& \gate{\pgate^a} & \gate{\zgate^k} & \ctrl{-1}
 & \qw  &\qw & \rstick{\xgate^{a \oplus c} \zgate^{a \oplus b \oplus k \oplus a\cdot c } \tgate \ket{\psi}}& }
 }
 \caption{Final circuit for $\tgate$ gate.}
 \label{fig:circuit-proof:3}
\end{figure}
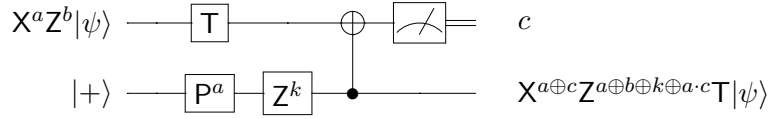

\end{enumerate}

\iftoggle{crypto}{
\section{Analysis of $\EPR$}
\label{app:analysis-epr}
We now analyse the various properties of $\EPR$, in order to prove Thm.~\ref{thm:main-EPR}.
Since the scheme $\EPR$ uses the same $\KeyGen$ and $\Enc$ procedures as~$\CL$, the following theorem follows from Thm.~\ref{thm:security:Clifford}.

\begin{theorem}
\label{thm:security:QHE}
If $\HE$ is q-IND-CPA secure, then $\EPR$ is q-IND-CPA secure.
\end{theorem}

The next theorem shows the  homomorphic property for all circuits (recall that this property is independent of compactness).

\begin{theorem}
\label{thm:correctness:QHE}
Let $\classS$ be the class of all quantum circuits. Then $\EPR$ is  $\classS$-homomorphic.
\end{theorem}
The proof follows from the circuits in App.~\ref{appendix:Key-update-rules-stabilizer}, Fig.~\ref{fig:T-gadget-EPR}, as well as the homomorphic property of~$\mathsf{HE}$.

Since the complexity of the decryption procedure depends on $R$, the number of $\tgate$-gates in the circuit, it is clear that the scheme $\EPR$ is not compact. However, by analysing the circuit's dependence on $R$, we can see that for a very large class of quantum circuits, $\EPR$ is non-trivially quasi-compact. The following theorem is immediate from the decryption procedure.

\begin{theorem}\label{cor:EPR-dec-complexity}
Let $p$ be a polynomial such that $\HE.\Dec$ has complexity $O(p(\kappa))$. Then {the decryption procedure} $\EPR.\Dec$ has complexity $O(R^2+Rp(\kappa)+mp(\kappa)+mR)$.
\end{theorem}

\noindent Thus, the dependence of the complexity of $\EPR.\Dec$ on the evaluated circuit $\mathsf{C}$ is $R^2$:

\begin{corollary}
Let $R(\mathsf{C})$ denote the number of $\tgate$-gates in a circuit $\mathsf{C}$. Then $\EPR$ is $R^2$-quasi-compact.
\end{corollary}

This beats the compactness of the trivial scheme for all circuits $\mathsf{C}$ such that the number of $\tgate$-gates is less than the squareroot of the number of gates; that is $R\ll \sqrt{G}$.

}{}

\iftoggle{crypto}{
\section{Analysis of $\AUX$}
\label{app:analysis-aux}

We now analyse the various properties of $\AUX$ in order to prove Thm.~\ref{thm:main-AUX}. Consider a layered quantum circuit $\mathsf{C}$ with $L$ layers of $\tgate$-gates.
To simplify the analysis, we assume that the ordering of gates $\mathsf{g}_1,\dots,\mathsf{g}_G$ has the property that if $\mathsf{g}_i$ is a $\tgate$-gate in level $\ell$, and $\mathsf{g}_j$ is a $\tgate$-gate in level $\ell+1$, then $i<j$; that is, we completely evaluate level $\ell$ before we begin to evaluate level $\ell+1$.

\begin{lemma}\label{lem:terms}
Let $f_{a,i}$ be a key-polynomial going into the \th{\ell} layer of $\tgate$-gates. Then $f_{a,i}$ is a sum of terms in $T_\ell$.
\end{lemma}
\begin{proof}
We prove this statement by induction on $\ell$. Before any gates have been applied, the key-polynomial are $f_{a,i}=a_i$ and $f_{b,i}=b_i$ for $i=1,\dots,n$. We can easily see from the update rules that applying Clifford gates results in keys of the form $f$ or $f+f'$, where $f$ and $f'$ were previous keys. Thus, after a Clifford circuit has been applied, all key-polynomial are sums of terms from $\{a_1,\dots,a_n,b_1,\dots,b_n\}=T_1$.

Let $f_{a,1},\dots,f_{a,n},f_{b,1},\dots,f_{b,n}$ be the key-polynomials going into the \th{\ell} layer, and suppose they are sums of terms in $T_{\ell}$. Let $f'_{a,1},\dots,f'_{a,n},f'_{b,1},\dots,f'_{b,n}$ be the key-polynomials right after the \th{\ell} layer of $\tgate$-gates has been applied. If no $\tgate$ is applied on the \th{i} wire, then $f_{a,i}'=f_{a,i}$ and $f_{b,i}'=f_{b,i}$, so $f_{a,i}',f_{b,i}'$ are both sums of terms in $T_{\ell}\subset T_{\ell+1}$. Suppose on the other hand that we apply a $\tgate$-gate to the \th{i} wire at level $\ell$.
From the $\tgate$-gadget (Fig.~\ref{fig:T-gate-AUX}), we  see that after applying a $\tgate$ to the \th{i} wire, we have new keys $f_{a,i}'=f_{a,i}\oplus c$ for a known constant $c$, so $f_{a,i}'$ is a sum of terms in $T_\ell\subset T_{\ell+1}$; and $f_{b,i}'=(1\oplus c)f_{a,i}\oplus f_{b,i}\oplus k$, where $k$ is the auxiliary state key of the auxiliary state used to implement the gadget. If $f_{a,i}=t_1\oplus \dots\oplus t_r$, for $t_1,\dots,t_r\in T_\ell$, then we construct $\ket{+_{f_{a,i},k}}$ from auxiliary states $\ket{+_{t_1,k_1}},\dots,\ket{+_{t_r,k_r}}$ for some $k_1,\dots,k_r\in\{k^{(\ell)}_{q,i}\}_{q=1}^{|T_{\ell}|}\subset T_{\ell+1}$, so we have
$k=\bigoplus_{j=1}^rk_j\oplus \bigoplus_{j=2}^rc_jt_j\oplus \bigoplus_{j=1}^r\bigoplus_{j'=1}^{j-1}t_jt_{j'}$ for known $c_2,\dots,c_r$, which is the sum of terms in $T_{\ell+1}$, since $t_1,\dots,t_r\in T_{\ell}$. Thus, $f_{b,i}'$ is the sum of terms in $T_{\ell+1}$.

Thus, after applying the \th{\ell} layer of $\tgate$-gates, all key-polynomials are sums of terms from $T_{\ell+1}$. To complete the proof, we simply observe again that Clifford circuits act additively on the keys, and so do not introduce new terms, so just before the \th{(\ell+1)} layer of $\tgate$-gates, the key-polynomials are still sums of terms in $T_{\ell+1}$.
\end{proof}

\noindent The bottleneck in this scheme is the number of auxiliary states required:
\begin{lemma}
The number of auxiliary states output by $\AUX.\KeyGen(1^\kappa,1^n)$ grows in $n$ as $O(n^{2^{L-1}+1})$.
\end{lemma}
\begin{proof}
The number of qubits encoded in $\sigma_{aux}^{a,b,k}$ is
$$|k^{(1)}|+|k^{(2)}|+\dots+|k^{(L)}|=n|T_1|+n|T_2|+\dots+n|T_L|=n\sum_{\ell=1}^{L}|T_\ell|.$$
From the definition of $T_{\ell}$, we  see that:
$$|T_1|=2n,\quad\mbox{and for $\ell>1$,}\quad |T_\ell|=|T_{\ell-1}|+\binom{|T_{\ell-1}|}{2}+n|T_{\ell-1}|.$$
So certainly for all $\ell>1$, $|T_{\ell}|\leq c|T_{\ell-1}|^2$ for some constant $c$, and thus $|T_{\ell}|\leq c^{\ell-1}(2n)^{2^{\ell-1}}\in O(n^{2^{\ell-1}})$. Thus $n\sum_{\ell=1}^L|T_{\ell}|\in O(n^{2^{L-1}+1})$.
\end{proof}
\noindent We thus have the following theorem:
\begin{theorem}
Let $\classS_n$ be the class of all quantum circuits on $n$ wires with $\tgate$-depth at most~$L$, and let $\classS=\{\classS_n\}_{n\in\mathbb{N}}$. Then $\AUX$ is $\classS$-homomorphic and compact.
\end{theorem}

We now consider the security of the scheme.

\begin{theorem}
\label{thm:security:AUX}
If $\HE$ is q-IND-CPA secure, then $\AUX$ is q-IND-CPA secure.
\end{theorem}

We will prove Thm.~\ref{thm:security:AUX} in several parts. To begin, we will show that an adversary that interacts with $\AUX.\KeyGen$ can't do much better than an adversary that interacts instead with an altered version of $\AUX.\KeyGen$, $\KeyGen'$, in which every classical encryption has been replaced with $\HE.\Enc_{pk}(0)$ (Lemma \ref{lem:aux-security-1}). Then we will be able to complete the proof by showing that an adversary interacting with $\KeyGen'$ instead of $\AUX.\KeyGen$ can't win the q-IND-CPA experiment for $\AUX$ with probability better than $\frac{1}{2}$.

\begin{lemma}\label{lem:aux-security-1}
Define a QHE scheme $\AUX'$ such that $\AUX'.\KeyGen(1^{\kappa},1^n)=\KeyGen'(1^\kappa,1^n)$, where $\KeyGen'$ behaves identically to $\AUX.\KeyGen$, except it replaces every classical encryption $\HE.\Enc_{pk}(x)$ with $\HE.\Enc_{pk}(0)$. Let $\AUX'.\Enc=\AUX.\Enc$, $\AUX'.\Dec=\AUX.\Dec$ and $\AUX'.\Eval=\AUX.\Eval$. Then for any quantum polynomial-time adversary $\advA=(\advA_1,\advA_2)$ with encryption oracle access, there exists a negligible function $\eta$ such that:
$$\Pr[\mathsf{SymK_{\advA,\AUX}^{cpa}}(\kappa)=1]-\Pr[\mathsf{SymK_{\advA,\AUX'}^{cpa}}(\kappa)=1]\leq \eta(\kappa).$$
Thus, we can restrict our attention to adversaries that make no use of the classical encryptions, since they add at most a negligible advantage.
\end{lemma}
\begin{proof}
We will define an adversary $\advA'=(\advA_1',\advA_2')$ for the quantum CPA-mult indistinguishability experiment for $\HE$, $\mathsf{PubK_{\advA',\HE}^{cpa\text{-}mult}}(\kappa)$.
Essentially, $\advA'$ will run $\AUX.\KeyGen$, except it will use the challenger $\Xi_{\HE}^{\textsf{cpa-mult}}$ in place of $\HE.\Enc_{pk}$, so that it is either running $\AUX.\KeyGen$ or $\KeyGen'$. It will then simulate the $\sf SymK$ experiment, and if $\advA$ wins, it will guess that it ran the original version of $\AUX.\KeyGen$, and otherwise it will guess that it ran $\KeyGen'$.

\begin{description}
\item[$\advA_1'(pk,evk)$:] $\advA_1'$ chooses uniform random bit strings $a,b\in\{0,1\}^n$ and $k\in\{0,1\}^N$, where $N=n|T_1|+\dots+n|T_L|$, and gives $m_0=\mathbf{0}=0^{2n+N}$ and $m_1=(a,b,k)$ to the challenger $\Xi_{\HE}^{\textsf{cpa-mult}}$, which outputs either $c_1=\HE.\Enc_{pk}(a,b,k)$, or $c_0=\HE.\Enc_{pk}(\mathbf{0})$.
\item[$\advA_2'(c)$:] $\advA_2'$ computes $\sigma_{aux}^{a,b,k}$ and gives $\sigma_{aux}^{a,b,k},pk,evk,c$ to $\advA_1$. $\advA_1$ may make several oracle calls, which $\advA_2'$ can simulate, because it has $a,b$ and so can run $\AUX.\Enc$. When $\advA_1$ outputs a message to the challenger, $\advA_2'$ samples a random bit $r$, and runs $\Xi_{\AUX}^{\mathsf{cpa},r}$, which it can simulate, since it has $a,b$, and so can run $\AUX.\Enc$. $\advA_2'$ then gives the challenge to $\advA_2$, and if $\advA_2$ outputs $r$, $\advA_2'$ outputs $1$, and otherwise, $\advA_2'$ outputs $0$.
\end{description}
We now calculate the probability that $\advA'$ correctly guesses which of $c_0$ and $c_1$ it received from the challenger, which we know must be less than $\frac{1}{2}+\eta(\kappa+n)$ for some negligible function, since $\HE$ is q-IND-CPA, $\kappa+n$ is the security parameter given to $\HE.\KeyGen$, and $|m_0|=|m_1|=2n+N = O(\mathrm{poly}(n))=O(\mathrm{poly}(n+\kappa))$. If $\advA'$ received $c_0$, then it acted as $\KeyGen'$, whereas if it received $c_1$, it acted as $\AUX.\KeyGen$. In the former case, the probability that $\advA'$ correctly guesses 0 is the probability that $\advA$ loses the $\sf SymK$ experiment when it interacts with $\AUX'$, $\Pr[\mathsf{SymK_{\advA,\AUX'}^{cpa}}(\kappa)=0]$. In the latter case, the probability that $\advA'$ correctly guesses 1 is the probability that $\advA$ wins the $\sf SymK$ experiment when it interacts with $\AUX$, $\Pr[\mathsf{SymK_{\advA,\AUX}^{cpa}}(\kappa)=1]$. Thus, since $\HE$ is q-IND-CPA, there exists a negligible function $\eta'$ such that
\begin{align*}
\frac{1}{2}\Pr[\mathsf{SymK_{\advA,\AUX'}^{cpa}}(\kappa)=0]+\frac{1}{2}\Pr[\mathsf{SymK_{\advA,\AUX}^{cpa}}(\kappa)=1] & \leq \frac{1}{2}+\eta'(\kappa)\\
1-\Pr[\mathsf{SymK_{\advA,\AUX'}^{cpa}}(\kappa)=1]+\Pr[\mathsf{SymK_{\advA,\AUX}^{cpa}}(\kappa)=1] &\leq 1+2\eta'(\kappa)\\
\Pr[\mathsf{SymK_{\advA,\AUX}^{cpa}}(\kappa)=1]-\Pr[\mathsf{SymK_{\advA,\AUX'}^{cpa}}(\kappa)=1] &\leq 2\eta'(\kappa).
\end{align*}
Setting $\eta=2\eta'$ completes the proof.
\end{proof}

The next lemma shows that the output of $\KeyGen'$ is actually $(pk,evk,\maxmix)$, which is independent of $a,b,k$.
\begin{lemma}\label{lem:aux-security-2}
Let $N=n|T_1|+\dots+n|T_L|$. For any $a,b\in\{0,1\}^{n}$,
$\displaystyle\sum_{k\in\{0,1\}^{N}}\!\!\sigma_{aux}^{a,b,k}=\frac{1}{2^{N}}\mathbb{I}_{2^{N}}$.
\end{lemma}
\begin{proof}
We first note that for any string $s$ of length $|s|$:
\begin{align*}
\sum_{k\in \{0,1\}^{|s|}}\sigma(s,k) &= \sum_{k\in \{0,1\}^{|s|}}\bigotimes_{i=1}^{|s|}\zgate^{k_i}\pgate^{s_i}\ket{+}\bra{+}\pgate^{s_i}\zgate^{k_i}
= \bigotimes_{i=1}^{|s|}\sum_{k\in\{0,1\}}\zgate^k\pgate^{s_i}\ket{+}\bra{+}\pgate^{s_i}\zgate^k\\
&= \bigotimes_{i=1}^{|s|}\sum_{k\in\{0,1\}}\(\ket{0}\bra{0}+\ket{1}\bra{1}+i^{2k+a}\ket{0}\bra{1}+i^{2k-a}\ket{1}\bra{0}\)\\
&= \bigotimes_{i=1}^{|s|}\(2\mathbb{I}_2+(i^{2+a}+i^a)\ket{0}\bra{1}+(i^{2-a}+i^{-a})\ket{1}\bra{0}\)
= \(2\mathbb{I}_2\)^{\otimes |s|}=2^{|s|}\mathbb{I}_{2^{|s|}}.
\end{align*}
Then it is easy to see that for any $a,b\in \{0,1\}^n$:
\begin{align*}
 \sum_{k\in\{0,1\}^{N}}\sigma_{aux}^{a,b,k}
 &= \sum_{\substack{k^{(1)}\in\{0,1\}^{n|T_1|},\dots,\\ k^{(L)}\in\{0,1\}^{n|T_L|}}}\sigma(s^{(1)}(a,b)^{*n},k^{(1)})\otimes \dots\otimes \sigma(s^{(L)}(a,b,k^{(1)},\dots,k^{(L-1)})^{*n},k^{(L)})\\
&= \sum_{k^{(1)}\in\{0,1\}^{n|T_1|}}\sigma(s^{(1)}(a,b)^{*n},k^{(1)})\otimes \sum_{k^{(2)}\in\{0,1\}^{n|T_2|}}\sigma(s^{(2)}(a,b,k^{(1)})^{*n},k^{(2)})\otimes\dots \\
&\qquad\qquad\otimes \sum_{k^{(L)}\in\{0,1\}^{n|T_L|}}\sigma(s^{(L)}(a,b,k^{(1)},\dots,k^{(L-1)})^{*n},k^{(L)})\\
&= \sum_{k^{(1)}\in\{0,1\}^{n|T_1|}}\sigma(s^{(1)}(a,b)^{*n},k^{(1)})\otimes \sum_{k^{(2)}\in\{0,1\}^{n|T_2|}}\sigma(s^{(2)}(a,b,k^{(1)})^{*n},k^{(2)})\otimes\dots\\
&\qquad\qquad\otimes \sum_{k^{(L-1)}\in\{0,1\}^{n|T_{L-1}|}}\sigma(s^{(L-1)}(a,b,k^{(1)},\dots,k^{(L-2)})^{*n},k^{(L-1)})\otimes 2^{n|T_L|}\mathbb{I}_{2^{n|T_L|}}\\
&= 2^{n|T_1|+\dots+n|T_L|}\mathbb{I}_{2^{n|T_1|+\dots+n|T_L|}}. \iftoggle{crypto}{}{\qedhere}
\end{align*}
\end{proof}

To complete the proof of Thm.~\ref{thm:security:AUX}, we simply show that no adversary interacting with $\AUX.\KeyGen'$ can win the experiment $\sf SymK^{cpa}$ with probability better than $\frac{1}{2}$.

\begin{lemma}\label{lem:aux-security-3}
For any adversary $\advA$ with access to an encryption oracle,
$\Pr[\mathsf{SymK_{\advA,\AUX'}^{cpa}}(\kappa)=1]= \frac{1}{2}$.
\end{lemma}
{\begin{proof}
Let $q$ be the number of oracle calls made by $\advA_1$, and write $\advA_1=(\advA_1^{(1)},\dots,\advA_1^{(q+1)})$. Let $q'$ be the number of oracle calls made by $\advA_2$, and write $\advA_2=(\advA_2^{(1)},\dots,\advA_2^{(q'+1)})$. If $q\geq n$, then the challenger just outputs $\bot$, independent of $r$, so certainly in that case $\advA$ cannot win with probability more than $\frac{1}{2}$, so suppose $q<n$. If $q+q'+1>n$, then the last $q+q'+1-n$ oracle calls made by $\advA_2$ simply return $\bot$, which $\advA$ could simulate without actually making these oracle calls, so suppose without loss of generality that $q+q'+1\leq n$.

The output of $\AUX'.\KeyGen=\KeyGen'$ to $\advA_1^{(1)}$ is $(pk,evk,\sigma_{aux}^{a,b,k})$, and
by Lemma \ref{lem:aux-security-2}, for any $a,b$, this is equal to $(pk,evk,\frac{1}{2^{N}}\mathbb{I}_{2^{N}})$. Thus, the interaction of $\KeyGen'$ with the experiment is shown in part (a) of Fig.~\ref{fig:aux-security}. $\KeyGen'$ chooses random bits $a_1,b_1,\dots,a_{q+q'+1},b_{q+q'+1}$, for use in oracle calls and the challenge itself, but these are independent of the information given to $\advA$ by $\KeyGen'$. (The other random bits selected by $\KeyGen'$, $a_{q+q'+2},b_{q+q'+2},\dots,a_n,b_n$ and the string $k$, are independent of the interaction with the adversary, so we ignore them.)

It is then easy to see from Fig.~\ref{fig:aux-security} that every call to the encryption oracle can be replaced by a channel that discards the input and returns a completely mixed state, since for any input $\rho^{\mathcal{M}}$, the encryption oracle returns
$$\Tr_1\(\frac{1}{4}\sum_{a,b\in\{0,1\}}\ket{a,b}\bra{a,b}_1\otimes\xgate^{a}\zgate^{b}\rho^{\cal M}\zgate^{b}\xgate^{a}\)=\frac{1}{4}\sum_{a,b\in\{0,1\}}\xgate^{a}\zgate^{b}\rho^{\cal M}\zgate^{b}\xgate^{a}=\frac{1}{2}\mathbb{I}_{2}.$$
In other words, we have:
\begin{figure}[H]
\centering
\begin{tikzpicture}
\node at (0,0){
\begin{tikzpicture}
\draw (0,.79)--(.75,.79);
\draw (0,.71)--(.75,.71);
\draw (0,0)--(2.75,0);

\filldraw[fill=white] (-.5,.5) rectangle (0,1);
\node at (-.25,.75) {$\maxmix$};

\node at (.4,.95) {$a,b$};
\node at (.5,.2) {\small $\cal M$};

\filldraw[fill=white] (.75,-.25) rectangle (2,1);
\node at (1.375,0) {\small $\QEnc_{a,b}$};

\node at (2.25,.2) {\small $\cal C$};
\end{tikzpicture}};
\node at (2.5,0) {$\equiv$};
\node at (5,0) {
\begin{tikzpicture}
\draw (0,0) -- (1,0); \draw (2,0) -- (2.75,0);

\node at (.4,.2) {\small $\cal M$};

\node at (1,0) {\meas};

\filldraw[fill=white] (1.7,-.25) rectangle (2.2,.25);
\node at (1.95,0) {$\maxmix$};

\node at (2.45,.2) {\small $\cal C$};
\end{tikzpicture}
};
\end{tikzpicture}
\end{figure}
\noindent Here $\maxmix$ denotes the channel that outputs a completely  mixed state, or equivalently, a uniform random variable.

For the same reason, the call to the challenger $\Xi_{\AUX}^{{\sf cpa},r}$ can also be replaced with the channel that discards the input and returns $\maxmix$, since $\Xi_{\AUX}^{{\sf cpa},r}$ applies a quantum one-time pad using random keys $a_{q+1},b_{q+1}$ to the input or to $\ket{0}\bra{0}$, and in either case, the resulting state is the completely mixed state. Thus, from the perspective of $\advA$, the experiment is independent of $r$, as shown in part (b) of Fig.~\ref{fig:aux-security}. Thus, an adversary cannot win with probability better than $\frac{1}{2}$.
\end{proof}
}
\begin{figure}[b]
\centering
\begin{tikzpicture}[scale=.9]
\node at (0,0) {
\begin{tikzpicture}
\draw (1,4.25)--(13,4.25);
\draw (1,3.25)--(13,3.25);
\draw (1,2.5)--(10.5,2.5);
\draw (1,1.5)--(3.5,1.5);	
\draw (0,.75)--(13,.75);
\draw (1,0) -- (13,0);

\filldraw[fill=white] (-.5,-.25) rectangle (0,2.75);
\node[rotate=90] at (-.25,1.25) {$\HE.\KeyGen(1^{\kappa+n})$};
\node at (.75,.95) {$pk,evk$};

\filldraw[fill=white] (.8,4) rectangle (1.3,4.5);
\node at (1.05,4.25) {$\maxmix$};
\node at (1.05,3.85) {$\vdots$};
\filldraw[fill=white] (.8,3) rectangle (1.3,3.5);
\node at (1.05,3.25) {$\maxmix$};
\filldraw[fill=white] (.8,2.25) rectangle (1.3,2.75);
\node at (1.05,2.5) {$\maxmix$};
\node at (1.05,2.1) {$\vdots$};
\filldraw[fill=white] (.8,1.25) rectangle (1.3,1.75);
\node at (1.05,1.5) {$\maxmix$};

\filldraw[fill=white] (.8,-.25) rectangle (1.3,.25);
\node at (1.05,0) {$\maxmix$};

\filldraw[fill=white] (1.8,-.25) rectangle (2.7,1);
\node at (2.25,.375) {$\advA_1^{(1)}$};

\node at (2.8,4.45) {\small $a_{q+q'+1},b_{q+q'+1}$};
\node at (2.35,3.45) {\small $a_{q+2},b_{q+2}$};
\node at (2.35,2.7) {\small $a_{q+1},b_{q+1}$};
\node at (2,1.7) {\small $a_1,b_1$};
\node at (2.95,.95) {\small $\cal M$};
\node at (2.95,.2) {\small $\cal E$};

\filldraw[fill=white] (3.25,.5) rectangle (4.75,1.75);
\node at (4,.75) {\small$\QEnc_{a_1,b_1}$};

\node at (6,2) {\small (encryption oracle)};
\draw[dashed,->] (6,1.8) -- (4.75,1.25);

\node at (4.95,.95) {\small $\cal C$};

\filldraw[fill=white] (5.3,-.25) rectangle (6.2,1);
\node at (5.75,.375) {$\advA_1^{(2)}$};

\node at (6.45,.95) {\small $\cal M$};
\node at (6.45,.2) {\small $\cal E$};

\fill[fill=white] (6.9,-.25) rectangle (7.85,1.75);
\node at (7.375,.75) {$\dots$};

\node at (8.325,.95) {\small $\cal C$};
\node at (8.325,.2) {\small $\cal E$};

\filldraw[fill=white] (8.55,-.25) rectangle (9.7,1);
\node at (9.125,.375) {$\advA_1^{(q+1)}$};

\node at (9.95,.95) {\small $\cal M$};
\node at (9.95,.2) {\small $\cal E$};

\filldraw[fill=white] (10.4,.5) rectangle (11.6,2.75);
\node at (11,1.625) {$\Xi_{\AUX}^{{\sf cpa},r}$};

\node at (11.85,.95) {\small $\cal C$};

\filldraw[fill=white] (12.3,1.25) rectangle (13.2,4.5);
\node[rotate=90] at (12.75,2.875) {$\AUX.\Enc$};
\draw[->] (12.4,1) -- (12.4,1.25);
\draw[->] (13.1,1.25) -- (13.1,1);
\filldraw[fill=white] (12.3,-.25) rectangle (13.2,1);
\node at (12.75,.375) {$\advA_2$};

\draw (13.2,.415)--(13.7,.415);
\draw (13.2,.335)--(13.7,.335);
\node at (13.95,.375) {$r'$};

\draw[dashed] (-.75, -.5) rectangle (1.5,5.25);
\node at (.375,5) {$\KeyGen'$};

\end{tikzpicture}
};

\node at (0,-3.25) {(a)};

\node at (0,-5){
\begin{tikzpicture}
\draw (0,.75)--(13,.75);
\draw (1,0) -- (13,0);

\filldraw[fill=white] (-.5,-.25) rectangle (0,2.5);
\node[rotate=90] at (-.25,1.125) {\small$\HE.\KeyGen(1^{\kappa+n})$};
\node at (.75,.95) {$pk,evk$};

\filldraw[fill=white] (.8,-.25) rectangle (1.3,.25);
\node at (1.05,0) {$\maxmix$};

\filldraw[fill=white] (1.8,-.25) rectangle (2.7,1);
\node at (2.25,.375) {$\advA_1^{(1)}$};

\node at (2.95,.95) {\small $\cal M$};
\node at (2.95,.2) {\small $\cal E$};

\fill[fill=white] (3.4,.5) rectangle (4.6,1.75);
\node at (3.6,.75) {\meas};
\filldraw[fill=white] (4.25,.5) rectangle (4.75,1);
\node at (4.5,.75) {$\maxmix$};

\node at (5,.95) {\small $\cal C$};

\filldraw[fill=white] (5.3,-.25) rectangle (6.2,1);
\node at (5.75,.375) {$\advA_1^{(2)}$};

\node at (6.45,.95) {\small $\cal M$};
\node at (6.45,.2) {\small $\cal E$};

\fill[fill=white] (6.9,-.25) rectangle (7.85,1.75);
\node at (7.375,.75) {$\dots$};

\node at (8.325,.95) {\small $\cal C$};
\node at (8.325,.2) {\small $\cal E$};

\filldraw[fill=white] (8.55,-.25) rectangle (9.7,1);
\node at (9.125,.375) {$\advA_1^{(q+1)}$};

\node at (9.95,.95) {\small $\cal M$};
\node at (9.95,.2) {\small $\cal E$};

\fill[fill=white] (10.4,.5) rectangle (11.6,1.75);
\node at (10.6,.75) {\meas};
\filldraw[fill=white] (11.25,.5) rectangle (11.75,1);
\node at (11.5,.75) {$\maxmix$};

\node at (12,.95) {\small $\cal C$};

\node at (12.3,1.5) {\meas};
\filldraw[fill=white] (12.85,1.25)rectangle(13.35,1.75);
\node at (13.1,1.5) {$\maxmix$};
\draw[->] (12.4,1)--(12.4,1.25);
\draw[->] (13.1,1.25)--(13.1,1);
\filldraw[fill=white] (12.3,-.25) rectangle (13.2,1);
\node at (12.75,.375) {$\advA_2$};

\draw (13.2,.415)--(13.7,.415);
\draw (13.2,.335)--(13.7,.335);
\node at (13.95,.375) {$r'$};

\end{tikzpicture}
};
\node at (0,-6.75) {(b)};
\end{tikzpicture}

\caption{Proof of Lemma \ref{lem:aux-security-3}. (a) shows how $\KeyGen'$ interacts with the experiment. The channel $\maxmix$ outputs a completely  mixed state, or equivalently, a uniform random variable. Since the random bits $a_i,b_i$ are independent of the other outputs of $\KeyGen'$, for each $i$, we can replace each of the oracle calls as well as the challenger with a channel that discards the input and returns $\maxmix$, as shown in (b). Thus, the experiment is independent of $r$ from the perspective of $\advA$, and so $\advA$ can do no better than guessing $r$. }\label{fig:aux-security}
\end{figure}

Combining Lemma \ref{lem:aux-security-1} and Lemma \ref{lem:aux-security-3} proves Thm.~\ref{thm:security:AUX} immediately.

}{
\section{Proof of Lemma \ref{lem:aux-security-2}}\label{app:misc}

\noindent\textbf{Lemma \ref{lem:aux-security-2}}\emph{
Let $N=n|T_1|+\dots+n|T_L|$. For any $a,b\in\{0,1\}^{n}$,
$\sum_{k\in\{0,1\}^{N}}\sigma_{aux}^{a,b,k}=\frac{1}{2^{N}}\mathbb{I}_{2^{N}}.$
}
\begin{proof}
We first note that for any string $s$ of length $|s|$:
\begin{align*}
\sum_{k\in \{0,1\}^{|s|}}\sigma(s,k) &= \sum_{k\in \{0,1\}^{|s|}}\bigotimes_{i=1}^{|s|}\zgate^{k_i}\pgate^{s_i}\ket{+}\bra{+}\pgate^{s_i}\zgate^{k_i}
= \bigotimes_{i=1}^{|s|}\sum_{k\in\{0,1\}}\zgate^k\pgate^{s_i}\ket{+}\bra{+}\pgate^{s_i}\zgate^k\\
&= \bigotimes_{i=1}^{|s|}\sum_{k\in\{0,1\}}\(\ket{0}\bra{0}+\ket{1}\bra{1}+i^{2k+a}\ket{0}\bra{1}+i^{2k-a}\ket{1}\bra{0}\)\\
&= \bigotimes_{i=1}^{|s|}\(2\mathbb{I}_2+(i^{2+a}+i^a)\ket{0}\bra{1}+(i^{2-a}+i^{-a})\ket{1}\bra{0}\)\\
&= \(2\mathbb{I}_2\)^{\otimes |s|}=2^{|s|}\mathbb{I}_{2^{|s|}}.
\end{align*}
Then it is easy to see that for any $a,b\in \{0,1\}^n$:
\begin{align*}
 \sum_{k\in\{0,1\}^{N}}\sigma_{aux}^{a,b,k}
 &= \sum_{\substack{k^{(1)}\in\{0,1\}^{n|T_1|},\dots,\\ k^{(L)}\in\{0,1\}^{n|T_L|}}}\sigma(s^{(1)}(a,b)^{*n},k^{(1)})\otimes \dots\otimes \sigma(s^{(L)}(a,b,k^{(1)},\dots,k^{(L-1)})^{*n},k^{(L)})\\
&= \sum_{k^{(1)}\in\{0,1\}^{n|T_1|}}\sigma(s^{(1)}(a,b)^{*n},k^{(1)})\otimes \sum_{k^{(2)}\in\{0,1\}^{n|T_2|}}\sigma(s^{(2)}(a,b,k^{(1)})^{*n},k^{(2)})\otimes\dots \\
&\qquad\qquad\otimes \sum_{k^{(L)}\in\{0,1\}^{n|T_L|}}\sigma(s^{(L)}(a,b,k^{(1)},\dots,k^{(L-1)})^{*n},k^{(L)})\\
&= \sum_{k^{(1)}\in\{0,1\}^{n|T_1|}}\sigma(s^{(1)}(a,b)^{*n},k^{(1)})\otimes \sum_{k^{(2)}\in\{0,1\}^{n|T_2|}}\sigma(s^{(2)}(a,b,k^{(1)})^{*n},k^{(2)})\otimes\dots\\
&\qquad\qquad\otimes \sum_{k^{(L-1)}\in\{0,1\}^{n|T_{L-1}|}}\sigma(s^{(L-1)}(a,b,k^{(1)},\dots,k^{(L-2)})^{*n},k^{(L-1)})\otimes 2^{n|T_L|}\mathbb{I}_{2^{n|T_L|}}\\
&= 2^{n|T_1|+\dots+n|T_L|}\mathbb{I}_{2^{n|T_1|+\dots+n|T_L|}}.
\end{align*}
\vskip-20pt
\end{proof}
}

\end{document}

%% file: fig-adv-with-oracle.tex
\begin{tikzpicture}
\node at (0,0){
\begin{tikzpicture}
\draw (-.5,2.04) -- (1.5,2.04);
\draw (-.5,1.96) -- (1.5,1.96);
\draw (-.5,.125)--(1.5,.125);

\filldraw[fill=white] (.25,2.25) rectangle (.75,.75);
\node[rotate=90] at (.5,1.5) {\small$\QHE.\Enc$};

\draw[->] (.35,.5)--(.35,.75);
\draw[->] (.65,.75)--(.65,.5);

\filldraw[fill=white] (.25,-.25) rectangle (.75,.5);
\node at (.5,.125) {$\advA$};

\node at (-.25,2.2) {$sk$};
\node at (1,2.2) {$sk$};


\node at (-.25,.325) {\small $\cal X$};
\node at (1,.325) {\small $\cal Y$};
\end{tikzpicture}
};

\node at (1.7,0) {$:=$};

\node at (7.75,0) {
\begin{tikzpicture}[xscale=1.1]
\draw (0,1.79)--(9.75,1.79);
\draw (0,1.71)--(9.75,1.71);
	\draw (1,.75)--(9,.75);
\draw (0,.375) -- (1,.375);	\draw (9.25,.375)--(9.75,.375);
	\draw (1,0)--(9,0);

\node at (.25,1.95) {$sk$};
\node at (.25,.575) {\small $\cal X$};

\filldraw[fill=white] (.5,-.25) rectangle (1.25,1);
\node at (.875,.375) {$\advA^{(1)}$};

\node at (1.5,.95) {\small $\cal M$};
\node at (1.5,.2) {\small $\cal E$};

\filldraw[fill=white] (2,.5) rectangle (2.5,2);
\node[rotate=90] at (2.25,1.25) {\small$\QHE.\Enc$};

\node at (2.75,1.95) {$sk$};
\node at (2.75,.95) {\small $\cal C$};

\filldraw[fill=white] (3.25,-.25) rectangle (4,1);
\node at (3.625,.375) {$\advA^{(2)}$};

\node at (4.25,.95) {\small $\cal M$};
\node at (4.25,.2) {\small $\cal E$};

\filldraw[fill=white] (4.75,.5) rectangle (5.25,2);
\node[rotate=90] at (5,1.25) {\small$\QHE.\Enc$};

\node at (5.5,1.95) {$sk$};
\node at (5.5,.95) {\small $\cal C$};

\fill[fill=white] (5.75,-.25) rectangle (6.25,2);
\node at (6,1.75) {$\dots$};
\node at (6,.75) {$\dots$};
\node at (6,0) {$\dots$};

\node at (6.5,1.95) {$sk$};
\node at (6.5,.95) {\small $\cal M$};
\node at (6.5,.2) {\small $\cal E$};

\filldraw[fill=white] (6.75,.5) rectangle (7.25,2);
\node[rotate=90] at (7,1.25) {\small$\QHE.\Enc$};

\node at (7.5,1.95) {$sk$};
\node at (7.5,.95) {\small $\cal C$};

\filldraw[fill=white] (8,-.25) rectangle (9.25,1);
\node at (8.625,.375) {$\advA^{(q+1)}$};

\node at (9.5,.575) {\small $\cal Y$};

\end{tikzpicture}
};

\end{tikzpicture}
\caption{An adversary $\advA$ that makes at most $q$ encryption oracle calls is a list of quantum channels $\advA^{(1)},\dots,\advA^{(q+1)}$ such that for $j=1,\dots,q$, $\advA^{(j)}$ sends a message to an encryption oracle, and $\advA^{(j+1)}$ receives the output. The full interaction is shown on the right, but we use the figure on the left as a short-hand for this interaction.}\label{fig:adv-with-oracle}

%% file: fig-adv-with-oracle-bounded.tex
\begin{tikzpicture}
\node at (0,0){
\begin{tikzpicture}
\draw (-.5,2.04) -- (1.5,2.04);
\draw (-.5,1.96) -- (1.5,1.96);
\draw (-.5,1.04)--(1.5,1.04);
\draw (-.5,.96)--(1.5,.96);
\draw (-.5,.125)--(1.5,.125);

\filldraw[fill=white] (.25,2.25) rectangle (.75,.75);
\node[rotate=90] at (.5,1.5) {\small$\QHE.\Enc$};

\draw[->] (.35,.5)--(.35,.75);
\draw[->] (.65,.75)--(.65,.5);

\filldraw[fill=white] (.25,-.25) rectangle (.75,.5);
\node at (.5,.125) {$\advA$};

\node at (-.25,2.2) {$sk$};
\node at (1,2.2) {$sk$};

\node at (-.25,1.2) {$d$};
\node at (1.25,1.2) {$d+q$};


\node at (-.25,.325) {\small $\cal X$};
\node at (1,.325) {\small $\cal Y$};
\end{tikzpicture}
};

\node at (1.7,0) {$:=$};

\node at (7.75,0) {
\begin{tikzpicture}[xscale=1.1]
\draw (0,1.79)--(9.75,1.79);
\draw (0,1.71)--(9.75,1.71);
\draw (0,1.29)--(9.75,1.29);
\draw (0,1.21)--(9.75,1.21);
	\draw (1,.75)--(9,.75);
\draw (0,.375) -- (1,.375);	\draw (9.25,.375)--(9.75,.375);
	\draw (1,0)--(9,0);

\node at (.25,1.95) {$sk$};
\node at (.25,1.45) {$d$};
\node at (.25,.575) {\small $\cal X$};

\filldraw[fill=white] (.5,-.25) rectangle (1.25,1);
\node at (.875,.375) {$\advA^{(1)}$};

\node at (1.5,.95) {\small $\cal M$};
\node at (1.5,.2) {\small $\cal E$};

\filldraw[fill=white] (2,.5) rectangle (2.5,2);
\node[rotate=90] at (2.25,1.25) {\small$\QHE.\Enc$};

\node at (2.75,1.95) {$sk$};
\node at (3,1.45) {$d+1$};
\node at (2.75,.95) {\small $\cal C$};

\filldraw[fill=white] (3.25,-.25) rectangle (4,1);
\node at (3.625,.375) {$\advA^{(2)}$};

\node at (4.25,.95) {\small $\cal M$};
\node at (4.25,.2) {\small $\cal E$};

\filldraw[fill=white] (4.75,.5) rectangle (5.25,2);
\node[rotate=90] at (5,1.25) {\small$\QHE.\Enc$};

\node at (5.5,1.95) {$sk$};

\node at (5.5,.95) {\small $\cal C$};

\fill[fill=white] (5.75,-.25) rectangle (6.25,2);
\node at (6,1.75) {$\dots$};
\node at (6,1.25) {$\dots$};
\node at (6,.75) {$\dots$};
\node at (6,0) {$\dots$};

\node at (5.75,1.45) {$d+2$};

\node at (6.5,1.95) {$sk$};
\node at (6.5,.95) {\small $\cal M$};
\node at (6.5,.2) {\small $\cal E$};

\filldraw[fill=white] (6.75,.5) rectangle (7.25,2);
\node[rotate=90] at (7,1.25) {\small$\QHE.\Enc$};

\node at (7.5,1.95) {$sk$};
\node at (7.75,1.45) {$d+q$};
\node at (7.5,.95) {\small $\cal C$};

\filldraw[fill=white] (8,-.25) rectangle (9.25,1);
\node at (8.625,.375) {$\advA^{(q+1)}$};

\node at (9.5,.575) {\small $\cal Y$};

\end{tikzpicture}
};

\end{tikzpicture}
\caption{An adversary $\advA$ that makes at most $q$ encryption oracle calls to a bounded encryption oracle is a list of quantum channels $\advA^{(1)},\dots,\advA^{(q+1)}$ with the interaction shown on the right. We use the figure on the left as a short-hand for this interaction.}\label{fig:adv-with-oracle-bounded}

%% file: fig-cl-security.tex
\begin{tikzpicture}
\draw (0,2.54)--(5,2.54);
\draw (0,2.46)--(5,2.46);
		\draw (4.5,2.04)--(5,2.04);
		\draw (4.5,1.96)--(5,1.96);
\draw (2,1.54)--(6.25,1.54);	
\draw (2,1.46)--(6.25,1.46);
\draw (0,.79)--(1,.79);	\draw (1,.75)--(6.5,.75);	\draw(6.75,.79)--(7.25,.79);
\draw (0,.71)--(1,.71);					\draw(6.75,.71)--(7.25,.71);
\draw (0,0.04)--(1,0.04);	\draw (1,0)--(6.5,0);
\draw (0,-.04)--(1,-.04);

\filldraw[fill=white] (-.5,-.25) rectangle (0,2.75);
\node[rotate=90] at (-.25,1.25) {$\KeyGen$};

\node at (.25,2.7) {$pk$};
\node at (.25,.95) {$pk$};
\node at (.35,.2) {$evk$};

\filldraw[fill=white] (.75,-.25) rectangle (1.25,1);
\node at (1,.375) {$\advA_1$};

\node at (1.5,.95) {\small $\cal M$};
\node at (1.5,.2) {\small $\cal E$};

\filldraw[fill=white] (2,1.75) rectangle (2.5,1.25);
\node at (2.25,1.5) {$\maxmix$};
\filldraw[fill=white] (2,1) rectangle (2.5,.5);
\node at (2.25,.75) {$\Psi^{\bar{r}}$};

\node at (2.8,1.7) {$a,b$};
\node at (2.75,.95) {\small $\cal M$};

\filldraw[fill=white] (3.25,1.75) rectangle (4.25,.5);
\node at (3.75,.75) {$\xgate^a\zgate^b$};

\node at (4.15,2) {$0,0$};
\node at (4.55,1.7) {$a,b$};
\node at (4.5,.95) {\small $\cal X$};

\filldraw[fill=white] (5,2.75) rectangle (5.5,1.25);
\node[rotate=90] at (5.25,2) {\small$\Xi_{\HE}^{\textsf{cpa-mult},s}$};

\node at (5.75,1.7) {$c$};

\filldraw[fill=white] (6.25,-.25) rectangle (6.75,1.75);
\node at (6.5,.75) {$\advA_2$};

\node at (7.5,.75) {$r'$};

\draw (7.75,.79)--(8.25,.79);
\draw (7.75,.71)--(8.25,.71);
\node at (9.5,.75) {$s'=\bar{r}\oplus r'$};

\draw[dashed] (6.15,-.35) rectangle (7.75,1.85);
\draw[dashed] (.65,-.35) rectangle (4.9,2.25);
\node at (2.775,-.6) {$\advA_1'$};
\node at (6.95,-.6) {$\advA_2'$};

\draw[dashed] (1.9,.4) rectangle (5.6,2.85);
\node at (3.75,3.1) {$\Xi_{\CL}^{{\sf cpa},r}$/$\Xi_{\CL'}^{{\sf cpa},r}$};
\end{tikzpicture}
\caption{The new adversary $\advA'$ for $\mathsf{PubK_{\advA',\HE}^{\textsf{cpa-mult}}}$, where $\Psi$ is the channel that replaces the system with $\ket{0}\bra{0}$. Here $\maxmix$ denotes the channel that outputs a completely mixed state, or equivalently, a uniform random variable. If $s=1$, the middle dashed box is $\Xi_{\CL}^{{\sf cpa},r}$, and if $s=0$, $\Xi_{\CL'}^{{\sf cpa},r}$.}\label{fig:cl-security}

%% file: fig-cl-security-s0.tex
\begin{tikzpicture}
\draw (0,2.54)--(5,2.54);
\draw (0,2.46)--(5,2.46);
		\draw (4.5,2.04)--(6.5,2.04);
		\draw (4.5,1.96)--(6.5,1.96);
\draw (2,1.54)--(4,1.54);	
\draw (2,1.46)--(4,1.46);
\draw (0,.79)--(1,.79);	\draw (1,.75)--(6.5,.75);	\draw(6.75,.79)--(7.25,.79);
\draw (0,.71)--(1,.71);					\draw(6.75,.71)--(7.25,.71);
\draw (0,0.04)--(1,0.04);	\draw (1,0)--(6.5,0);
\draw (0,-.04)--(1,-.04);

\filldraw[fill=white] (-.5,-.25) rectangle (0,2.75);
\node[rotate=90] at (-.25,1.25) {$\KeyGen$};

\node at (.25,2.7) {$pk$};
\node at (.25,.95) {$pk$};
\node at (.35,.2) {$evk$};

\filldraw[fill=white] (.75,-.25) rectangle (1.25,1);
\node at (1,.375) {$\advA_1$};

\node at (1.5,.95) {\small $\cal M$};
\node at (1.5,.2) {\small $\cal E$};

\filldraw[fill=white] (2,1.75) rectangle (2.5,1.25);
\node at (2.25,1.5) {$\maxmix$};
\filldraw[fill=white] (2,1) rectangle (2.5,.5);
\node at (2.25,.75) {$\Psi^{\bar r}$};

\node at (2.8,1.7) {$a,b$};
\node at (2.75,.95) {\small $\cal M$};

\filldraw[fill=white] (3.25,1.75) rectangle (4.25,.5);
\node at (3.75,.75) {$\xgate^a\zgate^b$};

\node at (4.15,2) {$0,0$};
\node at (4.5,.95) {\small $\cal X$};

\filldraw[fill=white] (5,2.85) rectangle (5.5,1.65);
\node[rotate=90] at (5.25,2.25) {\small$\HE.\Enc$};

\node at (5.75,2.2) {$c$};

\filldraw[fill=white] (6.25,-.25) rectangle (6.75,2.25);
\node at (6.5,.75) {$\advA_2$};

\node at (7.5,.75) {$r'$};

\draw[dashed] (1.9,1.85) -- (4.75,1.85)--(4.75,.4)--(3.15,.4)--(3.15,1.15)--(1.9,1.15)--(1.9,1.85);
\end{tikzpicture}
\caption{When $s=0$, $\advA'$ is simulating $\mathsf{PubK_{\advA,\CL'}^{cpa}}$. In this case, $c$ is independent of $a,b$, and so the only dependence on $a,b$ is the quantum-one-time-pad encrypted message in~$\cal X$.}\label{fig:cl-security-s0}

%% file: fig-cl-security-last.tex
\begin{tikzpicture}
\draw (0,2.54)--(5,2.54);
\draw (0,2.46)--(5,2.46);
		\draw (4.5,2.04)--(6.5,2.04);
		\draw (4.5,1.96)--(6.5,1.96);
\draw (0,.79)--(1,.79);	\draw (1,.75)--(6.5,.75);	\draw(6.75,.79)--(7.25,.79);
\draw (0,.71)--(1,.71);					\draw(6.75,.71)--(7.25,.71);
\draw (0,0.04)--(1,0.04);	\draw (1,0)--(6.5,0);
\draw (0,-.04)--(1,-.04);

\filldraw[fill=white] (-.5,-.25) rectangle (0,2.75);
\node[rotate=90] at (-.25,1.25) {$\KeyGen$};

\node at (.25,2.7) {$pk$};
\node at (.25,.95) {$pk$};
\node at (.35,.2) {$evk$};

\filldraw[fill=white] (.75,-.25) rectangle (1.25,1);
\node at (1,.375) {$\advA_1$};

\node at (1.5,.95) {\small $\cal M$};
\node at (1.5,.2) {\small $\cal E$};

\filldraw[fill=white] (2,1) rectangle (2.5,.5);
\node at (2.25,.75) {$\Psi^{\bar r}$};

\node at (2.75,.95) {\small $\cal M$};

\fill[fill=white] (3.5,.8) rectangle (4.25,.7);

\node at (3.5,.75) {\meas};

\filldraw[fill=white] (4.25,1) rectangle (4.75,.5);
\node at (4.5,.75) {$\maxmix$};

\node at (4.15,2) {$0,0$};
\node at (5,.95) {\small $\cal X$};

\filldraw[fill=white] (5,2.85) rectangle (5.5,1.65);
\node[rotate=90] at (5.25,2.25) {\small$\HE.\Enc$};

\node at (5.75,2.2) {$c$};

\filldraw[fill=white] (6.25,-.25) rectangle (6.75,2.25);
\node at (6.5,.75) {$\advA_2$};

\node at (7.5,.75) {$r'$};

\draw[dashed] (3.05,1.1) rectangle (5.25,.4);
\end{tikzpicture}

\caption{The circuit from Figure \ref{fig:cl-security-s0} is equivalent to the above circuit, in which the system in $\mathcal{M}$ is replaced with the completely mixed state. Then from the perspective of $\advA$, the experiment is independent of $r$.}\label{fig:cl-security-last}